\documentclass[12pt]{article}
\usepackage[T1]{fontenc}
\usepackage[mathscr]{euscript}
\usepackage{amsmath}
\usepackage{amsthm}
\usepackage{stackrel, amssymb}
\usepackage{physics}
\usepackage{tensor}
\usepackage{mathtools}
\usepackage{slashed}
\usepackage{quiver}
\usepackage{hyperref}
\hypersetup{ 
	colorlinks=true,
	linkcolor=blue,
	urlcolor=blue,
	citecolor=blue
}
\usepackage{geometry}
\usepackage{sectsty}
\subsectionfont{\normalfont\itshape}
\subsubsectionfont{\normalfont\itshape}
\usepackage[toc, page]{appendix}
\usepackage{authblk}
\usepackage{xcolor}
\usepackage[style=alphabetic]{biblatex}
\numberwithin{equation}{section}

\newtheorem{prop}{Proposition}[section]

\newtheorem{definition}[prop]{Definition}

\newtheorem{theorem}[prop]{Theorem}

\newtheorem{lemma}[prop]{Lemma}

\newtheorem{conj}[prop]{Conjecture}

\title{A Scattering Transform for Noncommutative Instantons}
\author{Spencer Tamagni}
\affil{\textit{Leinweber Institute for Theoretical Physics, University of California, Berkeley}}
\date{\today}
\begin{document}

\maketitle
\setcounter{tocdepth}{2}

\begin{abstract}
We give a detailed and mathematically rigorous analysis of the path integrals of chiral fermions supported on holomorphic curves on $T^* \mathbb{C}$ in a general noncommutative instanton background. It is shown that such path integrals can be interpreted as computing instanton analogs of matrix coefficients of monopole scattering matrices. Generalizing the known relation between monopole scattering matrices and $R$-matrices of (shifted) Yangians $\mathsf{Y}(\mathfrak{gl}_r)$, our formalism gives rise to a novel geometric method to calculate $R$-matrices of (shifted) affine Yangians $\mathsf{Y}(\widehat{\mathfrak{gl}}_r)$. This may also be viewed as an explicit description of double affine Grassmannian slices by $\infty \times \infty$ matrices, compatible with the factorization morphism of Braverman, Finkelberg, and Gaitsgory. Our approach unifies a number of earlier results in the literature, and also leads to interesting new results and conjectures.
\end{abstract}

\tableofcontents 

\section{Introduction}

The scattering transform for monopoles is a beautiful construction in classical gauge theory and differential geometry, first appearing in the work of Hitchin \cite{hitchinspectral} and Hurtubise \cite{hurtubise}. Historically, it was of interest because it gave a direct method to understand the geometry of the moduli space of monopoles on $\mathbb{R}^3$ (which is defined in terms of solutions to nonlinear partial differential equations) in simple, finite-dimensional terms, see for example the textbook \cite{atiyahhitchin}. More recently, it has been substantially generalized to establish a precise dictionary between moduli spaces of monopoles and (generalized) affine Grassmannian slices studied in the context of the geometric Langlands correspondence and the mathematical approach to Coulomb branches, see \cite{KW}, \cite{BDG}, \cite{bfnslice}. 

In this paper, we will reexamine and considerably generalize the notion of a monopole scattering matrix. Our focus will be on extending the construction of a scattering matrix from moduli spaces of monopoles on $\mathbb{R}^3$ to moduli spaces of (noncommutative) instantons on $\mathbb{R}^4$ and $\mathbb{R}^3 \times S^1$ (the significance of these backgrounds is that they are the simplest $S^1$-fibrations over $\mathbb{R}^3$). This requires introducing a number of new ideas, and reveals a hidden link between scattering transforms and other structures in mathematical physics and representation theory, such as the Maulik-Okounkov Yangians \cite{mo}, double affine Grassmannian and Langlands correspondence for affine Lie algebras \cite{Braverman_2010}, and $\mathscr{W}_{1 + \infty}$ symmetry of the $B$-model of topological string theory on the backgrounds studied in \cite{Aganagic_2005}. The instanton scattering construction admits a straightforward uplift into string theory, which gives a natural explanation for all these connections, building on earlier works such as \cite{DHS}, \cite{DHSV}, \cite{witten2009geometriclanglandsdimensions}, \cite{costello2016}. 

Despite the fact that the construction is motivated by string theory, and naturally takes place within quantum field theory, we are able to show a number of results in a mathematically rigorous manner. This is because the only piece of quantum field theory we need to make essential use of is the theory of determinants of Dirac-like operators on Riemann surfaces, and it has been well-known since the 1980s that this may be formalized using infinite Grassmannians (we call them Sato Grassmannians) developed in the study of  integrable systems and the KP hierarchy, see e.g. \cite{miwa2000solitons}, \cite{Mulase2002ALGEBRAICTO}, \cite{segal-wilson}, \cite{witten} or Appendix \ref{fermionreview} for a quick review. Because we hope the results will be of interest to both physicists and representation theorists, the paper has been organized to ideally remain accessible to both audiences. This means we have written theorems and proofs when possible, and have tried to clearly distinguish the parts of the analysis which rely on physics arguments from the rest of the paper. The mathematical arguments should be straightforward and self-contained, though may seem rather ad hoc without knowing the context from physics. 

\subsection{Results}
In this section we will establish some notations and conventions, then state the results of this paper likely to be of interest to mathematicians, and finally state the results likely to be of interest to physicists. 

\subsubsection{Notations and conventions}
Before stating results we need to set up some notation that will be used throughout the paper. $\hbar$ is a $\mathbb{C}$-valued parameter, nonzero unless explicitly stated otherwise. $\mathscr{D}_\hbar$ is the $\mathbb{C}$-algebra generated by $z$ and $w$ subject to the relation $\comm{w}{z} = \hbar$. $\mathscr{D}_\hbar$ is the ring of functions on the noncommutative cotangent bundle $T^*_\hbar \mathbb{C}$. For $\hbar \neq 0$ it is isomorphic to the ring of differential operators on $\mathbb{A}^1$; for $\hbar = 0$ it specializes to the coordinate ring of $T^* \mathbb{C}$. In this paper, when we say $\mathscr{D}_\hbar$-module we usually mean right $\mathscr{D}_\hbar$-module. 

By ``holomorphic curve in $T^*_\hbar \mathbb{C}$'' we mean the following. Fix $P(z, w) \in \mathscr{D}_\hbar$, nonvanishing modulo $\hbar$, and define the right $\mathscr{D}_\hbar$-module 
\begin{equation}
    \widehat{\mathscr{O}}_C := \mathscr{D}_\hbar/P \mathscr{D}_\hbar. 
\end{equation}
This is viewed as the noncommutative analog of the structure sheaf of the affine curve $C$ defined by $P(z, w) = 0$ in the $\hbar = 0$ limit. 

An essential role in this paper is played by the moduli space of instantons (of rank $r$ and charge $n$) on the noncommutative space $T^*_\hbar \mathbb{C}$, see e.g. \cite{bgk01}, \cite{Kapustin_2001}, \cite{Nekrasov_1998}. It shown in those articles that this moduli space, in its various guises, may be presented via an ADHM-like \cite{adhm} construction as a Nakajima quiver variety
\begin{equation}
    \widetilde{M}(n, r) := \{ (B_1, B_2, I, J) \in \text{End}(V) \oplus \text{End}(V) \oplus \text{Hom}(W, V) \oplus \text{Hom}(V, W) | \comm{B_1}{B_2} + IJ = \hbar \}/GL(V).
\end{equation}
Here $V$ and $W$ are vector spaces, with $\dim V = n$ and $\dim W = r$. As long as $\hbar \neq 0$, the affine/categorical quotient, GIT quotient for either choice of stability condition, and stack quotient by $GL(V)$ all coincide (see Lemma \ref{stabandcostab}) and produce a smooth affine algebraic variety of dimenson $2nr$. In this paper it will be convenient to be agnostic about the total instanton charge and work over the union 
\begin{equation}
    \widetilde{M}(r) := \bigsqcup_{n \geq 0} \widetilde{M}(n, r)
\end{equation}
of all instanton moduli of fixed rank. Over $\widetilde{M}(r)$ we have tautological bundles $\mathscr{W}$ and $\mathscr{V}$ of (locally constant) ranks $r$ and $n$ associated to the quiver data $W$ and $V$; note $\mathscr{W}$ is trivial and we have a canonical isomorphism $\mathscr{W} \simeq W \otimes \mathscr{O}_{\widetilde{M}(r)}$. 

Also over $\widetilde{M}(r)$, we have the universal instanton $\mathscr{E}$, the middle cohomology sheaf of the ADHM complex
\begin{equation}
\begin{tikzcd}
    \mathscr{V} \otimes \mathscr{D}_\hbar \arrow[r, "\alpha"] & (\mathscr{V} \oplus \mathscr{V} \oplus \mathscr{W}) \otimes \mathscr{D}_\hbar \arrow[r, "\beta"] & \mathscr{V} \otimes \mathscr{D}_\hbar 
\end{tikzcd}
\end{equation}
with maps (when $z$ and $w$ are written we understand the operators of left multiplication by them)
\begin{equation}
\begin{split}
    \alpha & = \begin{pmatrix} B_1 - z \\ B_2 - w \\ J \end{pmatrix} \\
    \beta & = \begin{pmatrix} -B_2 + w && B_1 - z && I  \end{pmatrix}
\end{split}
\end{equation}
where the quiver data $(B_1, B_2, I_, J)$ are interpreted as the corresponding tautological morphisms of tautological bundles. $\mathscr{E}$ is a sheaf of $\mathscr{O}_{\widetilde{M}(r)}$-modules valued in right $\mathscr{D}_\hbar$-modules; if its rank is counted over the base field $\mathbb{C}$ then it is infinite, but if its rank is counted over $\mathscr{D}_\hbar$ then it has rank $r$. 

In Section \ref{sect:monopolefrominstanton} we will have occasion to consider the singular affine varieties 
\begin{equation}
    M_0(n, r) := \text{Spec} \, \mathbb{C}[\mu^{-1}(0)]^{GL(V)}
\end{equation}
where $\mu(B_1, B_2, I, J) = \comm{B_1}{B_2} + IJ$ is the moment map for the $GL(V)$ action on $T^*( \text{End}(V) \oplus \text{Hom}(W, V))$ with the canonical symplectic form. We will call these the Uhlenbeck moduli spaces of instantons. 

We will also make use of various standard constructions related to the free fermion conformal field theory, semi-infinite wedge representation of $GL(\infty)$ and the corresponding infinite Grassmannian; we refer to Appendix \ref{fermionreview} for notations and review, or the references \cite{miwa2000solitons}, \cite{Mulase2002ALGEBRAICTO}, \cite{segal-wilson}, \cite{witten} for more thorough discussions from several points of view. 

\subsubsection{Statement of mathematical results}
In concrete terms, the mathematical results shown in this paper are the following.
\begin{enumerate}
    \item For the curves $P(z, w) = w$ and $P(z, w) = zw - u$ (corresponding to the $z$-axis and a generic orbit of the torus action preserving the symplectic form on $T^* \mathbb{C}$) we study the (infinite rank) vector bundles
    \begin{equation}
        \text{Ext}^1_{\mathscr{D}_\hbar}(\widehat{\mathscr{O}}_C, \mathscr{E})
    \end{equation}
    over $\widetilde{M}(r)$. We prove that these can be used to define maps from $\widetilde{M}(r)$ to Sato-type Grassmannians (see e.g. \cite{Mulase2002ALGEBRAICTO}, \cite{segal-wilson}), which may be lifted via the Pl\"{u}cker embedding to define certain tensors in the fermionic Fock space. This is the content of Propositions \ref{maptoGr}, \ref{vaccomponent}, \ref{prop:VmaptoGr}, \ref{prop:S(u)vaccoeff}. These tensors in the fermionic Fock space are what we call ``instanton scattering matrices''. By construction, all of their components in the standard basis are global functions on the instanton moduli space. Also by construction, they define $\tau$-functions of multicomponent KP hierarchies for every point of $\widetilde{M}(r)$. 

    \item Proposition \ref{taufnformula} expresses an equivariance property of the instanton scattering matrix associated to $P(z, w) = w$, intertwining the action of the subgroup $\Gamma_+$ of the loop group of $\mathbb{C}^\times$ on the Sato Grassmannian (see \cite{segal-wilson}) with the action of a certain commutative subgroup of symplectic diffeomorphisms of $T^* \mathbb{C}$ on $\widetilde{M}(1)$. This was anticipated in \cite{Aganagic_2005}. 

    \item Theorem \ref{thm:instantonfactor} shows that on a distinguished open subset of $\widetilde{M}(1)$ (corresponding to the situation where the instantons are, in an appropriate sense, well-separated along the $z$-axis), the instanton scattering matrix associated to $P(z, w) = w$ factors into elementary pieces. The individual pieces are identified as explicit representation-theoretic objects: they are (Poisson limits of) \textit{Miura operators}. The higher rank analog of this theorem is also straightforward to obtain by the methods of this paper and discussed at the end of Section \ref{section:scatterandmiura}.

    Theorem \ref{thm:instantonfactor} can be viewed as a very explicit expression of the compatibility of the instanton scattering matrix construction with the factorization morphism of \cite{bfg} (those authors do not work in the noncommutative setting, but upon translating their construction to ADHM matrices it generalizes straightforwardly). 

    \item In Theorem \ref{thm:Rmatrixcoeff} we prove that, in a similar Poisson limit as above, all matrix coefficients of an arbitrary product of Miura operators (in the trigonometric normalization) extend from an analogous open subset to global functions on the affine deformation $\text{Hilb}_n(T^*_\hbar \mathbb{C}^\times)$ of the Hilbert scheme of points on $T^* \mathbb{C}^\times$. We also get an extremely concise formula \eqref{formula:cylindermiuracoeff} for the generating function of all matrix elements, with a geometric proof. In view of the $RTT$ relations with the $R$-matrix of $\mathsf{Y}(\widehat{\mathfrak{gl}}_1)$ from \cite{mo} enjoyed by Miura operators \cite{Prochazka_2019}, this gives an explicit formula assigning $R$-matrix coefficients to functions on $\text{Hilb}_n(T^*_\hbar \mathbb{C}^\times)$, in the semiclassical limit.

    \item The instanton scattering matrix associated to $P(z, w) = zw - u$ is denoted by $\mathbb{S}(u)$ in the main body. Theorem \ref{thm:fermionpropagator} gives an explicit formula which fully determines all the matrix coefficients of $\mathbb{S}(u)$ in terms of the ADHM data. Conjecturally, $\mathbb{S}(u)$ is the semiclassical limit of an $R$-matrix of $\mathsf{Y}_{+1}(\widehat{\mathfrak{gl}}_r)$ (see discussion in \ref{subsubsect:scatterandR} below for remarks on why this is not yet a theorem). 

    \item We give a conjectural formula \eqref{eq:Smatrixfromprop} extracting an infinite family of monopole scattering matrices (also known as matrix descriptions of affine Grassmannian slices, see \cite{bfnslice}, \cite{krylov}) from $\mathbb{S}(u)$, compatible with the embedding of affine Grassmannian slices into instanton moduli as fixed loci of an $S^1$-action. In this sense instanton scattering matrices are direct generalizations of their monopole counterparts. 
\end{enumerate}

\subsubsection{Statement of physics results}
For a string theorist, what is achieved in this paper is the following. We revisit the string theory setup of \cite{DHS}, \cite{DHSV}, which studied Type IIA superstring theory on the background $\mathbb{R} \times T^* \mathbb{C} \times \mathbb{C} \times \mathbb{R}^3$ in the presence of a $B$-field on $T^* \mathbb{C}$, together with a D6 brane wrapping $\mathbb{R} \times T^* \mathbb{C} \times \mathbb{C} \times \{ 0 \}$ and a D4 brane wrapping $\text{pt} \times \{ C \subset T^* \mathbb{C} \} \times \{ 0 \} \times \mathbb{R}^3$. The novel ingredient considered here is a collection of $n$ D2 branes wrapping $\mathbb{R} \times \text{points} \times \mathbb{C} \times \{ 0 \}$; owing to the nonzero $B$-field, the D2 branes dissolve into the D6 brane and are described in the low-energy gauge theory on the D6 brane as noncommutative instantons \cite{Nekrasov_1998}, \cite{seibergwitten} (we also consider at times a stack of $r > 1$ D6 branes, but with the noncommutativity turned on this generalization is comparatively trivial). 

The D4-D6 intersection supports chiral fermions, studied in \cite{DHS}, \cite{DHSV} in relation to the description of $B$-model topological strings obtained in \cite{Aganagic_2005}. The perspective of this paper is rather different: we study these ``noncommutative chiral fermions'' as probes of the noncommutative instanton configuration of the bulk gauge theory. The interaction with the bulk turns out to clarify certain features of these fermions. Proceeding along these lines, we show the following.

\begin{enumerate}
    \item As was the case in \cite{Aganagic_2005}, \cite{DHS}, \cite{DHSV}, the path integral over these fermionic degrees of freedom can be characterized via the chiral Ward identities: the Ward identities determine that the resulting path integral is fully described by a holomorphic map from the instanton moduli space $\widetilde{M}(r)$ to an infinite Grassmannian, and we characterize this map explicitly for the two simplest choices of D4 brane support $C \subset T^* \mathbb{C}$. This holomorphic map may be regarded as ``scattering data'' assigned to noncommutative instantons, generalizing the assignment of scattering data to monopoles via the classical constructions in \cite{atiyahhitchin}, \cite{hitchinspectral}, \cite{hurtubise}, \cite{jarvis}. This seems related to a very early observation of Atiyah \cite{atiyahinstanton}. 
    
    Making this statement work relies heavily on the approach to free fermions advocated by Witten in \cite{witten}. We call the corresponding fermion path integrals ``instanton scattering matrices''. 

    \item Uplifting to M-theory and invoking a related analysis of Costello \cite{costello2016}, instanton scattering matrices can be interpreted as computing (in the semiclassical limit) intersections of line and surface observables in the five-dimensional Chern-Simons theory on $\mathbb{R} \times T^*_\hbar \mathbb{C}$ studied in \cite{costello2016}. Using the known relation between defect intersections in such Chern-Simons theories and representation theory of Yangians, it is in this sense natural that instanton scattering matrices give rise to $R$-matrices of (shifted) affine Yangians $\mathsf{Y}_{+ 1}(\widehat{\mathfrak{gl}}_r)$, in a certain semiclassical limit.

    $R$-matrices for shifted Yangians are not yet defined in the correct level of generality (see \ref{subsubsect:scatterandR} below), so in cases like that considered in Section \ref{section:newRmatrix} our formalism leads to predictions and constraints on the $RTT$ formalism for these Yangians. 

    \item Viewed in the 5d Chern-Simons context, the techniques of this paper give rise to new methods to compute such defect intersections. Two especially important features that should be highlighted are: 
    \begin{itemize}
        \item The methods are \textit{nonperturbative}, in the sense that they do not rely on a power series expansion about the trivial field configuration in the bulk Chern-Simons theory. Indeed, by construction they work over the whole instanton moduli space. 

        \item Our methods do not limit the support of the surface operator to be a linear subspace in $T^* \mathbb{C}$, the only case considered so far in the literature. Indeed the formalism of Appendix \ref{adhmdetails} works, at least in principle, equally well for arbitrary supports. Section \ref{section:newRmatrix} gives a detailed analysis of the simplest nontrivial example; the operator assigned to the intersection is denoted $\mathbb{S}(u)$.
    \end{itemize}

    The main limitation of the approach of this paper is that we must stay in the aforementioned semiclassical limit, but this is not as severe as it may seem, especially if one uses the (highly constraining) relation to the Yangian.
    
    \item In the approach of this paper, instanton scattering matrices are placed on equal footing with monopole scattering matrices. This is made precise in \eqref{eq:Smatrixfromprop} giving an explicit formula for the monopole scattering matrices in terms of the instanton ones, in certain situations. The embedding of monopole $S$-matrices into the four-dimensional Chern-Simons theory was discussed (in a slightly different language) in \cite{costelloQop}. 

    \item Theorem \ref{thm:instantonfactor} gives a new approach to the appearance of Miura operators in the M-theory setup of \cite{costello2016}, previously studied in \cite{miuraCSfeynman2}, \cite{gaiottorapcak2020}, \cite{gaiottorapcakzhou24}, \cite{haouzijeong}, \cite{miuraCSfeynman1}, \cite{zenkevich24}. It says that the instanton scattering matrix corresponding to an M5 brane supported on $w = 0$ intersecting an arbitrary number of M2 branes factors into Miura operators in an appropriate coordinate chart of the instanton moduli space.  

    This point of view allows us to prove results like Theorem \ref{thm:Rmatrixcoeff}, or \eqref{formula:cylindermiuracoeff} giving very precise characterizations of arbitrary correlation functions involving Miura operators, which are extremely striking from the purely two-dimensonal CFT perspective. In this sense we have placed Miura operators on equal footing with monopole $S$-matrices/affine Grassmannian slices, explaining their significance in a much wider context.
\end{enumerate}

\subsubsection{Instanton scattering matrices and $R$-matrices} \label{subsubsect:scatterandR}

A major motivation for this paper is the observation that monopole scattering matrices are closely linked with $R$-matrices of certain quantum groups called shifted Yangians (say of the Lie algebra $\mathfrak{g} = \mathfrak{gl}_r$). This observation goes back at least to \cite{Gerasimov_2005} and will be reviewed in Section \ref{sect:monopoleyangian} below; see also \cite{costelloQop} for a more recent discussion from the point of view of the four-dimensional Chern-Simons theory. Our instanton scattering matrices are constructed so that the analogous relation continues to hold if the shifted Yangians are replaced by their affine counterparts, i.e. shifted Yangians of $\widehat{\mathfrak{gl}}_r$. This is expected based on string theory considerations in Section \ref{section:string}, and in some situations (like Lemma \ref{lemma:RTTaffine} below) it is easy to verify the Yang-Baxter equation directly, justifying the $R$-matrix property independently of the string theory argument. 

Strictly speaking, it is \textit{quantizations} of scattering matrices that are $R$-matrices; scattering matrices themselves are semiclassical objects, though they detect the Yang-Baxter equation via the Poisson bracket on the relevant monopole/instanton moduli space (see e.g. the computation done in Appendix \ref{reviewscatter}). 

Despite the fact that we give in the main body a mathematically precise definition of the instanton scattering matrices, the relation to $R$-matrices remains, strictly speaking, conjectural for the following reason. Shifted Yangians $\mathsf{Y}_\mu(\mathfrak{g})$ depend on a choice of coweight $\mu$ of Lie algebra $\mathfrak{g}$ called the shift, and the existence of a Yang-Baxter formalism for the Yangian depends on if $\mu$ is dominant or antidominant. If $\mu$ is antidominant, the Yang-Baxter formalism exists and was constructed for $\mathfrak{g} = \mathfrak{gl}_r$ in \cite{frassekpestuntsymb}. If $\mu$ is dominant, the construction of \cite{frassekpestuntsymb} fails, there does not appear to be a Yang-Baxter equation, and the naively defined coproduct fails to be coassociative. This is a well-known technical problem in representation theory. 

What this means is, in the dominantly shifted case, it is not possible to prove that instanton scattering matrices are $R$-matrices because the statement is simply devoid of content in absence of an $RTT$ formalism for shifted Yangians. We prefer to read this correspondence in the other direction: because scattering matrices make sense in a higher level of generality than $R$-matrices, they should instead be made into a goalpost for the eventual development of the $RTT$ theory of shifted Yangians. Namely, one should formulate an $RTT$ formalism for the $1$-shifted affine Yangian $\mathsf{Y}_{+1}(\widehat{\mathfrak{gl}}_r)$ and show that the operators $\mathbb{S}(u)$ constructed in Section \ref{section:newRmatrix} satisfy the defintion of its $R$-matrices (in the semiclassical limit).

\subsection{Acknowledgements}
This paper would not exist without invaluable discussions with Mina Aganagic and Yegor Zenkevich. I am grateful to each of them for providing feedback and suggestions throughout the development of this project, and especially to Yegor for generously sharing his formulas and thoughts on $q$-deformed Miura operators. I am also grateful to Tommaso Botta and Sujay Nair for collaboration on related work which very much influenced my thinking and motivated me to pursue the constructions in this paper. 

Part of this work was completed during a visit to Columbia University, and I am grateful to the Department of Mathematics there for its hospitality. Preliminary results were announced at the Columbia Informal Mathematical Physics Seminar (January 2025). I am especially grateful to Andrei Okounkov for the invitation and for his feedback and interest in the project. 

\subsubsection{Outline of paper}
This paper is organized as follows. In Section \ref{sect:monopoleyangian} we give a brief review of monopole scattering matrices, emphasizing their relations to Yangians and the features that are paralleled by their instanton analogs. Section \ref{sect:monopoletoinst} explains, in physics language, the step-by-step generalization from monopole scattering matrices to their instanton analogs. Section \ref{sect:Smatrix1stpass} gives a mathematically rigorous treatment of the simplest instanton scattering matrices, associated to fermions on the zero section $\mathbb{C} \subset T^*_\hbar \mathbb{C}$ or $\mathbb{C}^\times \subset T^*_\hbar \mathbb{C}^\times$. Section \ref{section:string} uplifts the construction to string theory and explains the connections with other parts of the physics literature in more depth. Sections \ref{sect:defineS}, \ref{sect:determineS} give a mathematical treatment of the next simplest instanton scattering matrix (denoted $\mathbb{S}(u)$), associated to fermions on the curve $zw = u$, motivated by the discussion in Section \ref{section:string}. Section \ref{sect:monopolefrominstanton} closes the circle of ideas and gives an explicit recipe to recover monopole scattering matrices from their instanton analogs. 

Appendix \ref{reviewscatter} gives a much more detailed introduction to monopole scattering matrices, aimed at readers unfamiliar with the construction. Appendix \ref{fermionreview} collects notations, concepts, and results from the theory of free fermions and infinite Grassmannians that will be used throughout the paper. Appendix \ref{adhmdetails} contains various technical computations involving the ADHM complex. Some parts of the main body of the paper, especially the proofs, rely on notions or results introduced in Appendices \ref{fermionreview} and \ref{adhmdetails}. Readers less familiar with this background may wish to consult those appendices before reading the rest of the paper. 

\textit{Remark on differential geometry.} For conceptual purposes, and especially to make an honest connection with physics, it is very important that various moduli spaces under consideration in this paper have interpretations in differential geometry. However, none of the proofs in this paper depend on such a description, and the reader will note that they are carried out entirely in algebra-geometric terms. Still, we will comment at various points on the differential-geometric interpretation, without addressing any serious analytic complications that may or may not arise. 

\section{Monopole moduli spaces and Yangians} \label{sect:monopoleyangian}
The central observation that we aim to reexamine and generalize in this paper is the relationship between the monopole scattering transform and the $RTT$ formalism for Yangians. This connection goes back at least to \cite{Gerasimov_2005} and has reappeared in the modern literature in relation to the quantization of Coulomb branches of quiver gauge theories in various dimensions \cite{nekrasovpestun}, \cite{BDG}, \cite{costelloQop}. In this section we will survey the important aspects of monopole scattering that will motivate our main construction. While monopole moduli may be studied for arbitrary ADE Lie algebras $\mathfrak{g}$, attention will be restricted to $\mathfrak{g} = \mathfrak{gl}_r$ for concreteness. Explicit formulas will be written only for $\mathfrak{gl}_2$. 

Readers familiar with this material should just skim this section for notation; readers unfamiliar with it may wish to consult Appendix \ref{reviewscatter} for a more detailed and complete introduction to the subject. 

\subsection{Scattering transform and monopole moduli}

\subsubsection{General setup}
The setting for studying BPS monopole equations (also called Bogomolny equations) is the following. We study unitary connections $A$ on rank $r$ vector bundles $E \to \mathbb{R}^3$, and also real-valued sections $\phi$ of $\text{End}(E)$. We require that they solve the equation 
\begin{equation}
    F_A + \star D_A \phi = 0
\end{equation}
where $F_A = dA + A \wedge A$ is the curvature of the connection, $D_A \phi = d\phi + \comm{A}{\phi}$ is the covariant derivative, and $\star$ is the Hodge star determined by the Euclidean metric on $\mathbb{R}^3$. This will be referred to as the monopole equation in this paper. 

The classical setup is to study solutions of the above equation smooth throughout $\mathbb{R}^3$ and satisfying certain boundary conditions. It is also interesting to study solutions defined on $\mathbb{R}^3 \setminus \{ 0 \}$ with prescribed singularities near zero. We consider solutions with boundary conditions at $r \to \infty$ specified by equation \eqref{monopolebcinfty}, which involves the choice of a coweight $\mu$ of $GL_r$, and prescribed singularity \eqref{singularbcmonopole} at $r \to 0$ labeled by another choice of coweight $\lambda$. For concreteness, we write the boundary conditions on the scalar $\phi$ here, referring to Appendix \ref{reviewscatter} for more detail:
\begin{equation}
\begin{split}
    \phi(\vec{x}) & \to -i \text{diag}(\phi_1, \dots, \phi_r) - \frac{i\mu}{2r} + \dots \, \, \, \text{as $r \to \infty$} \\
    \phi(\vec{x}) & \to -\frac{i\lambda}{2r} + \dots \, \, \, \text{as $r \to 0$}.
\end{split}
\end{equation}
The $\phi_i$ are real numbers in the order $\phi_1 < \dots < \phi_r$. The gauge field $A$ is required to approach the standard Dirac monopole configurations at those points specified by the choice of coweights; the moduli space of solutions to $F_A + \star D_A \phi  = 0$ subject to these boundary conditions and considered up to natural isomorphisms (gauge equivalence, approaching $1$ at $r \to \infty$ and centralizing $\lambda$ near $0$) is denoted $\mathscr{M}^\lambda_\mu$. For the moduli space to be nonempty, $\mu - \lambda$ must be a coroot. When the space is nonempty it is of real dimension $4 | \langle \rho, \mu  - \lambda \rangle|$ ($\rho$ is the Weyl vector, half the sum of the positive roots); this can be seen differential-geometrically \cite{Moore_2014}. 

The $\mathscr{M}^\lambda_\mu$ fail to be compact for two reasons: one is due to the noncompactness of $\mathbb{R}^3$, and the other is the ``monopole bubbling effect'' \cite{KW}, which is in a precise sense analogous to effects related to small instantons in four dimensions. As is the case with instantons, there is an Uhlenbeck-type partial compactification of the monopole moduli space $\mathscr{M}^\lambda_\mu \subset \overline{\mathscr{M}}^\lambda_\mu$, which is in general singular (it is smooth if and only if $\lambda$ is minuscule). The construction of this compactification, at least in algebro-geometric terms, may be done using the scattering transform and is sketched in Appendix \ref{reviewscatter}. 

\subsubsection{Scattering matrix}
The differential-geometric description of $\mathscr{M}^\lambda_\mu$ is conceptually very nice, but rather difficult to work with because it involves solving partial differential equations. Historically, the monopole scattering matrix was of interest because it gave an accessible description of $\mathscr{M}^{\lambda = 0}_\mu$ in finite-dimensional terms \cite{jarvis}; the rank $r = 2$ case is discussed thoroughly in \cite{atiyahhitchin}. 

The basic idea is as follows: pick a splitting $\mathbb{R}^3 \simeq \mathbb{C} \times \mathbb{R} \ni (u, y)$. This induces a preferred choice of complex structure on $\mathscr{M}^\lambda_\mu$, see Appendix \ref{reviewscatter} or \cite{atiyahhitchin}. The equations $F_A + \star D_A \phi = 0 $ become 
\begin{equation}
    \comm{\overline{D}_u}{D_y + i \phi} = 0
\end{equation}
and another equation that can be interpreted as a real moment map condition. The quantity (Wilson line)
\begin{equation}
    S(u) = \mathscr{P}\exp{ - \int_{\gamma_{u}} dy(A_y + i \phi) }
\end{equation}
is the holonomy of the complexified connection $A_y + i \phi$ along the vertical line $\gamma_u$ given by $u = \text{const}$. It is covariantly holomorphic in $u$ as a consequence of the above equation. Since the restriction of the bundle $E$ to $y = \text{const}$ is a vector bundle on a Riemann surface with unitary connection, $\overline{D}_u$ gives it a natural holomorphic structure and when expressed with respect to holomorphic trivializations at $y \to \pm \infty$, the matrix elements of $S(u)$ are holomorphic in $u$. By design, they are also holomorphic functions on the monopole moduli space in the complex structure induced by the splitting $\mathbb{R}^3 \simeq \mathbb{C} \times \mathbb{R}$. As reviewed in Appendix \ref{reviewscatter}, the boundary conditions on the monopole equations can be used to determine the asymptotics of the matrix elements of $S(u)$ as $u \to \infty$ in some class of preferred asymptotic holomorphic frames, and this is strong enough to fix $\overline{\mathscr{M}}^\lambda_\mu$ completely (and moreover give it the structure of a singular affine algebraic variety). 

\subsubsection{Explicit scattering matrices in rank 2}
To make this more concrete, let us just state the outcome of this analysis for $U(2)$ monopoles, where the asymptotic charge is $\mu = (k, n - k)$ and the Dirac singularity at $0$ has charge $\lambda = (0, n)$ where $n, k \geq 0$ are integers. 

Then as detailed in Appendix \ref{reviewscatter}, the asymptotics of the monopole solution require the scattering matrix to be of the form 
\begin{equation}
    S(u) = \begin{bmatrix} Q(u) & U^+(u) \\ U^-(u) & \widetilde{Q}(u) \end{bmatrix}
\end{equation}
where all entries are polynomials in $u$. Moreover $Q(u)$ is a monic polynomial of degree $k$, $\deg U^\pm(u) \leq k - 1$ and 
\begin{equation}
    \det S(u) = Q(u) \widetilde{Q}(u) - U^+(u) U^-(u) = u^n. 
\end{equation}
Taking coefficients of powers of $u$ in this equation gives the complete set of equations cutting out $\overline{\mathscr{M}}^\lambda_\mu$ as an (in general singular) affine algebraic variety; ${\mathscr{M}}^\lambda_\mu$ may be recovered as its smooth locus. When $\lambda$ vanishes or is minuscule, there is no singularity. 

It is easy to see (Section 2(xi) in \cite{bfnslice}) that this agrees with the description of generalized affine Grassmannian slices in \cite{bfnslice}, which is of course well-known going back to \cite{BDG}.

\subsection{Integrable structures of monopole moduli}

\subsubsection{Symplectic structure of the monopole moduli space}
Recall that the moduli space of monopoles is a holomorphic symplectic variety with symplectic structure descending formally from the following two-form on the space of all gauge fields and scalars:
\begin{equation}
    \omega_{\mathbb{C}} = \int_{\mathbb{R}^3 \setminus \{ 0 \}} d^2u dy \Tr( \delta A_{\bar{u}} \wedge \delta(A_y + i \phi)). 
\end{equation}
The matrix elements of the $S$-matrix are holomorphic functions of $u$ with coefficients in holomorphic functions on the moduli space of monopoles. In Appendix \ref{reviewscatter} it is shown by direct computation that the Poisson bracket of $S$-matrix coefficients takes the form 
\begin{equation} \label{eq:monopoleRTT}
    \acomm{S\indices{^i_j}(u)}{S\indices{^k_\ell}(v)} = \frac{S\indices{^i_\ell}(u) S\indices{^k_j}(v) - S\indices{^i_\ell}(v) S\indices{^k_j}(u)}{u - v} + \Delta\indices{^{ik}_{j \ell}}(u, v)
\end{equation}
where the term $\Delta$ involves higher powers of $u$ and $v$ that we may avoid determining explicitly by the procedure sketched in Appendix \ref{reviewscatter} (this will also be illustrated in the next paragraph). If the $\Delta$ term vanishes or is ignorable, this agrees with the Poisson limit of the $RTT$ relation of the Yangian of $\mathfrak{gl}_r$ (see e.g. \cite{FADDEEV_1995} for an introduction, or \cite{frassekpestuntsymb} for a modern discussion in a closely related setting). Closer inspection shows that we in fact have a shifted Yangian, where the value of the shift is fixed by the asymptotic charge $\mu$, \cite{frassekpestuntsymb}. 

In the rank $r = 2$ case the complete list of Poisson brackets is, with notations for the $S$-matrix elements as above in our running example,
\begin{equation}
\begin{split}
    \acomm{Q(u)}{Q(v)} & = 0 \\ 
    \acomm{Q(u)}{ \frac{U^\pm(v)}{Q(v)}} & = \pm \Bigg[ \frac{Q(u)}{u - v} \Big( \frac{U^\pm(u)}{Q(u)} - \frac{U^\pm(v)}{Q(v)} \Big) \Bigg]_{-} \\
    \acomm{\frac{U^+(u)}{Q(u)}}{\frac{U^-(v)}{Q(v)}} & = \Bigg[ \frac{1}{u - v} \Big( \frac{P(v)}{Q(v)^2} - \frac{P(u)}{Q(u)^2} \Big) \Bigg]_-
\end{split}
\end{equation}
where the subscript $-$ means take the negative part in the expansion in $u$ or $v$ whenever there is a $U^\pm/Q$ on the LHS. The polynomial $P(u) = \det S(u) = u^n$. The $\Delta$ term vanishes whenever $n \leq 2k$ since in that case the ``$-$'' subscripts on the RHS are superfluous; when $n > 2k$ the term $\Delta$ is nonzero and its sole role is to cancel the positive degree terms appearing in the bracket between $U^\pm/Q$.

From this point of view, the appearance of Yangians appears rather surprising and unexpected. A more conceptual, but also more sophisticated, rationale for its appearance will be sketched in Section \ref{section:string}. For now we will take it for granted and deduce consequences. 

\subsubsection{Coordinate charts and comultiplication}
The relationship of monopole moduli spaces to Yangians, or equivalently the fact that monopole moduli turn out to be the same as the moduli space of $S$-matrices, gives them a considerable amount of rigidity and additional structure. This is because ``tautological'' structures defined on the Yangian can descend nontrivially to the monopole moduli spaces. 

For a simple example, suppose $S_1(u)$ is a monopole $S$-matrix corresponding to values of the charges $(\lambda_1, \mu_1)$ and likewise $S_2(u)$ is a  monopole $S$-matrix for charges $(\lambda_2, \mu_2)$. Assume further for simplicity that we are in the rank $r = 2$ case, and both $S_1$ and $S_2$ have $n \leq 2 k$ in notations above (corresponding to antidominantly shifted Yangians, so the analysis of \cite{frassekpestuntsymb} applies). Then it is easy to see that 
\begin{equation}
    S\indices{^i_j}(u) = \sum_k (S_1)\indices{^i_k}(u) (S_2)\indices{^k_j}(u)
\end{equation}
is another monopole $S$-matrix, associated to charges $(\lambda_1 + \lambda_2, \mu_1 + \mu_2)$. It is also easy to check that if $\Delta = 0$, this map is compatible with Poisson brackets. This describes a Poisson morphism 
\begin{equation}
    \overline{\mathscr{M}}^{\lambda_1}_{\mu_1} \times \overline{\mathscr{M}}^{\lambda_2}_{\mu_2} \longrightarrow \overline{\mathscr{M}}^{\lambda_1 + \lambda_2}_{\mu_1 + \mu_2}. 
\end{equation}
Moreover, in the $RTT$ formalism, matrix multiplication in the ``auxiliary'' space $\mathbb{C}^2$ is the \textit{comultiplication} on the Yangian, see \cite{FADDEEV_1995}, \cite{frassekpestuntsymb}. Thus the $S$-matrix description of monopole moduli spaces gives rise to a collection of canonical maps between monopole moduli spaces, simply induced by the coproduct in the underlying Yangian. This gives a systematic procedure to decompose monopole moduli spaces into elementary constituents. More precisely, these maps can be used to construct a system of coordinate charts on resolutions of $\overline{\mathscr{M}}^\lambda_\mu$ (see \cite{krylov} for details; the most basic example will be discussed below). The appearance of resolutions is unsurprising given the description of resolutions in \cite{KW}. 

If we allow $n > 2k$ in the rank 2 case, the reader may verify that there is a still well-defined map as above but it is no longer given by naive matrix multiplication. It is necessary to correct the matrix multiplication by further multiplying by strictly upper/lower triangular matrices to get the correct degrees in $u$ (this is traced back to the Gauss factorizations of $S$-matrices appearing in Appendix \ref{reviewscatter}). The correction introduces no new information because it is uniquely determined by the degree conditions. This basic observation is the source of difficulties extending the $RTT$ formalism of \cite{frassekpestuntsymb} to dominantly shifted Yangians.

\subsubsection{Action-angle variables}
The ``comultiplication'' coordinate charts discussed algebraically above can be interpreted more physically as follows. Splitting $\lambda = \lambda_1 + \lambda_2$ corresponds to partially resolving the singularities of $\overline{\mathscr{M}}^\lambda_\mu$ by splitting apart singular monopoles as in \cite{KW}. The locus where the $S$-matrix factors is the locus where the smooth monopoles are separated into two ``clusters'' along the $y$-axis, one which is allowed to bubble into the $\lambda_1$ Dirac monopole and one which is allowed to bubble into the $\lambda_2$ Dirac monopole. This is an open set in the partial resolution, isomorphic to a product of two smaller monopole moduli spaces. This intuitive picture is confirmed by the results in \cite{krylov}. 

There is another possible open subset in the monopole moduli space, in which we suppose the smooth monopoles are well-separated along the $u$-plane. This leads to a coordinate chart in which the integrable system structure of the monopole moduli space is manifest, and it gives rise to the ``action-angle'' coordinates of the integrable system. The ``factorized $S$-matrix'' coordinate charts instead correspond to the ``coordinate basis'' of the integrable system, in accordance with general features of algebraic Bethe ansatz \cite{FADDEEV_1995}. 

In the rank $2$ case with $\mu = (k, n - k)$, $\lambda = (0, n)$ $\overline{\mathscr{M}}^\lambda_\mu$ has dimension $2k$, and from the above Poisson brackets it follows that the coefficients of $Q(u)$ supply $k$ independent Poisson-commuting functions on it. Rephrased geometrically, we have a map 
\begin{equation}
    \overline{\mathscr{M}}^\lambda_\mu \to \text{Sym}^k(\mathbb{C})
\end{equation}
given by the assignment $S(u) \mapsto \text{zeroes of $Q(u)$}$, and it has Lagrangian fibers. Introducing coordinates $\varphi_i$ on the base, $Q(u) := \prod_i (u - \varphi_i)$, in the generic fiber we may introduce coordinates $x^\pm_i := U^\pm(\varphi_i)$, which due to the relation $\det S(u) = P(u)$ must satisfy $-x^+_i x^-_i = P(\varphi_i)$ for each $i = 1, \dots, k$; thus the generic fiber is identified with $(\mathbb{C}^\times)^k$. It is straightforward to see that the symplectic form is identified with $\Omega = \sum_i d\varphi_i \wedge d\log(x_i^+)$ in these coordinates. In this way one sees that the monopole moduli space is an integrable system and that the $(\varphi_i, x_i^+)$ provide an explicit set of action-angle coordinates; these are the coordinates naturally associated to abelianization of \cite{BDG}, \cite{bfnslice}. 

\subsubsection{Example: Open Toda chain}
Let us illustrate the above remarks in the rank $r = 2$ case with $\mu = (k, - k)$, $\lambda = 0$. In this case $\overline{\mathscr{M}}^0_\mu$ is a smooth affine variety of dimension $2k$; this is the classical case studied in \cite{atiyahhitchin}. The coordinate chart arising from comultiplication as above identifies it as the completed phase space of the open Toda lattice.

For $k = 1$, the scattering matrix is determined to be of the form 
\begin{equation}
S(u) = \begin{bmatrix} u - p && z \\ -z^{-1} && 0  \end{bmatrix}
\end{equation}
and the nontrivial Poisson bracket is $\acomm{p}{z} = z$; therefore we see that $\overline{\mathscr{M}}^0_{(1, -1)} \simeq T^* \mathbb{C}^\times$ with the standard symplectic form, and $(p, z)$ are the standard coordinates on the cotangent bundle. In this case the integrable system is the trivial one corresponding to the projection to the $p$-plane. 

Then we may write a factorized $S$-matrix via coordinates $(p_i, z_i)$, $i = 1, \dots, k$ as
\begin{equation} \label{eq:monopoleSfactor}
\begin{split}
    S(u) & = \begin{bmatrix} Q(u) && U^+(u) \\ U^-(u) && \widetilde{Q}(u) \end{bmatrix} \\
    & = \begin{bmatrix} u - p_k && z_k \\ -z_k^{-1} && 0  \end{bmatrix} \dots \begin{bmatrix} u -  p_1 && z_1 \\ -z_1^{-1} && 0 \end{bmatrix}.
\end{split}
\end{equation}
This equation defines a map 
\begin{equation}
    T^* (\mathbb{C}^\times)^k \simeq (\overline{\mathscr{M}}^0_{(1, - 1)})^{\times k} \to \overline{\mathscr{M}}^0_{(k, -k)}
\end{equation}
and it is not hard to check directly that it is actually the inclusion of an open subset. Moreover, by the compatibility with Poisson brackets discussed above, $(p_i, z_i)$ are Darboux coordinates with $\Omega = \sum_i dp_i \wedge d \log (z_i)$.

The Hamiltonians $H_k$ of the integrable system are the coefficients of $Q(u)$: $Q(u) = u^k - H_1 u^{k - 1} + H_2 u^{k - 2} + \dots$. From the above we read off 
\begin{equation}
\begin{split}
    H_1  & = \sum_i p_i \\
    H_2 & = \sum_{i < j} p_i p_j - \sum_{i = 1}^k z_{i + 1} z_i^{-1}
\end{split}
\end{equation}
and observe immediately that $H_2$ is the Hamiltonian of the $k$-particle open Toda lattice. 

The following are the essential features of this example that we would like to understand in more generality, in particular for instantons.

\begin{enumerate}
    \item There is an essentially canonical construction of an operator $S(u)$ from a monopole solution $(A, \phi)$; the matrix coefficients are holomorphic functions on the monopole moduli space and give a distinguished set of generators for its coordinate ring. (Instanton analogs: Propositions \ref{maptoGr}, \ref{vaccomponent}, \ref{prop:VmaptoGr}, \ref{prop:S(u)vaccoeff}). 

    \item The symplectic structure on the monopole moduli space is such that $S(u)$ may be viewed as a semiclassical limit of a transfer matrix for the Yangian $\mathsf{Y}(\mathfrak{g})$, where $\mathfrak{g}$ is the Lie algebra of the monopole gauge group, and we ignore details about the shift. (Instanton analog, in 0-shifted case: Lemma \ref{lemma:RTTaffine}). 

    \item The Yangian coproduct (multiplication of the $T$-matrices) induces canonical maps between the monopole moduli spaces; in particular there is a distinguished open set in the multi-monopole moduli space characterized by the complete factorization of $S(u)$ into elementary constituents. (Instanton analogs: Theorems \ref{thm:instantonfactor}, \ref{thm:Rmatrixcoeff}, \ref{thm:Scolumns factor}).
\end{enumerate}

\subsubsection{Quantization and $R$-matrices}
The next three sections have no direct relevance to the rest of the paper and exist only for the purpose of motivation and context; readers finding them confusing are encouraged to skip them.

In this paper, to be brief it will often be said that scattering matrices ``are'' $R$-matrices or $T$-matrices, but the precise statement is that this is only true upon passing to a quantization of the ring of functions on the monopole moduli space and choosing a module over the quantization to evaluate the $S$-matrix coefficients in. For simplicity, in the unshifted $\mathsf{Y}(\mathfrak{gl}_r)$ case (corresponding to $\mu = 0$), the quantized $S$-matrix is expected to give $R$-matrices of $\mathsf{Y}(\mathfrak{gl}_r)$ braiding 
\begin{equation}
    \widehat{S}(u) \in \text{End}(\mathbb{C}^r \otimes M)(u)
\end{equation}
where $M$ is some module over the quantized ring of global functions on $\overline{\mathscr{M}}^\lambda_{\mu = 0}$. The quantized coordinate ring of monopole moduli is a quotient of the corresponding Yangian and the Yangian action on $M$ is defined by pullback. If $M$ is a tensor product of $\mathsf{Y}(\mathfrak{gl}_r)$-modules then $\widehat{S}(u)$ factors according to the corresponding coordinate chart on $\overline{\mathscr{M}}^\lambda_{\mu = 0}$ induced by coproduct.

Using the identification of the monopole moduli space with an affine Grassmannian slice, and of the latter with a BFN-style Coulomb branch \cite{bfnslice} gives a canonical quantization of the $S$-matrices. That this quantization is compatible with factorization of $S$-matrices (i.e. quantized multiplication morphisms, or the honest Yangian coproduct) can be shown using enumerative geometry of quasimaps to flag varieties, see Section 4 of \cite{tamagnishiftop} and the forthcoming work \cite{nairtamagni}. This can be thought of heuristically as a ``stable envelope'' construction in the style of \cite{AO}. It has a parallel for monopole moduli on $\mathbb{R}^2 \times S^1$ in the context of instanton counting in 4d $\mathscr{N} = 2$ theories in the NS limit, see \cite{nekrasovlee}, quantizing the discussion from \cite{nekrasovpestun}.

\subsubsection{General story: multiplication morphisms of affine Grassmannian slices}
In the rank 2 case, the construction of coordinate charts on resolved monopole moduli can be done in a very ``by hand'' fashion by multiplying $2 \times 2$ matrices as above, and it is an instructive exercise to do so. These charts identify the monopole moduli spaces as completed phase spaces for ``spin chain'' like integrable systems. That a similar construction works in much higher generality (though always for Yangians of finite-dimensional Lie algebras) has been shown in the affine Grassmannian language by Krylov and Perunov \cite{krylov}. 

Resolutions of $\overline{\mathscr{M}}^\lambda_\mu$ corresponding to splitting Dirac charge as in \cite{KW} can be constructed in algebra-geometric language using the convolution diagram of the affine Grassmannian, and the authors of \cite{krylov} consider such partial resolutions of the generalized transversal slices introduced in \cite{bfnslice}. It is shown that when a full resolution exists (and when $\mu$ is ``sufficiently positive'', a condition analogous to $n > 2k$) that it can be covered by affine spaces $T^* \mathbb{A}^n$ in a fashion compatible with natural torus actions and arising from the image of morphisms defined by multiplication of $S$-matrices (perhaps up to triangular ambiguities as discussed above). The $T^* \mathbb{A}^n$ entering in the charts should be thought of as products of $\overline{\mathscr{M}}^\lambda_\mu$ for ``small enough'' $\lambda$ and $\mu$. The map on singular affine varieties considered above is the image of the corresponding chart in the blowdown (what is called $\overline{\mathscr{M}}^\lambda_\mu$ here is called $\overline{\mathscr{W}}^\lambda_\mu$ there, see also Appendix \ref{reviewscatter} for more discussion). 

\subsubsection{Instanton scattering and the double affine Grassmannian}
Naive generalization of the algebra-geometric description of $\overline{\mathscr{M}}^\lambda_\mu$ to the case where $\mathfrak{gl}_r$ is replaced by $\widehat{\mathfrak{gl}}_r$ would seem to involve the affine Grassmannian of a Kac-Moody group, called the double affine Grassmannian in \cite{Braverman_2010}. Unfortunately, it is notoriously challenging to reason with the double affine Grassmannian beyond the naive level, however in \cite{Braverman_2010} compelling evidence was presented that instanton moduli spaces should play the role of its transversal slices (this can also be seen from the point of view of Coulomb branches).  

In much the same way that monopole $S$-matrices can be viewed as ``matrix descriptions'' of affine Grassmannian slices in the terminology of \cite{bfnslice}, \cite{krylov}, instanton scattering matrices to be introduced below can be viewed as ``matrix descriptions'' of double affine Grassmannian slices. The caveat is that one must work with matrices of infinite size; this is roughly because the analog of the fundamental representation of $\mathfrak{gl}_r$ is the Fock module of $\widehat{\mathfrak{gl}}_r$ (at level $k = 1$). The sought-after ``instanton $S$-matrix'' has to be a semiclassical $R$-matrix of a shifted variant of $\mathsf{Y}(\widehat{\mathfrak{gl}}_r)$, which lacks a supply of finite-dimensional modules to use as the auxiliary space.  

From this perspective, it is clear why naive guesswork is insufficient to construct such a matrix description: the coordinate ring of instanton moduli is of the same size as the coordinate ring of monopole moduli (at given charge), but the Yangian $\mathsf{Y}(\widehat{\mathfrak{gl}}_r)$ is much larger than $\mathsf{Y}(\mathfrak{gl}_r)$, so the assignment of $S$-matrix coefficients to functions must factor through a much larger quotient. This problem will be circumvented below using a combination of geometry and field-theoretic reasoning. 

\section{Instanton scattering: first approach} \label{section:scatterandmiura}
In this section we will present a ``bottom-up'' approach to the definition and computation of instanton scattering matrices, by straightforwardly generalizing the field theory discussion of the previous section. In later parts of the paper, we will give a ``top-down'' approach to the problem using string theory, which will also explain how this construction is related to others in the literature. Section \ref{sect:monopoletoinst} is motivational and aimed at physicists, while Section \ref{sect:Smatrix1stpass} contains the precise mathematical statements. 

\subsection{Motivation and tentative definition} \label{sect:monopoletoinst}

\subsubsection{Instantons vs. monopoles} 
Let us first briefly recall the precise sense in which instantons may be regarded as an $S^1$-uplift of monopoles. We now study rank $r$ vector bundles $E \to \mathbb{R}^3 \times S^1$ with unitary connections $A$, assuming the standard flat metric on $\mathbb{R}^3 \times S^1$. In four dimensions, the Hodge star operator satisfies $\star^2 = 1$ and preserves the space $\Omega^2(\mathbb{R}^3 \times S^1, \text{End}(E))$ so we may impose the equation 
\begin{equation}
    F_A^+ = 0 
\end{equation}
on the curvature of $A$, meaning that $F_A = - \star F_A$. Let us write $(x^1, x^2, x^3, \theta)$ for the coordinates on $\mathbb{R}^3 \times S^1$. Solutions to $F_A^+ = 0$ approaching a flat connection on $S^1$ sufficiently rapidly as $|x| \to \infty$ in $\mathbb{R}^3$ are called instantons on $\mathbb{R}^3 \times S^1$ (in defining the asymptotic boundary condition, one also turns on some Dirac flux through the asymptotic $S^2$ boundary of $\mathbb{R}^3$ as in the monopole case). 

The smooth solutions (i.e. $\lambda = 0$ in the notation of the previous section) to the monopole equations on $\mathbb{R}^3$ are precisely the $S^1$-invariant instantons on $\mathbb{R}^3 \times S^1$. To see this, write an $S^1$-invariant connection in coordinates as $A = A_i(x) dx^i + A_\theta(x) d \theta$ and observe that if we identify $\phi(x) = A_\theta(x)$, then $F_A^+ = 0$ reduces in components to $F_{ij} + \epsilon_{ijk} D_k \phi = 0$. The asymptotic eigenvalues $\phi_1, \dots, \phi_r$ of $\phi(x)$ from the previous section specify the asymptotic flat connection on $S^1$. 

There is a version of this correspondence that incorporates Dirac singularities, essentially due to Kronheimer \cite{kronheimer} that will play a role in Section \ref{sect:monopolefrominstanton} but we will not need it in this section. 

\subsubsection{Free fermions}
The goal is to define the instanton analog of the quantity $S(u)$. To do this effectively it will be convenient to first obtain an alternate formula for $S(u)$ which generalizes more straightforwardly. 

The basic idea to obtain this description is the following. By definition, $S(u)$ is the classical value of a Wilson line of $A_y + i  \phi$ in the background of a monopole configuration. Wilson lines may be interpreted as the worldlines of heavy charged test particles, and this is implemented in practice by the possibility of writing the Wilson line as a path integral over certain auxiliary degrees of freedom supported on the line defect itself. It is this path integral description of the Wilson line that will generalize in a straightforward way to instantons.

To express a Wilson line in a representation $R$ of the gauge group as a path integral over auxiliary degrees of freedom, one wants the Hilbert space arising from quantization of these degrees of freedom to be the representation $R$ itself. A canonical way to arrange for this in general is to use the Borel-Weil-Bott theorem (see e.g. page 1218-19 of \cite{qftias}), but any choice of degrees of freedom with Hilbert space $R$ will do (see e.g. Section 2.1.3 of \cite{tonggaugethy}). $S$-matrices are Wilson lines in the fundamental representation $\mathbb{C}^r$ of $U(r)$; of course we lose or gain no information by looking instead at Wilson lines in the exterior algebra
\begin{equation}
    R = \Lambda^\bullet \mathbb{C}^r = \bigoplus_{k = 0}^r \Lambda^k (\mathbb{C}^r).  
\end{equation}
Now $R$ arises naturally via the canonical quantization of a system of $r$ free fermions in one dimension, parameterized by the coordinate $y$, with Lagrangian 
\begin{equation}
    L = \sum_{\alpha = 1}^r \widetilde{\psi}_\alpha(y) \partial_y \psi^\alpha(y)
\end{equation}
and canonical anticommutation relations $\acomm{\psi^\alpha}{\widetilde{\psi}_\beta} = \delta\indices{^\alpha_\beta}$. This action corresponds to fermions with zero Hamiltonian. This action has a $U(r)$ symmetry acting on $\psi, \widetilde{\psi}$ by $\psi \to U \psi$, $\widetilde{\psi} \to \widetilde{\psi} U^\dagger$; the action of this symmetry on the Hilbert space gives rise to the map $\rho:U(r) \hookrightarrow U(\Lambda^\bullet \mathbb{C}^r)$. At the infinitesimal level it sends a matrix $A\indices{^\alpha_\beta}$ in the Lie algebra of $U(r)$ to the operator $\widetilde{\psi}_\alpha A\indices{^\alpha_\beta} \psi^\beta = \widetilde{\psi} A \psi$. The monopole $S$-matrix valued in $R$ as above may then be written in terms of fermion creation and annihilation operators as 
\begin{equation}
    \rho(S(u)) = \mathscr{P} \exp \Big( - \int_{- \infty}^{+ \infty} dy \widetilde{\psi}(A_y(y, u, \bar{u}) + i \phi(y, u, \bar{u}) \psi \Big).
\end{equation}
Interpreting $y$ as ``time'', at each fixed $u$ this is the (Euclidean) time evolution operator for a time-dependent Hamiltonian $H = \widetilde{\psi} (A_y +  i \phi) \psi$, so by textbook arguments we get a functional integral expression for its matrix elements in terms of a gauged version the fermionic system: 
\begin{equation}
    \bra{\alpha} \rho(S(u)) \ket{\beta} = \int_{\substack{\psi(-\infty) = \beta \\ \psi(+ \infty) = \alpha}} D\widetilde{\psi}(y) D\psi(y) \exp \Big( - \int_{-\infty}^{+ \infty} dy \widetilde{\psi}(y) (\partial_y + A_y(y, u, \bar{u}) + i \phi(y, u, \bar{u})) \psi(y) \Big). 
\end{equation}
The boundary conditions at $y \to \pm \infty$ must be chosen in accordance with the states $\alpha$ and $\beta$, as usual. As discussed in Appendix \ref{reviewscatter}, for the matrix elements to be honest holomorphic functions of $u$, the boundary conditions must be chosen to vary over the $u$-plane according to the $\overline{D}_u$ operator entering the monopole solution. This is suppressed here for conciseness. 

The fermionic description, while nice, may seem a bit arbitrary. The string theory arguments in Section \ref{section:string} give a more sophisticated rationale for this choice; for now we will simply use the fermions and deduce the consequences. 

\subsubsection{Scattering matrix: path integral description}
By these considerations, it is straightforward to write a definition, at the level of functional integrals, for the instanton scattering matrix. To generalize from monopoles on $\mathbb{R}^3$ to instantons on $\mathbb{R}^3 \times S^1$, we should make the replacement 
\begin{equation}
    \partial_y + A_y + i \phi \to \partial_y + i \partial_\theta + A_y + i A_\theta
\end{equation} 
and promote the fermionic action to its natural two-dimensional analog: the fermionic fields become functions of the $\theta$ coordinate $\psi(y, \theta)$, $\widetilde{\psi}(y, \theta)$, and we have the path integral 
\begin{equation}
    \int D \widetilde{\psi} D \psi \exp ( - \int_{-\infty}^{+ \infty} \int_0^{2\pi} dy d\theta \widetilde{\psi}(\partial_y + i \partial_\theta + A_y + i A_\theta) \psi)
\end{equation}
with it understood that $A_y$ and $A_\theta$ are evaluated at $u = \text{const}$. If we introduce the $\mathbb{C}^\times$-valued coordinate $z = e^{y + i \theta}$ then the action above reduces to the standard action of gauged chiral fermions
\begin{equation}
    S = \int_{\mathbb{C}^\times_u} \widetilde{\psi} \overline{D}_z \psi
\end{equation}
and we have coordinates $(u, z) \in \mathbb{C} \times \mathbb{C}^\times \simeq \mathbb{R}^3 \times S^1$; $\mathbb{C}^\times_u$ denotes a slice of constant $u$. A matrix element of the instanton scattering matrix is then expressed as a functional determinant
\begin{equation}
    \det \overline{D}_z
\end{equation}
and which matrix element is thus obtained is determined by making a choice of boundary condition at $z \to 0$ and $z \to \infty$. 

Note that if we had instead written the Wilson line in the monopole case via a functional integral over bosonic degrees of freedom valued in the flag variety (as in \cite{qftias}, pg. 1218-19), we would have instead found for the two-dimensional theory a chiral gauged WZW model at level $k = 1$. The relation between the chiral WZW model and multicomponent free fermions is the nonabelian bosonization of \cite{witten84}. 

The rest of this section will be devoted to explaining how to make this functional determinant less formal, and obtain precise definitions and formulas for the instanton scattering matrices. What makes this possible, and gives the theory a considerable amount of additional structure relative to its monopole analog, is that the relevant two-dimensional free fermion theory is \textit{holomorphic}. This phenomenon has no analog in one dimension, where the fermions depend on $y$ only. 

\subsubsection{Two more deformations}
To simplify the program outlined in the last section, it is convenient to make two additional deformations to the problem. First, instead of considering instantons on $\mathbb{R}^3 \times S^1 \simeq \mathbb{C} \times \mathbb{C}^\times$, we will consider instantons on $\mathbb{R}^4 \simeq \mathbb{C} \times \mathbb{C}$. This choice is made for technical convenience and will turn out to be relatively inessential; owing to the underlying conformal symmetry of two-dimensional chiral fermions, it is not difficult to pass back and forth between $\mathbb{C}$ and $\mathbb{C}^\times$ once the problem has been set up correctly. The coordinates on $\mathbb{C}^2$ are denoted $(z, w)$; the $w$-coordinate was referred to as $u$ in earlier sections of this paper. In the problem on $T^* \mathbb{C}$, we will choose the fermions to be supported on the curve given by $w = 0$. In principle, they can be supported on any holomorphic curve, and an alternative choice will be considered in Section \ref{section:newRmatrix} after motivating the construction in Section \ref{section:string}. 

The second deformation is much more drastic: instead of studying instantons on $\mathbb{C}^2$, we will work with \textit{noncommutative} instantons on the noncommutative cotangent bundle $T^*_\hbar \mathbb{C}$ \cite{Nekrasov_1998}. It would take us too far afield to describe noncommutative instantons in any depth here, see e.g. \cite{Nekrasov_2001} for an introduction; the essentials are described as follows. Noncommutative $\mathbb{R}^4$ is described by writing $z = x^1 + i x^2$, $w = x^3 + i x^4$ and declaring that the algebra of smooth functions on noncommutative $\mathbb{R}^4$ is (an appropriate completion of the algebra) generated by $x^i$ satisfying $\comm{x^i}{x^j} = i \theta^{ij}$, where $\theta^{ij}$ are related to the parameter $\hbar$ of this paper and its complex conjugate by insisting $\comm{w}{z} = \hbar$. 

There is a theory of vector bundles, connections and curvature on noncommutative $\mathbb{R}^4$, and so the instanton equations $F_A^+ = 0$ make sense and one may study their solutions. Just as one may think about an ordinary instanton in terms of a tuple of functions $A_i(x)$ solving first order nonlinear PDEs dictated by the anti-self-duality of the curvature of $A = A_i(x) dx^i$, one may think of a noncommutative instanton as a solution to formally the same equations \textit{in operators}, i.e. expressed via the $x^i$ solving $\comm{x^i}{x^j} = i \theta^{ij}$; \cite{Nekrasov_1998}, \cite{Nekrasov_2001} write them as operators in a Fock module for the creation/annihilation operator algebra generated by the $x^i$. The Wick symbols of these operators solve a nonlocal deformation of the usual instanton equations, in which all products are replaced by Moyal products depending on $\theta^{ij}$. 

The insight of \cite{Nekrasov_1998} was that, in spite of the fact that the noncommutative deformation is rather drastic and nonlocal, the ADHM construction \cite{adhm} of instantons goes through with minimal changes. As a complex manifold, the moduli space of noncommutative instantons of rank $r$ and charge $n$ is the variety $\widetilde{M}(n, r)$ from the introduction; the only impact of noncommutativity is shifting the level of the moment map. 

To describe the coupling of the fermion fields to the noncommutative instanton connection is, in principle, straightforward. This is because $A_{\bar{z}}(z, w, \bar{z}, \bar{w})$ is now a function of operators satisfying $\comm{w}{z} = \hbar$. The correct coupling of chiral fermions supported on the curve $w = 0$ to the noncommutative instanton is captured by the action 
\begin{equation}
    S = \int_{\mathbb{C}} \widetilde{\psi} \overline{D}_z \psi = \int d^2 z \widetilde{\psi}(z, \bar{z}) \Bigg(\overline{\partial}_z + A_{\bar{z}}(z, w, \bar{z}, \bar{w}) \eval_{w = \hbar \partial_z} \Bigg) \psi(z, \bar{z}).
\end{equation}
Very concretely, one can think of this then as a system of conventional free fermions in two dimensions, albeit with a nonstandard and nonlocal kinetic term. The functional determinant $\det \overline{D}_z$ is then at least in principle an object that could be computed in the analytic setting, though this will not be pursued in this paper. 

Instead, we will determine the resulting fermion path integral by more algebraic means. Our results will not depend on doing any noncommutative differential geometry; it will suffice to understand the instanton bundle $\mathscr{E}$ as a ``holomorphic'' bundle on $T^*_\hbar \mathbb{C}$, i.e. as a right $\mathscr{D}_\hbar$-module. We will get completely explicit formulas for $\det \overline{D}_z$, with a certain choice of normalization, as holomorphic functions of the ADHM data $(B_1, B_2, I, J)$. If one chooses the boundary condition at $z \to \infty$ to be the standard one in which the fermions are required to decay sufficiently rapidly, the formula is strikingly simple:
\begin{equation}
    \det \overline{D}_z = \det B_2.
\end{equation}
This can be thought of as an instance of ``reciprocity'' in ADHM/Nahm transforms, see \cite{corrigangoddard}.

\subsection{Determination of scattering matrix} \label{sect:Smatrix1stpass}
With the setup finally understood, we are in a position to give a precise definition and determination of the instanton scattering matrix. Our goal is to understand the functional determinant $\det \overline{D}_z$ of the Cauchy-Riemann operator in the background of a general noncommutative instanton, as we vary the choice of boundary condition at $z \to \infty$. At first sight, this seems rather hopeless but there is fortunately a rigorous algebra-geometric theory of such determinants based on the study of determinant line bundles on certain infinite Grassmannians. This is a classical subject in mathematical physics which we review in Appendix \ref{fermionreview}.

\subsubsection{Determinants and infinite Grassmannians}
Assuming familiarity with the discussion in Appendix \ref{fermionreview}, the setup to define and compute $\det \overline{D}_z$ is the following. We study the vector space $\mathbb{W}$ of solutions to the linear equation $\overline{D}_z \psi = 0$, and considering their $z \to \infty$ expansions gives an inclusion $\mathbb{W} \subset \mathbb{C}^r((z^{-1})) = \mathscr{H}$, which (if the operator $\overline{D}_z$ is reasonable) defines a point in Sato's Grassmannian $\text{Gr}(\mathscr{H})$\footnote{This Grassmannian was called $\text{Gr}^*(\mathscr{H})$ in Appendix \ref{fermionreview}. Also in Appendix \ref{fermionreview} the discussion was carried out for $r = 1$, but the higher rank generalization is straightforward.}. This point has Pl\"{u}cker coordinates given by a state $\ket{\mathbb{W}}$ in the projectivization of the fermionic Fock space $\mathbb{P}(\mathscr{F})$. As explained in \cite{witten} and reviewed in Appendix \ref{fermionreview}, this state $\ket{\mathbb{W}}$ is to be thought of as the state prepared by the path integral of chiral fermions on a surface with a single $S^1$ boundary component, in accordance with general principles of quantum field theory. 

Choices of boundary condition at $z \to \infty$ are labeled by dual states $\bra{\alpha}$, and numerical determinants with given boundary conditions are the overlaps $\braket{\alpha}{\mathbb{W}}$ (i.e. certain linear combinations of the Pl\"{u}cker coordinates); thus to understand the determinant it suffices to describe the state $\ket{\mathbb{W}}$. Perhaps surprisingly, even for the complicated $\overline{D}_z$ operator arising from noncommutative instantons, the space $\mathbb{W}$ and the map $\mathbb{W} \to \mathscr{H}$ may nonetheless be described quite explicitly, and this is the key insight allowing for a description of $\det \overline{D}_z$. 

\subsubsection{The bundle $\mathbb{W}$ of classical solutions}
By the above discussion, the determinant $\det \overline{D}_z$ becomes well-defined as soon as we have a description of the space $\mathbb{W}$ of solutions to $\overline{D}_z \psi = 0$ (and a suitably nice map $\mathbb{W} \hookrightarrow \mathscr{H}$, which we return to momentarily). In our situation, $\mathbb{W}$ varies as the fiber of a vector bundle (of infinite rank) over the instanton moduli space $\widetilde{M}(r)$, which by abuse of notation we also denote by $\mathbb{W}$. In other words, the construction we give below is to be read in the universal family over $\widetilde{M}(r)$; we will never have use for taking the fiber at any point.

In ordinary (commutative) complex geometry, the meaning of the equation $\overline{D}_z \psi = 0$ is clear. The instanton bundle $\mathscr{E}$ is in particular a holomorphic bundle over $T^* \mathbb{C}$, so its sheaf of sections (which we again denote by $\mathscr{E}$) is a locally free sheaf of $\mathscr{O}_{T^* \mathbb{C}}$-modules. The equation $\overline{D}_z \psi  = 0$ just says that $\psi$ is a holomorphic section of the restriction of the bundle $\mathscr{E} \eval_{w = 0}$ (we recall that $\psi$ is supported on the holomorphic curve $\{ w = 0 \} \subset T^* \mathbb{C}$), i.e. $\psi \in H^0(\mathbb{C}, \mathscr{E} \eval_{w = 0})$. 

Restriction of the bundle $\mathscr{E} \eval_{w = 0}$ is described algebraically by the tensor product $\mathscr{E} \otimes_{\mathscr{O}_{T^* \mathbb{C}}} \mathscr{O}_{\mathbb{C}}$ where $\mathscr{O}_{\mathbb{C}}$ denotes the structure sheaf of the zero section $\mathbb{C} \subset T^* \mathbb{C}$; since $\mathbb{C}$ is affine the space of sections $H^0(\mathbb{C}, \mathscr{E} \eval_{w = 0})$ is just the vector space underlying $\mathscr{E} \otimes_{\mathscr{O}_{T^* \mathbb{C}}} \mathscr{O}_{\mathbb{C}}$. In the noncommutative setting over $T^*_\hbar \mathbb{C}$, $\mathscr{E}$ is understood as a right $\mathscr{D}_\hbar$-module. We have the left $\mathscr{D}_\hbar$ module 
\begin{equation}
    \widehat{\mathscr{O}}_{\mathbb{C}}^\ell := \mathscr{D}_\hbar / \mathscr{D}_\hbar \cdot w
\end{equation}
and it is clear that we may take the following, for algebra-geometric purposes, as a definition: 
\begin{equation} \label{Wtensorprod}
    \mathbb{W} := \mathscr{E} \otimes_{\mathscr{D}_\hbar} \widehat{\mathscr{O}}^\ell_{\mathbb{C}}.
\end{equation}
This is just the obvious analog of the tensor product description of the space of solutions to $\overline{D}_z \psi = 0$ in the commutative setting; the tensor product is over $\mathscr{D}_\hbar$ regarded as a bimodule over itself. It is interesting to give an alternative description of $\mathbb{W}$ using only right $\mathscr{D}_\hbar$-modules. The right $\mathscr{D}_\hbar$ module version of $\widehat{\mathscr{O}}^\ell_{\mathbb{C}}$ has, of course, a free resolution 
\begin{equation}
\begin{tikzcd}
    0 \arrow[r] & \mathscr{D}_\hbar \arrow[r, " w \cdot - " ] & \mathscr{D}_\hbar \arrow[r] & \widehat{\mathscr{O}}_{\mathbb{C}} \arrow[r] &  0
\end{tikzcd}
\end{equation}
where the first nonzero map is left multiplication by $w$. Using this free resolution to compute $\text{Ext}^\bullet$, we have the following elementary observation: 
\begin{equation}
    \mathbb{W} \simeq \text{Ext}^1_{\mathscr{D}_\hbar}(\widehat{\mathscr{O}}_{\mathbb{C}}, \mathscr{E}). 
\end{equation}
In this description, it is clear that $\mathbb{W}$ is indeed a rather natural vector bundle over $\widetilde{M}(r)$ (in Proposition \ref{extgroups} it is shown that this $\text{Ext}^\bullet$ is concentrated in degree 1). Readers familiar with the $B$-model of topological string theory may recognize that $\mathbb{W}$ has a natural interpretation as a space of open string states there; this is no coincidence as will be explained in Section \ref{section:string} and is part of the general connection between noncommutative gauge theory and open string field theory. 

\subsubsection{$\mathbb{W}$ as a map to $\textnormal{Gr}(\mathscr{H})$}
Now we want to show that $\mathbb{W}$ gives rise to a map from the instanton moduli space to Sato's Grassmannian $\text{Gr}(\mathscr{H})$ of subspaces in 
\begin{equation}
    \mathscr{H} = W(( z^{-1})).
\end{equation}
$W$ is the $r$-dimensional framing space entering the ADHM description of $\widetilde{M}( r)$. To show this, it suffices to construct an inclusion of vector bundles 
\begin{equation}
    \mathbb{W} \hookrightarrow \mathscr{W}((z^{-1}))
\end{equation}
and prove that the induced projection $\mathbb{W} \to \mathscr{H}_+ \otimes \mathscr{O}_{\widetilde{M}(r)}$ has finite rank kernel and cokernel (here $\mathscr{H}_+$ denotes $W[z]$). For readers following along with the differential-geometric/gauge theoretic description, that $\mathbb{W}$ gives rise to a map to $\text{Gr}(\mathscr{H})$ is essentially an algebraic incarnation of an ellipticity property of the Dirac operator $\overline{D}_z$ (which is rather nontrivial in the noncommutative setting), see \cite{witten}. 

To write down the inclusion we invoke Proposition \ref{extgroups}, which is a straightforward analysis of the ADHM description of $\mathscr{E}$. It describes $\mathbb{W} \simeq \text{Ext}^1_{\mathscr{D}_\hbar}(\widehat{\mathscr{O}}_{\mathbb{C}}, \mathscr{E})$ as the middle cohomology of the complex
\begin{equation} 
\begin{tikzcd}
    \mathscr{V}[z] \arrow[r, "\alpha"] & \mathscr{V}[z] \oplus \mathscr{V}[z] \oplus \mathscr{W}[z] \arrow[r, "\beta"] & \mathscr{V}[z]
\end{tikzcd}
\end{equation}
with maps (ADHM matrices are interpreted as the corresponding tautological sections)
\begin{equation}
\begin{split}
    \alpha & = \begin{pmatrix} B_1 - z \\ B_2 - \hbar \partial_z \\ J \end{pmatrix} \\
    \beta & = \begin{pmatrix} -B_2 + \hbar \partial_z && B_1 - z && I \end{pmatrix}.
\end{split}
\end{equation}
Then denoting any local section of $\mathbb{W}$ by a triple $(\psi_+(z), \psi_-(z), \xi(z))$ in $\ker \beta / \text{im} \, \alpha$, it is clear there is essentially only one well-defined natural map $\mathbb{W} \to \mathscr{W}((z^{-1}))$:
\begin{equation} \label{Winclude}
    (\psi_+(z), \psi_-(z), \xi(z)) \mapsto f(z) = J \frac{1}{z - B_1} \psi_+(z) + \xi(z) = \sum_{n \geq 0} J B_1^n \psi_+(z) z^{- n - 1} + \xi(z).
\end{equation}

\begin{prop} \label{maptoGr}
The map $\textnormal{Ext}^1_{\mathscr{D}_\hbar}(\widehat{\mathscr{O}}_{\mathbb{C}}, \mathscr{E}) \simeq \mathbb{W} \to \mathscr{W}((z^{-1}))$ written in \eqref{Winclude} is an inclusion of vector bundles, moreover we have (via Proposition \ref{explicitext}) isomorphisms
\begin{equation}
\begin{split}
    \ker(\mathbb{W} \to \mathscr{H}_+ \otimes \mathscr{O}_{\widetilde{M}(r)}) & \simeq \ker(B_2) \\
    \textnormal{coker}(\mathbb{W} \to \mathscr{H}_+ \otimes \mathscr{O}_{\widetilde{M}(r)}) & \simeq \textnormal{coker}(B_2)
\end{split}
\end{equation}
of sheaves over $\widetilde{M}(r)$, where $B_2: \mathscr{V} \to \mathscr{V}$ is the tautological section corresponding to the ADHM matrix. 
\end{prop}
Note that a corollary of this proposition is that $\mathbb{W}$ defines a map from $\widetilde{M}(n, r)$ to the index zero connected component of $\text{Gr}(\mathscr{H})$. For the proof, we will need various technical results from Appendix \ref{adhmdetails}. 

In particular we need Proposition \ref{explicitext}, which says informally that we may pick a gauge fixing condition for the quotient by $\text{im} \,  \alpha$ by imposing $\psi_+ = \text{constant}$, and in such a gauge $\psi_-(z)$ is uniquely determined by the equation $\beta(\psi_+, \psi_-(z), \xi(z)) = 0$. 

\begin{proof}
The proof will take a few steps. Throughout, we will use the description of $\mathbb{W}$ from Proposition \ref{explicitext}. In this description, the map $\mathbb{W} \to \mathscr{W}((z^{-1}))$ is (notations as in Proposition \ref{explicitext}) 
\begin{equation}
    (\xi(z), \psi_+) \mapsto f(z) = J \frac{1}{z - B_1} \psi_+ + \xi(z) = \sum_{n \geq 0} J B_1^n \psi_+ z^{- n - 1} + \xi(z).
\end{equation}
To show that this map is an inclusion of vector bundles, we need to show that the induced map in each fiber is injective.

To see the injectivity, note that $f(z) = 0$ if and only if $\xi(z) = 0$ and $JB_1^n \psi_+ = 0$ for all $n \geq 0$. Since we have $\xi(B_1) I = B_2 \psi_+$, $\xi(z) = 0$ implies $B_2 \psi_+ = 0$. 

Let $S = \text{Span}(B_1^n \psi_+)_{n \geq 0} \subseteq V$; clearly $S \subseteq \ker J$ and $B_1$ preserves $S$. $B_2$ also preserves $S$, because 
\begin{equation}
\begin{split}
    B_2 B_1^n \psi_+ & = \sum_{k = 0}^{n - 1} B_1^k \comm{B_2}{B_1} B_1^{n - k - 1} \psi_+ + B_1^n B_2 \psi_+ \\
    & = - \hbar n B_1^{n - 1} \psi_+ + \sum_{k = 0}^{n - 1} B_1^k IJ B_1^{n - k - 1} \psi_+ = -\hbar n B_1^{n - 1}\psi_+.
\end{split}
\end{equation}
By Lemma \ref{stabandcostab} we must have $S = 0$, therefore $\psi_+ = 0$ establishing the injectivity. 

In the description of Proposition \ref{explicitext}, the map $\mathbb{W} \to \mathscr{H}_+$ is $(\xi(z), \psi_+) \mapsto \xi(z)$. Thus $\ker(\mathbb{W} \to \mathscr{W}[z]) \simeq \ker(B_2)$, because the pair $(\xi(z), \psi_+)$ must solve $\xi(B_1) I = B_2 \psi_+$. 

Finally, observe that we have a morphism of sheaves $\text{coker}(\mathbb{W} \to \mathscr{W}[z]) \to \text{coker}(B_2)$ induced by the assignment $$W[z] \ni p(z) \mapsto p(B_1) I \in \text{coker}(B_2),$$ and the induced morphism is well-defined and injective by Proposition \ref{explicitext}. To see surjectivity, invoke Lemma \ref{stabilitybasis} to conclude any $\chi \in \text{coker}(B_2)$ may be written as (notations as in Lemma \ref{stabilitybasis})
\begin{equation}
    \chi = \sum_{n, i} c_{n, i} B_1^n I(e_i) = p(B_1) I 
\end{equation}
for $p(z) = \sum_{n, i} c_{n, i} e_i z^n$. 
\end{proof}
Denote the morphism constructed in Proposition \ref{maptoGr} as 
\begin{equation}
    F: \widetilde{M}(r) \to \text{Gr}(\mathscr{H}). 
\end{equation}
Composing with the Pl\"{u}cker embedding we get a map 
\begin{equation}
    \widetilde{M}(r) \to \mathbb{P}(\mathscr{F})
\end{equation}
denoted by $\ket{\mathbb{W}}$. The components of $\ket{\mathbb{W}}$ in the standard basis of $\mathscr{F}$ are the Pl\"{u}cker coordinates of $\mathbb{W}$, and vary as sections of the line bundle $F^* \text{DET}^*$ over the instanton moduli space, where $\text{DET}^* \to \text{Gr}(\mathscr{H})$ is the tautological determinant line. 
\begin{prop} \label{vaccomponent}
    $\ket{\mathbb{W}}$ lifts to a map $\widetilde{M}(r) \to \mathscr{F}$, and there is a unique choice of lift satisfying $\braket{0}{\mathbb{W}} = \det(B_2)$. 
\end{prop}

\begin{proof}
    From definitions in Appendix \ref{fermionreview}, $\braket{0}{\mathbb{W}}$ is a section of $F^*\text{DET}^*$ that vanishes precisely where $\ker(\mathbb{W} \to \mathscr{H}_+ \otimes \mathscr{O}_{\widetilde{M}(r)}) \neq 0$. By Proposition \ref{maptoGr}, this is precisely the divisor of the global function $\det(B_2)$, showing that $F^*\text{DET}^*$ is trivial as a line bundle over $\widetilde{M}(r)$, so $\ket{\mathbb{W}}$ lifts to a map to $\mathscr{F}$. Two choices of lift differ by multiplication by a nowhere-vanishing function on $\widetilde{M}(r)$; clearly we may choose the normalization of $\ket{\mathbb{W}}$ so that $\braket{0}{\mathbb{W}} = \det(B_2)$.
\end{proof}

The map $F$ above should be viewed as the assignment of an instanton to its scattering data; the ``scattering data'' is the possible asymptotics at $z \to \infty$ of solutions to $\overline{D}_z \psi = 0$. Proposition \ref{vaccomponent} just expresses a certain resonance effect: the coefficient $\braket{0}{\mathbb{W}}$ vanishes precisely when there exists a solution to $\overline{D}_z \psi = 0$ that goes to $0$ at $z \to \infty$, i.e. a bound state wavefunction. Unsurprisingly, this happens at $\det B_2 = 0$ because the eigenvalues of $B_2$ correspond to the positions of the noncommutative instantons in the $w$-direction, and our fermions are supported on the curve $\{ w = 0 \}$. 

\subsubsection{The $\tau$-function of $\ket{\mathbb{W}}$}
In the previous section we have given a geometric definition of the state $\ket{\mathbb{W}}$ in the fermionic Fock space assigned to an arbitrary noncommutative instanton. In this section and the next, we will give two algebraic characterizations of this state. For simplicity, attention will be restricted to the rank $r = 1$ case; in this case $\widetilde{M}(n, 1) = \text{Hilb}_n(T^*_\hbar \mathbb{C})$ is the Hilbert scheme of points on the noncommutative plane. 

As recalled in Appendix \ref{fermionreview}, any state $\ket{\mathbb{W}}$ capturing the Pl\"{u}cker coordinates of some subspace in $\mathscr{H}$ is characterized by its $\tau$-function, a function of infinitely many variables defined by 
\begin{equation}
    \tau_{\mathbb{W}}(t) := \bra{0} g(t) \ket{\mathbb{W}}
\end{equation}
with 
\begin{equation}
    g(t) = \exp(\sum_{k = 1}^\infty t_k \alpha_k)
\end{equation}
and notations as in Appendix \ref{fermionreview}; all but finitely many of the variables $t_k$ are zero. $g(t)$ defines an action of (half of) the loop group of $GL_1$ (called $\Gamma_+$ in \cite{segal-wilson}) on $\text{Gr}(\mathscr{H})$ via the Pl\"{u}cker embedding. Moreover $\text{DET}^*$ becomes in a natural way a $\Gamma_+$-equivariant line bundle. Composing the map $F$ above with the action map by $g(t)$, we get a new map from $\widetilde{M}(1) \to \text{Gr}(\mathscr{H})$, and denote the corresponding vector bundle over $\widetilde{M}(1)$ by $g(t) \cdot \mathbb{W}$. By construction, $g(t) \ket{\mathbb{W}} = \ket{g(t) \cdot \mathbb{W}}$. 

\begin{prop} \label{taufnformula}
    For $\ket{\mathbb{W}}$ constructed in Propositions \ref{maptoGr} and \ref{vaccomponent} above, we have
    \begin{equation}
        \bra{0} g(t) \ket{\mathbb{W}} = \det\Big( B_2 - \hbar \sum_{n = 1}^\infty nt_n B_1^{n - 1} \Big) . 
    \end{equation}
\end{prop}

\begin{proof}
    By the definition of the loop group action on $\text{Gr}(\mathscr{H})$ from Appendix \ref{fermionreview}, the vector bundle $g(t) \cdot \mathbb{W}$ over $\widetilde{M}(1)$ is given by 
    \begin{equation}
        g(t) \cdot \mathbb{W} \simeq \text{Ext}^1_{\mathscr{D}_\hbar}(\widehat{\mathscr{O}}_C(t), \mathscr{E})
    \end{equation}
    where
    \begin{equation}
        \widehat{\mathscr{O}}_C(t) = \mathscr{D}_\hbar/P(z, w; t)\mathscr{D}_\hbar
    \end{equation}
    and $P(z, w; t) = w - \hbar \sum_{n = 1}^\infty nt_n z^{n - 1}$. Adapting the proof of Proposition \ref{explicitext}, we have an embedding of vector bundles 
    \begin{equation}
       \text{Ext}^1_{\mathscr{D}_\hbar}(\widehat{\mathscr{O}}_C(t), \mathscr{E})  \hookrightarrow (\mathscr{O}_{\widetilde{M}(1)} \otimes \mathbb{C}[z]) \oplus \mathscr{V}
    \end{equation}
    as the locus of $(\xi(z), \psi_+)$ such that $\xi(B_1) I = (B_2 - \hbar \sum_{n = 1}^\infty nt_n B_1^{n - 1}) \psi_+$. Arguing as in the proof of Proposition \ref{vaccomponent}, it follows that that $\bra{0}g(t)\ket{\mathbb{W}} = \braket{0}{g(t) \cdot \mathbb{W}}$ is proportional to $\det(B_2 - \hbar \sum_{n = 1}^\infty n t_n B_1^{n - 1})$, and the proportionality factor is $1$ using the normalization from Proposition \ref{vaccomponent}.
\end{proof}

While the main focus of this paper is on the geometric approach, we give a more combinatorial computation of the $\tau$-function in Appendix \ref{adhmdetails}, which some readers may prefer. This $\tau$-function appeared much earlier in a landmark paper of Wilson \cite{wilson}, see also \cite{bgk02}, \cite{benzvi03}, though our perspective here is rather different.

The loop group action on $\text{Gr}(\mathscr{H})$ corresponds to the flows of the KP hierarchy, see e.g. \cite{segal-wilson}. The main idea entering the above proof is that the KP flows can be interpreted as generating a geometric deformation of the curve $C$ on which the fermions are supported; this was a central insight of \cite{Aganagic_2005}, \cite{DHSV}, \cite{DHS}. 

\subsubsection{Factorization and Miura operators}
The geometric construction above characterizes the scattering data of the noncommutative instantons via the map $F$, or equivalently the family of states $\ket{\mathbb{W}}$ in the fermionic Fock space. In comparison with the monopole scattering matrices, this description may seem rather abstract; now we will show that in fact the state $\ket{\mathbb{W}}$ has a very explicit representation-theoretic description that parallels many features of its monopole analog. 

The main idea is to characterize the state $\ket{\mathbb{W}}$ via the Ward identities of the free fermion conformal field theory, which are reviewed in Lemma \ref{cliffward}. Representation-theoretically, this means we take advantage of the fact that $\ket{\mathbb{W}}$ is a vector in an irreducible Clifford module annihilated by a maximal isotropic subalgebra, which characterizes it up to a scalar multiple. 

To set this up, it is useful to recall some preliminaries. The instanton moduli space has a distinguished open set $U \subset \widetilde{M}(n, 1)$ on which 
\begin{equation}
    B_1 = \text{diag}(z_1, \dots, z_n)
\end{equation}
with $z_i \neq z_j$. We may assume that $|z_1| < |z_2 | < \dots < |z_n|$. Set $p_i := (B_2)_{ii}$, the diagonal components of $B_2$ in the eigenbasis of $B_1$. Then $\comm{B_1}{B_2} + IJ = \hbar$ implies that we may assume $I = (1, \dots, 1)^T$, $J = \hbar(1, \dots, 1)$, and 
\begin{equation}
    (B_2)_{ij} = p_i \delta_{ij} - \frac{\hbar}{z_i - z_j}(1 - \delta_{ij})
\end{equation}
where $\delta_{ij}$ is the Kronecker delta symbol. On this open subset, the symplectic form $\omega_{\mathbb{C}}$ restricts to 
\begin{equation}
    \sum_i dp_i \wedge dz_i. 
\end{equation}
The inclusion map $\iota: U \hookrightarrow \widetilde{M}(n, 1)$ gives rise via pullback to a ring embedding $$\iota^*: H^0(\widetilde{M}(n, 1), \mathscr{O}_{\widetilde{M}(n, 1)}) \hookrightarrow \mathbb{C}[z_1, \dots, z_n, p_1, \dots, p_n, (z_i - z_j)^{-1}]$$
which is Poisson if we give $(z_i, p_i)$ the Poisson bracket $\acomm{p_i}{z_j} = \delta_{ij}$.

\begin{lemma} \label{WinU}
    Over $U$, $\iota^* \mathbb{W} \hookrightarrow \mathscr{O}_U \otimes \mathbb{C}((z^{-1}))$ is the subbundle of $z \to \infty$ expansions of rational functions with at most simple poles at $z = z_i$ such that 
    \begin{equation}
        \oint_{\gamma_i} \frac{dz}{2\pi i } \Big( \frac{\hbar}{z - z_i} - p_i \Big) f(z) = 0
    \end{equation}
    where $\gamma_i$ is a small loop enclosing $z_i$. 
\end{lemma}

\begin{proof}
    Write \eqref{Winclude} in the eigenbasis of $B_1$; the residue condition is equivalent to \eqref{kerbetaeq} upon using the explicit formula for $B_2$ above. 
\end{proof}

Recall from Appendix \ref{fermionreview} that the bundle $\mathbb{W}$ has a dual $\widetilde{\mathbb{W}}$ induced by the residue pairing on $\mathbb{C}((z^{-1}))$. $\iota^* \widetilde{\mathbb{W}} \hookrightarrow \mathscr{O}_U \otimes \mathbb{C}((z^{-1}))$ is characterized as the subbundle of $z \to \infty$ expansions of rational functions with simple poles at $z = z_i$ satisfying the same residue condition, but with $p_i \mapsto -p_i$. 

Finally recall from Appendix \ref{fermionreview} the chiral boson dual to $\psi(z)$, $\widetilde{\psi}(z)$ under bosonization: 
\begin{equation}
    \partial \varphi(z) = :\psi(z) \widetilde{\psi}(z): = \sum_{n \in \mathbb{Z}} \frac{\alpha_n}{z^{n + 1}}.
\end{equation}
Now we are ready to state one of our main results. Let $\mathscr{F}_0$ denote the charge zero component of the fermionic Fock space $\mathscr{F}$, which we recall is isomorphic to a (completed) Fock module over the Heisenberg algebra generated by the $\alpha_n$, in which $\alpha_0$ acts by zero. 

\begin{theorem} \label{thm:instantonfactor}
    Let $\ket{\mathbb{W}}$ denote the morphism from $\widetilde{M}(1) \to \mathscr{F}_0$ constructed in Proposition \ref{maptoGr} and \ref{vaccomponent}. Over the open set $U$, it factors as 
    \begin{equation} \label{miurafactor}
        \iota^* \ket{\mathbb{W}} = (p_n - \hbar \partial \varphi(z_n)) \dots (p_1 - \hbar \partial \varphi(z_1)) \ket{0}
    \end{equation}
    where $\ket{0} \in \mathscr{F}_0$ is the vacuum. 
\end{theorem}

Here by ``factors'' we mean factors as a product of matrices acting in the Fock module. This should be compared directly to the factorization of monopole scattering matrices \eqref{eq:monopoleSfactor} in the distinguished coordinate chart of the monopole moduli space. 

\begin{proof}
The idea will be to use Lemma \ref{WinU} to show that the operators $Q_f$, $\widetilde{Q}_g$ from Lemma \ref{cliffward} annihilate the right hand side of \eqref{miurafactor}; this shows that the left and right hand side must agree up to normalization, which we pin down using Proposition \ref{vaccomponent}. 

Indeed, if $f(z)$ is a section of $\iota^* \widetilde{\mathbb{W}}$ we have, by the standard contour pulling arguments from conformal field theory 
\begin{equation}
\begin{split}
    Q_f(\text{RHS of \eqref{miurafactor}}) & = \oint_{\gamma_\infty} \frac{dz}{2\pi i}f(z) \psi(z) (p_n - \hbar \partial \varphi(z_n)) \dots (p_1 - \hbar \partial \varphi(z_1)) \ket{0} \\
    & = \sum_{i = 1}^n (p_n - \hbar \partial \varphi(z_n)) \dots \oint_{\gamma_i} \frac{dz}{2\pi i} f(z) \psi(z) (p_i - \hbar \partial \varphi(z_i)) \dots (p_1 - \hbar \partial \varphi(z_1)) \ket{0}.
\end{split}
\end{equation}
Now using the operator product expansion 
\begin{equation}
    \psi(z) \widetilde{\psi}(w) = \frac{1}{z - w} + \text{reg}
\end{equation}
we find
\begin{equation}
    \psi(z)(p_i - \hbar :\psi(z_i) \widetilde{\psi}(z_i):) = p_i \psi(z_i) +  \frac{\hbar}{z - z_i} \psi(z_i) + O(z - z_i)
\end{equation}
(the absence of a constant term in the regular part of the second term can be seen e.g. by a direct calculation with mode expansions). From the absence of a constant term in this OPE, we see 
\begin{equation}
    \oint_{\gamma_i} \frac{dz}{2\pi i } f(z) \psi(z) (p_i - \hbar \partial \varphi(z_i)) = \psi(z_i) \oint_{\gamma_i} \frac{dz}{2\pi i} f(z) \Big( \frac{\hbar}{z - z_i} + p _i \Big) = 0
\end{equation}
since $f(z)$ is a section of $\iota^* \widetilde{\mathbb{W}}$. The computation with $\widetilde{Q}_g$ is nearly identical and left to the reader. 

By Lemma \ref{cliffward}, the left and right hand side of \eqref{miurafactor} agree up to a nowhere vanishing function on $U$, i.e. 
\begin{equation}
    (p_n - \hbar \partial \varphi(z_n)) \dots (p_1 - \hbar \partial \varphi(z_1)) \ket{0} = f_n(z_i, p_i) \iota^* \ket{\mathbb{W}}
\end{equation}
where $f_n \in H^0(U, \mathscr{O}_U)= \mathbb{C}[z_1, \dots, z_n, p_1, \dots, p_n, (z_i - z_j)^{-1}]$, nonzero at all closed points of $U$. Contract both sides with $\bra{0}$ and apply Proposition \ref{vaccomponent} to conclude 
\begin{equation}
    f_n = \frac{\bra{0} (p_n - \hbar \partial \varphi(z_n)) \dots (p_1 - \hbar \partial \varphi(z_1)) \ket{0}}{\iota^* \det (B_2)}.
\end{equation}
Trivially, $f_1 = 1$. Noting that $f_n$ must depend polynomially on $p_n$ and considering the limit $p_n \to \infty$ we deduce inductively $f_n = f_{n - 1} = \dots = f_1 = 1$. 
\end{proof}
The operator $p - \hbar \partial \varphi(z)$, where $z$ and $p$ are canonically conjugate variables with Poisson bracket $\acomm{p}{z} = 1$, is called (the Poisson limit of a) ``Miura operator'', though from the point of view of this paper a better name for it would be ``instanton scattering matrix''. Based on the discussion of monopole scattering matrices, one might expect that it should play the role of an $R$-matrix of $\mathsf{Y}(\widehat{\mathfrak{gl}}_1)$ (we have $\widehat{\mathfrak{gl}}_1$ since we are in the $r = 1$ case). At this point, we make contact with several other authors: it has been proposed, in various incarnations, in \cite{miuraCSfeynman2}, \cite{gaiottorapcak2020}, \cite{gaiottorapcakzhou24}, \cite{haouzijeong}, \cite{miuraCSfeynman1}, \cite{Prochazka_2019}, \cite{zenkevich24}, that Miura operators (and their $q$-deformed variants) satisfy properties analogous to those of $R$-matrices for affine Yangians (and their multiplicative analogs, the quantum toroidal algebras). Our framework provides a natural explanation for this, and places this observation in a wider context by relating it with the $R$-matrices arising from the moduli spaces of monopoles. 

A corollary of Proposition \ref{taufnformula} and Theorem \ref{thm:instantonfactor} is the nice formula 
\begin{equation} \label{formula:ISMcoefficients}
    \bra{0} e^{\sum_{m = 1}^\infty t_m \alpha_m} (p_n - \hbar \partial \varphi(z_n)) \dots (p_1 - \hbar \partial \varphi(z_1)) \ket{0} = \iota^* \det\Big(B_2 - \hbar \sum_{m = 1}^\infty m t_m B_1^{m - 1} \Big) 
\end{equation}
which says that, remarkably, all components of the vector $(p_n - \hbar \partial \varphi(z_n)) \dots (p_1 - \hbar \partial \varphi(z_1)) \ket{0}$ extend from the coordinate patch $U$ to global functions on $\widetilde{M}(n, 1)$. The $q$-deformed version of this formula was first communicated to the author by Y. Zenkevich (see also \cite{zenkevich24}); it appears to also be related to Feynman's work on Calogero-type models \cite{Polychronakos_2019}. Note that our proof of this formula involves no direct computation of either side.

\subsubsection{RTT Relation}
The fundamental property the Miura operator must satisfy if it is to be regarded as a kind of $R$-matrix is the $RTT$ relation with respect to the $R$-matrix between a pair of Fock modules of $\mathsf{Y}(\widehat{\mathfrak{gl}}_1)$ constructed in \cite{mo}. To match conventions with \cite{mo}, \textit{in this section only} we will rescale the Heisenberg generators $\alpha_n$ in such a way that the Heisenberg commutation relations are 
\begin{equation}
    \comm{\alpha_m}{\alpha_n} = -\frac{1}{\varepsilon_2 \varepsilon_3} m \delta_{m + n, 0}
\end{equation}
where $\varepsilon_2, \varepsilon_3$ are complex numbers. We consider a quantized version of the Miura operator, replacing $p \to \varepsilon_1 \partial_z$ with $\varepsilon_1 + \varepsilon_2 + \varepsilon_3 = 0$: 
\begin{equation}
    \varepsilon_1 \partial_z - \varepsilon_2 \varepsilon_3 \partial \varphi(z)
\end{equation}
in direct parallel with the monopole scattering matrices. This is an operator in the Fock module for the Heisenberg algebra generated by $\alpha_n$, with coefficients in differential operators in $z$. Denote by $\mathscr{F}_0(a)$ the Fock module in which the zero mode $\alpha_0$ acts by $a$. 

Recall that the matrix $R_{\mathsf{MO}}(a_1 - a_2)$ (which we take to be associated to instanton moduli in the $\varepsilon_2, \varepsilon_3$ directions) acts in a tensor product $\mathscr{F}^{(1)}(a_1) \otimes \mathscr{F}^{(2)}(a_2)$ of two Fock spaces and can be shown to satisfy the following two properties (see \cite{mo}, chapters 12-14): 
\begin{enumerate}
    \item It commutes with the diagonal Heisenberg algebra generated by the modes of $\partial \varphi^{(1)}(z) + \partial \varphi^{(2)}(z)$. 
    \item It intertwines the off-diagonal Virasoro algebras generated by the modes of the operators
    \begin{equation}
        T_{\pm}(z) = \varepsilon_2 \varepsilon_3 :(\partial \varphi^{(1)}(z) - \partial \varphi^{(2)}(z))^2: \pm 2(\varepsilon_2 + \varepsilon_3) \partial^2 (\varphi^{(1)} - \varphi^{(2)})
    \end{equation}
i.e. $R_{\mathsf{MO}} T_+(z) R_{\mathsf{MO}}^{-1} = T_-(z)$. 
\end{enumerate}
It is elementary to check that this these two properties imply  

\begin{lemma} (see also \cite{Prochazka_2019}) \label{lemma:RTTaffine} 
\begin{equation} \label{eq:RTTaffine}
\begin{split}
    R_{\mathsf{MO}} (\varepsilon_1 \partial_z - \varepsilon_2 \varepsilon_3 \partial \varphi^{(1)}(z))(\varepsilon_1 \partial_z - \varepsilon_2 \varepsilon_3 \partial \varphi^{(2)}(z)) R_{\mathsf{MO}}^{-1} & = \\ 
     & (\varepsilon_1 \partial_z - \varepsilon_2 \varepsilon_3 \partial \varphi^{(2)}(z))(\varepsilon_1 \partial_z - \varepsilon_2 \varepsilon_3 \partial \varphi^{(1)}(z)).
\end{split}
\end{equation}
\end{lemma}

\begin{proof} 
Indeed, note that:
\begin{equation} \label{eq:expandmiura}
\begin{split}
    (\varepsilon_1 \partial_z - \varepsilon_2 \varepsilon_3\partial \varphi^{(1)}(z))(\varepsilon_1 \partial_z - \varepsilon_2 \varepsilon_3 \partial \varphi^{(2)}(z))  & =
    \varepsilon_1^2 \partial_z^2 - \varepsilon_1 \varepsilon_2 \varepsilon_3 (\partial \varphi^{(1)}(z) + \partial \varphi^{(2)}(z)) \partial_z \\
    & - \varepsilon_1 \varepsilon_2 \varepsilon_3 \partial^2 \varphi^{(2)}(z) + (\varepsilon_2 \varepsilon_3)^2 \partial \varphi^{(1)}(z) \partial \varphi^{(2)}(z)
\end{split}
\end{equation}
so using 
\begin{equation}
\begin{split}
    \partial \varphi^{(1)}(z) \partial \varphi^{(2)}(z) & = \frac{1}{4} \Big(:(\partial \varphi^{(1)} + \partial \varphi^{(2)})^2: - :(\partial \varphi^{(1)} - \partial \varphi^{(2)})^2: \Big) \\
    \partial^2 \varphi^{(2)}(z) & = \frac{\partial}{2}( \partial \varphi^{(1)} + \partial \varphi^{(2)}) - \frac{\partial}{2} (\partial \varphi^{(1)} - \partial \varphi^{(2)})
\end{split}
\end{equation}
we see that the right hand side of \eqref{eq:expandmiura} takes the form 
\begin{equation}
    (\text{commutes with $R_{\mathsf{MO}}$}) \, \, \, - \frac{\varepsilon_2 \varepsilon_3}{4} \Big(  \varepsilon_2 \varepsilon_3 :(\partial \varphi^{(1)} - \partial \varphi^{(2)})^2: - 2\varepsilon_1 \partial^2(\varphi^{(1)} - \varphi^{(2)}) \Big). 
\end{equation}
Recalling that $\varepsilon_1 = -(\varepsilon_2 + \varepsilon_3)$, the quantity in parenthesis is $T_+(z)$. Thus conjugation by $R_{\mathsf{MO}}$ just flips the sign on the last term. But this precisely has the effect of switching $\partial^2 \varphi^{(2)}$ to $\partial^2 \varphi^{(1)}$ in \eqref{eq:expandmiura}. Noting that all other terms are symmetric in $\varphi^{(1)}, \varphi^{(2)}$, we conclude \eqref{eq:RTTaffine}. 
\end{proof}

Multiplying \eqref{eq:RTTaffine} on the right by $R_{\mathsf{MO}}$ this is just the $RTT$ relation. The products of Miura operators entering the instanton scattering matrices therefore also satisfy $RTT$, by the standard argument. In the semiclassical limit, this directly parallels \eqref{eq:monopoleRTT} enjoyed by monopole scattering matrices, see also discussion in Section 4 of \cite{tamagnishiftop} for the quantized case. 

\subsubsection{$T^*_\hbar \mathbb{C}$ vs. $T^*_\hbar \mathbb{C}^\times$}
The original motivation was to study instantons on $T^*_\hbar \mathbb{C}^\times$, but we studied instantons on $T^*_\hbar \mathbb{C}$ for convenience. As our analysis is geometric, it can of course be repeated for $T^*_\hbar \mathbb{C}^\times$. Instanton bundles on $T^*_\hbar \mathbb{C}^\times$ also admit an ADHM-like construction, so the technical details are largely the same. As such, we will be brief and summarize the results (actually the main result can be deduced directly from \eqref{formula:ISMcoefficients}, as will be explained). 

There is one key difference, though: the free fermion path integral for fermions wrapping the zero section of $T^*_\hbar \mathbb{C}^\times$ naturally produces an $\textit{operator}$ acting on the fermionic Fock space $\mathscr{F}$, rather than a state therein. We can understand this, however, by taking a shortcut and using conformal symmetry of the free fermions to directly map the problem on $\mathbb{C}$ to the problem on $\mathbb{C}^\times$. In essence, the state we produced by studying fermions on $\mathbb{C}$ was already observed to be of the form of some operator acting on the vacuum; this operator is what is canonically assigned to the fermions on $\mathbb{C}^\times$. We once again will remain in the rank $r = 1$ case for convenience. 

Recall that the Hilbert scheme of points on $T^*_\hbar \mathbb{C}^\times$ admits the following description ($V$ is $n$-dimensional as usual, and $\hbar \neq 0$): 
\begin{equation}
    \text{Hilb}_n(T^*_\hbar \mathbb{C}^\times) = \{ (g, B, I, J) \in GL(V) \times \text{End}(V) \times V \times V^* | B - g^{-1} B g + IJ = \hbar \}/GL(V). 
\end{equation}
$GL(V)$ acts on $g \in GL(V)$ by conjugation. The analogous open subset $\iota: U \hookrightarrow \text{Hilb}_n(T^*_\hbar \mathbb{C}^\times)$ is where $g$ is diagonalizable with distinct eigenvalues. $\text{Hilb}_n(T^*_\hbar \mathbb{C}^\times)$ has a symplectic form descending from
\begin{equation}
    \omega_{\mathbb{C}} = d\tr(B g^{-1} dg) + \tr dI \wedge dJ. 
\end{equation}
One may check that if $g = \text{diag}(z_1, \dots, z_n)$, $0 < |z_1 |< \dots < |z_n|$ that 
\begin{equation}
    B_{ij} = p_i \delta_{ij} - \frac{\hbar}{1 - z_j/z_i}(1 - \delta_{ij})
\end{equation}
and the symplectic form is $\iota^*\omega_{\mathbb{C}} = \sum_i dp_i \wedge d \log(z_i)$. Let $\mathscr{F}_0(u)$ denote the Fock representation of the Heisenberg algebra $\comm{\alpha_m}{\alpha_n} = m \delta_{m + n, 0}$ in which the central element $\alpha_0$ acts by $u$. Denote the vacuum by $\ket{0; u} \in \mathscr{F}_0(u)$, it satisfies $\alpha_0 \ket{0; u} = u \ket{0; u}$ and $\alpha_n \ket{0; u} = 0$ for $n >0$. Write the mode expansion for the chiral boson on the cylinder:
\begin{equation}
    \partial \varphi^{cyl}(z) = \sum_{n \in \mathbb{Z}} \frac{\alpha_n}{z^n}
\end{equation}
and introduce a second infinite set of variables $\tilde{t} = (\tilde{t}_1, \tilde{t}_2, \dots,)$ (all but finitely many zero). 

\begin{prop} We have
\begin{equation} \label{formula:cylindermiuracoeff}
\begin{split}
    \bra{0; u} e^{\sum_{m = 1}^\infty t_m \alpha_m} (p_n - \hbar \partial \varphi^{\text{cyl}}(z_n)) & \dots (p_1 -\hbar \partial \varphi^{\text{cyl}}(z_1)) e^{\sum_{m = 1}^\infty \tilde{t}_m \alpha_{-m}} \ket{0; u}  = \\
    & e^{\sum_{m = 1}^\infty m t_m \tilde{t}_m} \iota^* \det\Big( B - \hbar u  - \hbar \sum_{m \geq 1} (m t_m g^m + m \tilde{t}_m g^{-m})\Big). 
\end{split}
\end{equation}
\end{prop}

\begin{proof} This is actually a corollary of our formula \eqref{formula:ISMcoefficients}. Note that, by the explicit formula for $B_{ij}$ above, $B_{ij} = z_i (B_2)_{ij}$ up to redefining $z_i p_i \to p_i$. Moreover, we trivially have (since $\alpha_0$ is central)
\begin{equation}
\begin{split}
    \bra{0; u} & e^{\sum_{m \geq 1} t_m \alpha_m} (p_n - \hbar \partial \varphi^{cyl}(z_n)) \dots (p_1 - \hbar \partial \varphi^{cyl}(z_1)) e^{\sum_{m \geq 1} \tilde{t}_m \alpha_{-m}}\ket{0; u} \\
    & = \bra{0; 0} e^{\sum_{m \geq 1} t_m \alpha_m}(\text{same with $p_i \to p_i - \hbar u$}) e^{\sum_{m \geq 1} \tilde{t}_m \alpha_{-m}}\ket{0 ;0}
\end{split}
\end{equation}
Then multiplying \eqref{formula:ISMcoefficients} by $\prod_i z_i$, redefining $z_i p_i \to p_i$ and then shifting $p_i \to p_i - \hbar u$ we get the statement of the proposition at $\tilde{t}_m = 0$. The general case follows by commuting $e^{\sum_m \tilde{t}_m \alpha_{-m}}$ to the left using the Heisenberg commutation relations. 
\end{proof}

Taking derivatives of \eqref{formula:cylindermiuracoeff} at $t_m = \tilde{t}_m = 0$, we find
\begin{theorem} \label{thm:Rmatrixcoeff}
All matrix coefficients of the operator $$(p_n - \hbar \partial \varphi^{\text{cyl}}(z_n)) \dots (p_1 - \hbar \partial \varphi^{\text{cyl}}(z_1)) \in \textnormal{End}(\mathscr{F}_0(u)) \otimes \mathscr{O}_U$$ in the oscillator basis $\prod_i\alpha_{-\lambda_i}\ket{0; u}$ of $\mathscr{F}_0(u)$, where $\lambda$ is some partition, extend from the Darboux open set $U$ to regular functions on $\textnormal{Hilb}_n(T^*_\hbar \mathbb{C}^\times)$; moreover these matrix coefficients generate the coordinate ring of $\textnormal{Hilb}_n(T^*_\hbar \mathbb{C}^\times)$.
\end{theorem}
This is a direct analog of the factorization formula \eqref{eq:monopoleSfactor} for the monopole scattering matrix for smooth $U(2)$ monopoles on $\mathbb{R}^3$, and provides a complete justification for the name ``instanton scattering matrix'' for the semiclassical Miura operator $p - \hbar \partial \varphi^{\text{cyl}}(z)$, where the Poisson bracket is now $\acomm{p}{z} = z$. 

Moreover, these global functions are generators under Poisson bracket of the (semiclassical) $\mathsf{Y}(\widehat{\mathfrak{gl}}_1)$ symmetry of $\text{Hilb}_n(T^*_\hbar \mathbb{C}^\times)$, by the $RTT$ relation enjoyed by Miura operators. It is well-known that quantization of $\text{Hilb}_n(T^*_\hbar \mathbb{C}^\times)$ as a symplectic variety (e.g. by exhibiting it as a Coulomb branch after Braverman-Finkelberg-Nakajima \cite{bfnslice}) leads to (spherical) Cherednik algebras \cite{koderanakajima}, and there are presentations of affine Yangians using Cherednik algebras, e.g. \cite{arbesfeldschiffmann}. 

The approach of this paper is, however, considerably more direct: in the approach of \cite{mo}, the affine Yangian is by definition generated by matrix coefficients of $R$-matrices in Fock modules. The formula \eqref{formula:cylindermiuracoeff} concisely identifies all possible such $R$-matrix coefficients with particular functions on $\text{Hilb}_n(T^*_\hbar \mathbb{C}^\times)$ (in the semiclassical limit), while making contact with the noncommutative geometry of curves in $T^*_\hbar \mathbb{C}^\times$ as in \cite{Aganagic_2005}, \cite{DHS}. 

\subsubsection{Pl\"{u}cker relations and a conjecture}
Note that a corollary of the residue computation done in the proof of Theorem \ref{thm:instantonfactor} is that the operator $S = (p_n - \hbar \partial \varphi^{cyl}(z_n)) \dots (p_1 - \hbar \partial \varphi^{cyl}(z_1))$ featuring in Theorem \ref{thm:Rmatrixcoeff} satisfies Pl\"{u}cker relations, i.e. $\comm{\Omega}{S \otimes S} = 0$ where
\begin{equation}
    \Omega = \oint \frac{dz}{2\pi i} \psi(z) \otimes \widetilde{\psi}(z) = \sum_{r \in \mathbb{Z} + 1/2} \psi_r \otimes \widetilde{\psi}_{-r}.
\end{equation}
By analogy with monopole scattering matrices, one expects that these quadratic relations among $S$-matrix coefficients, together with the requirement that for each pair of partitions $\mu, \nu$, 
\begin{equation}
    \text{$\bra{ 0; u} \prod_i \alpha_{\mu_i} S \prod_j \alpha_{-\nu_j} \ket{0; u}$ is a polynomial in $u$, of the form $(-\hbar u)^n C_{\mu \nu } + $ lower degree}
\end{equation}
generate all the equations cutting out the affine variety (note $\hbar \neq 0$!) $\text{Hilb}_n(T^*_\hbar \mathbb{C}^\times)$ for each $n$. Compare also to the discussion in chapters 17-18 in \cite{mo} on screening operators generating all relations in the core Yangian of $\widehat{\mathfrak{gl}}_1$. The constant $$C_{\mu \nu} = |\text{Aut}(\mu)| \prod_i \mu_i \delta_{\mu \nu}$$
is a standard normalization; $\text{Aut}(\mu)$ is a subgroup in the symmetric group on $\text{length}(\mu)$ letters. 

\subsubsection{Higher rank generalizations}
Much of what has been said has a generalization to rank $r > 1$, though certain statements are more involved because the dual bosonic description is that of the WZW model of $\widehat{\mathfrak{gl}}_r$ at level $k = 1$ \cite{witten84}; in particular the role of $\mathscr{F}_0$ is played by the level $k = 1$ integrable module for $\widehat{\mathfrak{gl}}_r$ induced from the trivial $\mathfrak{gl}_r$-module (in terminology standard from vertex algebras, this is the simple quotient of the vacuum module $V_1(\mathfrak{gl}_r)$). In this section we will just state some of the results that can be obtained by a mild variation on the theme above. 

Denote by $J\indices{^a_b}(z)$, $a, b = 1, \dots, r$ the $\widehat{\mathfrak{gl}}_r$ currents at level $k = 1$, which satisfy OPEs 
\begin{equation}
    \tr(AJ(z)) \tr (B J(w)) \sim \frac{\tr(AB)}{(z - w)^2} -\frac{\tr(\comm{A}{B}J(w))}{(z - w)} + \text{reg}
\end{equation}
where $A, B$ are $r \times r$ matrices. 

The instanton moduli space $\widetilde{M}(n, r)$ has a distinguished open set $\iota: U \hookrightarrow \widetilde{M}(n, r)$ where $B_1$ is diagonalizable with distinct eigenvalues. Let $\mathscr{O}_{(\hbar, 0, \dots, 0)} \subset \mathfrak{gl}_r$ denote the conjugacy class of $\text{diag}(\hbar, 0, \dots, 0)$; we may view it as a holomorphic symplectic variety using the trace to identify $\mathfrak{gl}_r$ with $\mathfrak{gl}_r^*$ and view $\mathscr{O}_{(\hbar, 0, \dots, 0)}$ as a coadjoint orbit endowed the the Kirillov symplectic form. There is an isomorphism of symplectic varieties 
\begin{equation}
    U \simeq (T^*(\mathbb{A}^n \setminus \Delta) \times \mathscr{O}^{\times n}_{(\hbar, 0, \dots, 0)})/W
\end{equation}
where $W$ denotes the symmetric group $S_n$ acting by permutation of the factors, and $\Delta$ denotes the union of diagonals $z_i = z_j$ in $\mathbb{A}^n$. Let $(z_i, p_i, X_i)$, $i = 1, \dots, n$ be coordinates on $U$; we may parameterize $X_i \in \mathscr{O}_{(\hbar, 0, \dots, 0)}$ as $(X_i)\indices{^a_b} = J^a_i I_{ib}$ where $J_i, I_i$ are defined up to scaling $(J_i, I_i) \mapsto (t_i J_i, t_i^{-1} I_i)$ for $t_i \in \mathbb{C}^\times$, and satisfy $\sum_a I_{ia} J^a_i = \hbar$. 

Theorem \ref{thm:instantonfactor} has the following analog, with almost the same proof. $\ket{\mathbb{W}}$ denotes the map $\widetilde{M}(r) \to \mathscr{F}$ constructed in Propositions \ref{maptoGr} and \ref{vaccomponent}. We have 
\begin{equation}
    \iota^* \ket{\mathbb{W}} = (p_n - \tr(X_n J(z_n))) \dots (p_1 - \tr(X_1 J(z_1))) \ket{0}.
\end{equation}
Then a corollary of this result and Proposition \ref{vaccomponent} is the following formula for a correlation function in the $\widehat{\mathfrak{gl}}_r$ WZW model at level $k = 1$: 
\begin{equation}
\begin{split}
    \bra{0} (p_n - \tr(X_n J(z_n))) \dots (p_1 - \tr(X_1 J(z_1))) \ket{0} & = \iota^* \det(B_2) \\ 
    & = \det_{1 \leq i, j \leq n} \Big( p_i \delta_{ij} - \sum_{a = 1}^r \frac{I_{ia} J^a_j}{z_i - z_j}(1 - \delta_{ij}) \Big). 
\end{split}
\end{equation}

The only real difference in the higher rank case is that it is more difficult to get a formula for the general matrix coefficient such as \eqref{formula:cylindermiuracoeff}; we will return to this elsewhere.

\section{Top-down approach via String/M-theory} \label{section:string}
Having explained the essence of our construction in the language of noncommutative gauge theory and algebraic geometry, we will now explain a different path to the same results starting from string/M-theory constructions. This will clarify many aspects of this setup and explain its relation to other parts of the literature, as well as motivate certain conjectures. Some readers familiar with the physics literature may find this path more natural and may wish to start with this section. 

The string theory setup has been investigated before in \cite{DHS}, \cite{DHSV}, \cite{costello2016}, so we will gloss over various technicalities in order to emphasize the novel aspects of the situation under study in this paper. 

\subsection{M-theory background}
We will start out with M-theory on $\mathbb{R} \times T^* \mathbb{C} \times \mathbb{C}^3_{\varepsilon_1, \varepsilon_2, \varepsilon_3}$ with $\varepsilon_1 + \varepsilon_2 + \varepsilon_3 = 0$ ($\varepsilon_i$ denote $\Omega$-deformation parameters \cite{nekrasov2002}---see \cite{costello2016} for a study of the M-theory $\Omega$-background that seems suitable for our purposes. We will simply assume the background exists\footnote{See also \cite{nekrasovokounkov} for a discussion of M-theory in the $\Omega$-background from the point of view of index calculations.} and satisfies certain reasonable properties). Wrap a collection of $n$ M2 branes along $\mathbb{C}_{\varepsilon_1}$ and the $\mathbb{R}$-direction. Then add a single M5 brane which wraps the zero section of $T^*\mathbb{C}$ along with $\mathbb{C}^2_{\varepsilon_2, \varepsilon_3}$. In fact, in most of what follows it will be enough to set $\varepsilon_1 = 0$, and $\varepsilon_2 = - \varepsilon_3 = \hbar$. We will use $(w, z)$ as coordinates of $T^* \mathbb{C}$, following notations from Section \ref{section:scatterandmiura}. 

\subsubsection{Reduction to Type IIA}
The M-theory description is difficult to use directly, as it is very strongly coupled. To get a more accessible description of the system, we consider reducing to Type IIA string theory, and in fact there are two natural ways to do this. The idea is that we may view either $\mathbb{C}^2_{\varepsilon_2, \varepsilon_3}$ or $T^*\mathbb{C}$ as a Taub-NUT manifold and reduce along the $S^1$ fiber of the corresponding fibration 
\begin{equation}
    TN \to \mathbb{R}^3
\end{equation}
to get a description via Type IIA string theory with a D6 brane; the IIA string coupling is related to the radius of the $S^1$ fiber. In this paper, we will focus on the reduction along the $S^1$ fiber of $\mathbb{C}^2_{\varepsilon_2, \varepsilon_3}$; the other reduction is interesting and very relevant for certain questions in categorical Donaldson-Thomas theory, but these will be outside our main line of investigation here (they will be addressed, in part, in the forthcoming work \cite{bottatamagni}). 

It should be pointed out that $\mathbb{C}^2$ coincides with Taub-NUT only as a complex symplectic manifold; their hyperkahler metrics are different, and one expects that the full dynamics of M-theory certainly depends on the choice of metric. In fact the Taub-NUT metric is really a one-parameter family of metrics depending on a variable $R$ which is the asymptotic radius of the $S^1$ fiber; $R \to \infty$ corresponds to the flat metric on $\mathbb{C}^2$. For the sorts of computations we perform here, which are protected by supersymmetry and may be reformulated as computations in an appropriately twisted version of the theory, it is expected that the answers will not depend on the Taub-NUT radius (this is also claimed in \cite{costello2016}). The Type IIA description is valid, strictly speaking, for small values of $R$.

\subsubsection{The instanton scattering frame}
Consider reducing along $\mathbb{C}^2_{\varepsilon_2, \varepsilon_3}$, at $\varepsilon_1 = 0$ and $\varepsilon_2 = - \varepsilon_3 = \hbar$. We arrive at Type IIA string theory on $\mathbb{R} \times T^* \mathbb{C} \times \mathbb{C} \times \mathbb{R}^3$ with a D6 brane wrapping $\mathbb{R} \times T^* \mathbb{C} \times \mathbb{C}$, $n$ D2 branes wrapping $\mathbb{R} \times \{ \text{points} \} \times \mathbb{C} \times \{ 0 \}$, and a D4 brane wrapping $\{ w = 0 \} \times \mathbb{R}^3$. The $\Omega$-deformation in M-theory gives rise to an effective $B$-field with $(2, 0)$ part 
\begin{equation}
    B = \frac{1}{\hbar} dw \wedge dz
\end{equation}
in this IIA frame. This setup, with $n = 0$ D2 branes, was considered at length in \cite{DHSV}. We will be interested in the description via the low-energy effective gauge theories on the branes.

On the D6 brane, the low energy dynamics is described by a seven-dimensional maximally supersymmetric Yang-Mills theory with gauge group $U(1)$. The $B$-field may be incorporated in the low-energy description by following a well-established path \cite{seibergwitten}: it has the effect of making the $(w, z)$ directions of the worldvolume noncommutative, so that the gauge fields on the D6 brane worldvolume should be regarded as noncommutative Yang-Mills fields insofar as their dependence on $(w, z)$ is concerned. We will be treating this D6 brane as so massive that the gauge theory on its worldvolume may be regarded as semiclassical. While $U(1)$ gauge theory on a commutative space is a free field theory with no interesting instanton solutions, its noncommutative analog is a nonlocal interacting field theory that supports nontrivial (``stringy'') instanton solutions \cite{Nekrasov_1998}.

The D2 branes are codimension four defects inside the D6 branes and preserve half of the 16 supersymmetries on the D6 brane worldvolume. They may be described in the effective theory on the D6 brane as solutions to the noncommutative instanton equations in the $(w, z)$ directions, which are independent of the other coordinates. The instanton charge $n$ is the same as the number of D2 branes. 

Finally, it is well-known that the zero modes of D4-D6 open strings are chiral fermions supported on the D4-D6 intersection (this was called the ``I-brane configuration'' in \cite{DHSV}, \cite{DHS}). Of course, the 4-6 open strings are naturally charged in the fundamental representation of the D6 brane gauge fields. The natural coupling of these chiral fermions to the gauge fields is then described via the nonlocal effective action
\begin{equation}
S = \int_{\mathbb{C} \subset T^* \mathbb{C}} \widetilde{\psi}(z, \bar{z}) \Bigg( \overline{\partial}_z + A_{\bar{z}}(z, \overline{z}, w, \bar{w}) \eval_{w = \hbar \partial_z} \Bigg) \psi(z, \bar{z})
\end{equation}
as discussed in Section \ref{section:scatterandmiura}. The D2 branes are fully taken into account, in the semiclassical limit in which we are working, by assuming $A_{\bar{z}}$ above is the $\bar{z}$-component of a noncommutative instanton solution, thus fully described by the ADHM data $(B_1, B_2, I, J)$ via the noncommutative ADHM construction of \cite{Nekrasov_1998}. Instanton scattering matrix coefficients are given by partition functions that arise by integrating out the 4-6 strings; these are functional determinants $\det \overline{D}_z$ of the (highly nonlocal) chiral Dirac operator describing the holomorphic structure of the restriction of the instanton bundle to $w = 0$. Proposition \ref{taufnformula} gives a complete description of this chiral determinant in terms of the ADHM data.

A succinct way to link the proposal of \cite{DHS}, \cite{DHSV} with the gauge theory discussion of Section \ref{section:scatterandmiura} is to say that the dynamics of the zero modes of the 4-6 strings is well-described by a noncommutative $B$-model\footnote{Very abstractly, the noncommutative $B$-model with target $T^*_\hbar \mathbb{C}$ is the 2d topological field theory determined by asserting that its category of boundary conditions is the bounded derived category of finitely generated right $\mathscr{D}_\hbar$-modules. It may be equivalently described as an $A$-model with a $B$-field; for a related discussion see section 4.3 of \cite{aganagic2017}.} on $T^*_\hbar \mathbb{C}$. This means that, ignoring all directions in the IIA geometry besides $T^*_\hbar \mathbb{C}$, we may model the D4 brane as a holonomic right $\mathscr{D}_\hbar$-module of the form $\widehat{\mathscr{O}}_C$ (in the notation of Appendix \ref{adhmdetails}), and capture the D2-D6 configuration by the right $\mathscr{D}_\hbar$ module $\mathscr{E}$ described via the noncommutative ADHM construction. As in any $B$-model, the space of open string states is given by an Ext-group,
\begin{equation}
    \mathscr{H}_{4-6} = \text{Ext}^1_{\mathscr{D}_\hbar}(\widehat{\mathscr{O}}_C, \mathscr{E})
\end{equation}
which we explained directly in gauge theory in Section \ref{section:scatterandmiura}. Note that when the instanton charge $n = 0$, $\mathscr{E} = \mathscr{D}_\hbar$ and $\mathscr{H}_{4-6} \simeq \widehat{\mathscr{O}}^\ell_C$, the left $\mathscr{D}_\hbar$-module version of $\widehat{\mathscr{O}}_C$. This says that the space of 4-6 zero modes is the underlying vector space of a $\mathscr{D}$-module on $\mathbb{A}^1$---precisely the claim of \cite{DHS}, \cite{DHSV}! 

Another claim of \cite{DHS}, \cite{DHSV} is that these $\mathscr{D}_\hbar$-modules define points in the Sato Grassmannian $\text{Gr}(\mathscr{H})$ by considering asymptotic expansions of fundamental solutions to differential equations of the form $P(z, w) \Psi= 0$ for some $P \in \mathscr{D}_\hbar$; this was verified in \cite{DHS} by explicit computations. 

The presence of the $n \neq 0$ D2 branes modifies all this as follows. The space $\mathscr{H}_{4-6}$ of 4-6 open strings now varies as the fiber of a vector bundle over the moduli space of D2 branes; this moduli space is a copy of the moduli space of noncommutative $U(1)$ instantons, $\widetilde{M}(n, 1)$ in the notation of this paper. $\mathscr{H}_{4-6}$ defines a point in $\text{Gr}(\mathscr{H})$ for each point of the instanton moduli space, i.e. a map from the instanton moduli space to $\text{Gr}(\mathscr{H})$. This map, in the special case $P(z, w) = w$, was constructed rigorously in Section \ref{section:scatterandmiura} and denoted by $F$ there.

This vector bundle describes the space of classical solutions to the equation of motion $\overline{D}_z \psi = 0$ in the free fermion theory, as we vary over all possible choices of the instanton background. Understanding the geometry of this vector bundle makes it relatively straightforward to quantize; as explained by Witten long ago \cite{witten} and reviewed in Appendix \ref{fermionreview}, the Pl\"{u}cker embedding of $\text{Gr}(\mathscr{H})$ into the projectivization of the fermionic Fock space gives an explicit and mathematically rigorous way to construct the quantum states prepared by fermion path integrals on various algebraic curves (the Grassmannian machinery works, in particular, when the explicit $\overline{D}_z$ operator is too complicated to be useful or is unavailable entirely). This applies, and in some sense is an even more powerful tool, as we vary the choice of the instanton background. In this way we define the state denoted $\ket{\mathbb{W}}$ in Section \ref{section:scatterandmiura}. The rest of the analysis of Section \ref{section:scatterandmiura} just amounts to obtaining a more explicit description of $\ket{\mathbb{W}}$ using the standard tools of bosonization and the Ward identities. 

From this perspective, it is clear that a degree of freedom that we may incorporate, but has been ignored so far, is the ability to choose a more general algebraic curve $P(z, w) = 0$ in $T^* \mathbb{C}$ as the support of the fermions. This is somewhat orthogonal to our main line of development here, which is focused on quantum groups and $R$-matrices. Nonetheless, the construction of this paper gives rise to a deformation of the partition function of the $B$-model of topological string theory on all backgrounds of the form $uv + P(z, w) = 0$ for each noncommutative $U(1)$ instanton on the $(z, w)$-plane, generalizing the framework of \cite{Aganagic_2005}, \cite{DHS}, \cite{DHSV}. It would be interesting to find a more direct interpretation of the instantons in topological string theory. 

In the case of the conifold $P(z, w) = zw - u$ (the $B$-model of which is well-known to compute the partition function of the $c = 1$ noncritical string theory), this will be discussed explicitly in Section \ref{section:newRmatrix}. The operator $\mathbb{S}(u)$ constructed there is an instanton generalization of the $S$-matrix of the $c = 1$ noncritical string theory (justifying, in a different way, the name ``instanton scattering matrix''). 

Note also that the case of $U(r)$ instantons, from this perspective, corresponds to having $r$ D6 branes, which in turn arise from $M$-theory on $\mathbb{R} \times T^* \mathbb{C} \times \mathbb{C} \times A_{r - 1}$. 

\subsubsection{Instanton scattering matrices and Chern-Simons theories}
From the perspective of this paper, the instanton scattering matrices are expected to be $R$-matrices both by analogy with their monopole analogs and by the explicit computations in Section \ref{section:scatterandmiura}. The string theory embedding gives a more conceptual explanation for this as follows. 

It is argued in \cite{costello2016} that M-theory in the background we study can be described effectively by a noncommutative $U(1)$ Chern-Simons theory on $\mathbb{R} \times T^*_\hbar \mathbb{C}$. The argument is essentially by considering the same gauge theory on the D6 brane discussed above, in a ``$B$-type'' $\Omega$-background as in \cite{nekrasovtying}, \cite{Yagi_2014} dictated by the $\varepsilon_1$ parameter. If we think of $T^* \mathbb{C}$ as an $S^1$-fibration over $\mathbb{R}^3$, this Chern-Simons theory may be viewed as a loop/affine analog of the four-dimensional holomorphic-topological Chern-Simons theory studied in \cite{costelloyangian}, \cite{costellowittenyamazaki}. In much the same way that this Chern-Simons theory is related to the Yangians $\mathsf{Y}(\mathfrak{g})$ of finite-dimensional Lie algebras, one expects that the five-dimensional Chern-Simons theory is related to the Yangians $\mathsf{Y}(\widehat{\mathfrak{g}})$ of the corresponding affine Lie algebras. 

$R$-matrices arise naturally in this framework as follows: if we study the four-dimensional Chern-Simons theory on a background $\mathbb{R}^2 \times \mathbb{C}$ where $\mathbb{R}^2$ is topological, it is possible to defin line operators pointlike on $\mathbb{C}$ and extended along some line in $\mathbb{R}^2$. The most basic such line operators are Wilson lines in fundamental representations of the gauge group, and it was argued in \cite{costellowittenyamazaki} that crossings of Wilson lines in $\mathbb{R}^2$ are described by Yangian $R$-matrices, via explicit Feynman diagram computations. It was argued in \cite{costelloQop} (in a slightly different language) that monopole scattering matrices may be interpreted as the classical limit of an intersection of a Wilson line defect with a 't Hooft operator in the four-dimensional Chern-Simons theory. From this perspective, the appearance of Yangians of finite ADE type in the discussion of monopole scattering matrices is less surprising. 

In \cite{costello2016} it is argued that the M2 branes are realized in the Chern-Simons description via line defects extended along $\mathbb{R}$ and pointlike along $T^*_\hbar \mathbb{C}$; more precisely, described by a solution to the noncommutative instanton equations along $T^*_\hbar \mathbb{C}$. The M5 brane is a surface operator wrapping $\mathbb{C} \subset T^*_\hbar \mathbb{C}$. Upon projecting to the base of $T^* \mathbb{C} \to \mathbb{R}^3$, splitting $\mathbb{R}^3 = \mathbb{R} \times \mathbb{C}$ and forming an $\mathbb{R}^2$ via combining this $\mathbb{R}$ with the $\mathbb{R}$ the M2 branes wrap, we see that the system under study looks like the intersection of a line operator made out of M2 branes with a \textit{half}-line operator made out of the M5 brane. The main assertion of Section \ref{section:scatterandmiura} of this paper is that this intersection is fully described by the map $\ket{\mathbb{W}}$ from the instanton moduli space $\widetilde{M}(1)$ to the Fock space $\mathscr{F}$ constructed there, in complete parallel to the role of the monopole scattering matrix in \cite{costelloQop}.

Factoring the M2 brane line operator into elementary constituents corresponds to passing to the Darboux chart of the instanton moduli space called $U$ in Section \ref{section:scatterandmiura}. Theorem \ref{thm:instantonfactor} then expresses how $\ket{\mathbb{W}}$ is composed from elementary pieces on this open set; in particular the proposal from \cite{gaiottorapcak2020} that the Miura operator corresponds to an M2-M5 intersection can be regarded in our framework as a precise mathematical theorem.  Our framework seems totally different from other, more algebraic approaches to the same physics setup such as \cite{gaiottorapcakzhou24}. 

Our techniques should be compared and contrasted with \cite{miuraCSfeynman2}, \cite{miuraCSfeynman1} which compute Miura operators via Feynman graph calculations along the lines of \cite{costellowittenyamazaki}. The methods of this paper are \textit{nonperturbative}, in the sense that they do not require working in a power series expansion about the trivial configuration of the gauge fields of the bulk Chern-Simons theory. On the other hand, they apply only in the semiclassical regime, $\varepsilon_1 \to 0$. This is not as limiting as it may seem, e.g. as we saw in Section \ref{section:scatterandmiura}, the 1-instanton scattering matrix has an essentially unique $\varepsilon_1 \neq 0$ deformation and insisting on compatibility with the factorization in Theorem \ref{thm:instantonfactor} gives an explicit and unambiguous assignment of quantized instanton scattering matrix coefficients to differential operators in the $z_i$ variables. 

While we have emphasized the 5d Chern-Simons perspective in this discussion, that $R$-matrices of affine Yangians (or more generally, quantum toroidal algebras) correspond to intersections of branes in string theory is a key component of a general proposal of Y. Zenkevich \cite{zenkevich2023spirallingbranesrmatrices}, \cite{zenkevich24}; it was in this form that this author first learned of many of these results. 

\subsubsection{A second kind of scattering matrix} \label{subsect:introS(u)}
From the Chern-Simons perspective, it is clear that there is in some sense a more natural choice of support for the M5 brane worldvolume: the one corresponding to a straight line in the $\mathbb{R}^2$ in which line operators are intersecting to give $R$-matrices. Because the $\mathbb{R}$ direction supporting the M5 brane operator is the vertical axis in the base of $T^* \mathbb{C} \to \mathbb{R}^3 \simeq \mathbb{R} \times \mathbb{C}$, this is the curve $zw = u$ in the full geometry; the parameter $u$ is identified with the position of this line in the $\mathbb{C}$ plane, and therefore with the spectral variable of the shifted Yangian $\mathsf{Y}_{+1}(\widehat{\mathfrak{gl}}_1)$. 

Note that, in this language, studying the \textit{unshifted} affine Yangian corresponds to replacing $T^* \mathbb{C}$ with $T^* \mathbb{C}^\times$, in which case going from the half-line to a full line presents no essential difficulty. This is reflected in the comparative ease of proving Theorem \ref{thm:Rmatrixcoeff} as a corollary of Theorem \ref{thm:instantonfactor}. On the other hand, it is known from a number of perspectives in the geometric representation theory literature that dominant shifted Yangians present a considerable level of new complexity relative to their unshifted or antidominant shifted counterparts. The first nontrivial example, in the setting of this paper, corresponds to instanton moduli on $T^* \mathbb{C}$. More involved examples arise from instanton moduli on $A_{n - 1}$ surfaces, whose scattering matrices lead to $\mathsf{Y}_{+n}(\widehat{\mathfrak{gl}}_r)$ in general.

This manifests in our current problem via the following: the curve $zw = u$ has two distinct asymptotic ends in $T^* \mathbb{C}$; one approaches the $z$-axis and the other approaches the $w$-axis. Naively, one would want to ``glue together'' two copies of the discussion of Section \ref{section:scatterandmiura} at each end, but here the noncommutativity of $T^* \mathbb{C}$ enters: essential use was made of quantization in the $z$-polarization in Section \ref{section:scatterandmiura}, and the comparison between the $z$ and $w$-polarizations is nonlocal and given by Fourier transform. 

In Section \ref{section:newRmatrix} we will make this idea precise and explain how to correctly implement the Fourier transform; the output will be a new kind of scattering operator $\mathbb{S}(u)$ which should be thought of as an operator on the Fock space $\mathscr{F}$ with matrix coefficients in global functions on the instanton moduli space $\text{Hilb}_n(T^*_\hbar \mathbb{C})$ (or more generally, $\widetilde{M}(n, r)$, for $r > 1$). Note by Theorem \ref{thm:Rmatrixcoeff}, the product of Miura operators has matrix coefficients in global functions on $\text{Hilb}_n(T^*_\hbar \mathbb{C}^\times)$; to get global functions on $\text{Hilb}_n(T^*_\hbar \mathbb{C})$ we need a new operator, and defining it requires more effort. The operator $\mathbb{S}(u)$ appears to be a new kind of $R$-matrix for a Poisson limit of $\mathsf{Y}_{+1}(\widehat{\mathfrak{gl}}_r)$. We will analyze some of its basic structure in the next section, and leave more detailed investigation of its properties to the future. 
 
\section{Instanton scattering and new $R$-matrices} \label{section:newRmatrix}
In this section we will return to the mathematical world and carry out the plan outlined in Section \ref{subsect:introS(u)}. In we Section \ref{sect:defineS} give a geometric definition of an operator $\mathbb{S}(u): \mathscr{F}^{\text{in}} \to \mathscr{F}^{\text{out}}$ between two fermionic Fock spaces, from which it follows that it satisfies Pl\"{u}cker relations and is therefore determined by its vacuum matrix element and the fermionic two-point function. These two quantities are computed explicitly in terms of ADHM data in Section \ref{sect:defineS}, and conjectures are made on the role of $\mathbb{S}(u)$ as an $R$-matrix for $\mathsf{Y}_{+1}(\widehat{\mathfrak{gl}}_r)$ for $r = 1$. Because the $RTT$ formalism for dominantly shifted Yangians does not yet exist, we do not pursue the proof of these conjectures here. 

Finally, the discussion from Section \ref{section:string} makes clear that the operator $\mathbb{S}(u)$ is a direct uplift of monopole scattering matrices (the objects of Section \ref{section:scatterandmiura} are more like the first column of monopole scattering matrices). As monopoles on $\mathbb{R}^3$ with certain Dirac singularities can be obtained as components of the fixed locus for an $S^1$-action on the moduli space of instantons on $\mathbb{R}^4$ \cite{kronheimer}, this leads to a conjectural way to read off monopole scattering matrices from the instanton ones (formula \eqref{eq:Smatrixfromprop} below), at least in any case where the moduli space of monopoles arises by taking fixed points in instanton moduli. This is explained in Section \ref{sect:monopolefrominstanton}.

\textit{Note for the worried reader.} In this section, for the sake of brevity and clarity of exposition we are vague about technicalities related to duals and completions. For the goals of this paper, this is harmless since we are just interested in the (infinite) list of matrix coefficients of $\mathbb{S}(u)$ in the standard basis of the fermionic Fock space, each of which is well-defined. For further investigations of $\mathbb{S}(u)$, e.g. its compatibility with coproducts on $\mathsf{Y}_{+1}(\widehat{\mathfrak{gl}}_r)$, these issues should be addressed with more care. 

\subsection{Defining the operator $\mathbb{S}(u)$} \label{sect:defineS}
Generalizing the discussion of Section \ref{section:scatterandmiura}, we are instructed to study the vector bundle 
\begin{equation}
    \mathbb{V}_u := \mathscr{E} \otimes_{\mathscr{D}_\hbar} \widehat{\mathscr{O}}^\ell_{C_u} \simeq \text{Ext}^1_{\mathscr{D}_\hbar}(\widehat{\mathscr{O}}_{C_u}, \mathscr{E})
\end{equation}
over $\widetilde{M}(r)$, where $\mathscr{E}$ is the universal instanton bundle and 
\begin{equation}
\begin{split}
    \widehat{\mathscr{O}}^\ell_{C_u} & := \mathscr{D}_\hbar/\mathscr{D}_\hbar (zw - u) \\ 
    \widehat{\mathscr{O}}_{C_u} & := \mathscr{D}_\hbar/(zw - u) \mathscr{D}_\hbar
\end{split}
\end{equation}
are respectively left and right $\mathscr{D}_\hbar$-modules, depending on the complex parameter $u$. We assume $u \notin \hbar \mathbb{Z}$. 

\subsubsection{$\mathbb{V}_u$ as a map to $\textnormal{Gr}(\mathscr{H}_{\textnormal{in}} \oplus \mathscr{H}_{\textnormal{out}})$}
To fully set up the scattering problem, we need not just the vector bundle $\mathbb{V}_u$ but a piece of additional data which allows us to define a map from the instanton moduli space to a Sato Grassmannian denoted by $\text{Gr}(\mathscr{H}_{\text{in}} \oplus \mathscr{H}_{\text{out}})$. Let us first define this Grassmannian. We write 
\begin{equation}
\begin{split}
    \mathscr{H}_{\text{in}} & := W((z^{-1})) \\
    \mathscr{H}_{\text{out}} & := W(( w^{-1}))
\end{split}
\end{equation}
where $W$ is the framing space in the ADHM description of $\widetilde{M}(r)$. We have the quotient $\mathscr{H}_+ = W[z] \oplus wW[w]$ of $\mathscr{H}_{\text{in}} \oplus \mathscr{H}_{\text{out}}$, and the Sato Grassmannian of subspaces such that the induced projection to $\mathscr{H}_+$ has finite rank kernel and cokernel.

The following statements are obtained by a ``two-sided'' version of the computations done in Section \ref{section:scatterandmiura}; therefore the details are omitted here and presented in Appendix \ref{adhmdetails}. 

\begin{prop} \label{prop:VmaptoGr}
Assume $u \notin \hbar \mathbb{Z}$. Using the description of $\mathbb{V}_u$ of Proposition \ref{extgroups}, there is a canonical morphism of sheaves 
\begin{equation}
    \mathbb{V}_u \to \mathscr{O}_{\widetilde{M}(r)} \otimes (\mathscr{H}_{\textnormal{in}} \oplus \mathscr{H}_{\textnormal{out}})
\end{equation}
over $\widetilde{M}(r)$, which is an inclusion of vector bundles. Moreover under the description of $\mathbb{V}_u$ from Proposition \ref{prop:explicitextforzw}, there are canonical isomorphisms
\begin{equation}
\begin{split}
    \ker(\mathbb{V}_u \to \mathscr{O}_{\widetilde{M}(r)} \otimes \mathscr{\mathscr{H}_+}) & \simeq \ker(u - B_2 B_1) \\
    \textnormal{coker}(\mathbb{V}_u \to \mathscr{O}_{\widetilde{M}(r)} \otimes \mathscr{H}_+ ) & \simeq \textnormal{coker}(u - B_2 B_1)
\end{split}
\end{equation}
of sheaves over $\widetilde{M}(r)$.
\end{prop}

\begin{proof}
    This occupies Appendix \ref{subsect:explicitextforzw}, \ref{subsect:functlzw}, \ref{proof:VmaptoGr}.
\end{proof}

By the above proposition, $\mathbb{V}_u$ gives rise to a morphism 
\begin{equation}
    G: \widetilde{M}(r) \to \text{Gr}(\mathscr{H}_{\text{in}} \oplus \mathscr{H}_{\text{out}})
\end{equation}
which we may compose with the Pl\"{u}cker embedding $\text{Gr}(\mathscr{H}_{\text{in}} \oplus \mathscr{H}_{\text{out}}) \hookrightarrow \mathbb{P}(\mathscr{F}^{\text{in}} \otimes \mathscr{F}^{\text{out}})$ to get a morphism $$\ket{\mathbb{V}_u} : \widetilde{M}(r) \to \mathbb{P}(\mathscr{F}^{\text{in}} \otimes \mathscr{F}^{\text{out}}).$$

\begin{prop} \label{prop:S(u)vaccoeff}
    $\ket{\mathbb{V}_u}$ lifts to a map to $\mathscr{F}^{\text{in}} \otimes \mathscr{F}^{\text{out}}$, and there is a unique choice of lift satisfying $\braket{0}{\mathbb{V}_u} = \det(u - B_2 B_1)$. 
\end{prop}

\begin{proof}
    Given Proposition \ref{prop:VmaptoGr}, the same argument as in Proposition \ref{vaccomponent} applies. 
\end{proof}

It is convenient at this point to introduce some notation for the components of $\ket{\mathbb{V}_u}$ in the standard basis of $\mathscr{F}^{\text{in}} \otimes \mathscr{F}^{\text{out}}$. By the general yoga of semi-infinite wedge as reviewed in Appendix \ref{fermionreview}, each of these is given by the determinant of some finite matrix with coefficients in functions on the instanton moduli space. Passing to duals, we may think of these as the matrix coefficients of a family of operators $$\mathbb{S}(u): \widetilde{M}(r) \to \text{Hom}(\mathscr{F}^{\text{in}, \vee}, \mathscr{F}^{\text{out}})$$
and this will be our definition of the operator $\mathbb{S}(u)$ in this paper. Sometimes we will drop the dual on $\mathscr{F}^{\text{in}}$ for simplicity.

We have the vacuum vector $\ket{0} = \ket{0}_z \otimes \ket{0}_w \in \mathscr{F}^{\text{in}} \otimes \mathscr{F}^{\text{out}}$, which in the semi-infinite wedge notation is $\ket{0}_z = 1 \wedge z \wedge z^2 \wedge \dots$, $\ket{0}_w = w \wedge w^2 \wedge w^3 \wedge \dots$. The linear functional $\bra{0} = \bra{0}_z \otimes \bra{0}_w$, where $\bra{0}_w = \dots \wedge w^{-2} \wedge w^{-1} \wedge 1$ etc. Our notation for $\mathbb{S}(u)$ is such that the following equalities hold:
\begin{equation}
    \bra{0} \widetilde{\psi}^{\text{in}}(\zeta) \psi^{\text{out}}(w) \ket{\mathbb{V}_u} = (\bra{0}_{\zeta} \widetilde{\psi}(\zeta) \otimes \bra{0}_w \psi(w)) \ket{\mathbb{V}_u} =: \bra{0} \widetilde{\psi}(\zeta) \mathbb{S}(u) \psi(w) \ket{0}
\end{equation}
and likewise for strings of $\psi$'s. Note in this notation, the vacuum vector on the right kills the \textit{positive} powers of $w$. This choice is particularly natural from the perspective of the operator formalism of conformal field theory on the algebraic curve $zw = u$ \cite{witten}. In this notation, Proposition \ref{prop:S(u)vaccoeff} fixes the vacuum matrix element 
\begin{equation}
    \bra{0} \mathbb{S}(u) \ket{0} = \det(u - B_2 B_1)
\end{equation}
which determines the normalization of the operator $\mathbb{S}(u)$. 

\subsection{Determining the operator $\mathbb{S}(u)$} \label{sect:determineS}
Paralleling the developments of Section \ref{section:scatterandmiura}, we want to give a more direct algebraic characterization of $\mathbb{S}(u)$. Because of the nonlocal nature of $\mathbb{S}(u)$, interpolating between the region of $z \to \infty$ and $w \to \infty$, it is a somewhat more involved object than the strings of Miura operators encountered in Section \ref{section:scatterandmiura}. 

\subsubsection{$\tau$-functions and Miura operators}
As a warm-up, pass to the rank $r = 1$ case and observe that there are two actions of the group $\Gamma_+$ (in the notation of \cite{segal-wilson}) on $\text{Gr}(\mathscr{H}_{\text{in}} \oplus \mathscr{H}_{\text{out}})$, one for each summand. 

\begin{prop} \label{prop:taufnforS(u)}
We have 
\begin{equation} 
\begin{split}
    \bra{0} e^{\sum_{m = 1}^\infty t_m \alpha_m^{\text{in}}} \ket{\mathbb{V}_u} & = \bra{0} e^{\sum_{m = 1}^\infty t_m \alpha_m} \mathbb{S}(u) \ket{0} = \det \Big( u - B_2 B_1 + \hbar\sum_{m = 1}^\infty m t_m B_1^m  \Big)  \\
    \bra{0} e^{\sum_{m = 1}^\infty \tilde{t}_m \alpha^{\text{out}}_m} \ket{\mathbb{V}_u} & = \bra{0} \mathbb{S}(u) e^{\sum_{m = 1}^\infty \tilde{t}_m \alpha_{-m}} \ket{0} = \det \Big(u - B_2 B_1 - \hbar \sum_{m = 1}^\infty m \tilde{t}_m B_2^m \Big).  
\end{split}
\end{equation}
\end{prop}

\begin{proof}
By the same argument as Proposition \ref{taufnformula}, details omitted. 
\end{proof}
Recall the open subset $\iota: U \hookrightarrow \widetilde{M}(n, 1)$, on which $B_1$ is diagonalizable with distinct eigenvalues $z_i$ and $B_2$ has diagonal elements $p_i$. Introduce a second open subset $\jmath : V \hookrightarrow \widetilde{M}(n, 1)$ on which $B_2$ is diagonalizable with distinct eigenvalues $\widetilde{p}_i$ and in which $B_1$ has diagonal elements $\widetilde{z}_i$. As a state in the Fock space is uniquely determined by its $\tau$-function, we have from \eqref{formula:ISMcoefficients}
\begin{theorem} \label{thm:Scolumns factor}
In rank $r = 1$, the first row and column of the instanton scattering matrix $\mathbb{S}(u)$ satisfy (we write formulas in terms of $\ket{\mathbb{V}_u}$ to avoid confusion with notation)
\begin{equation}
\begin{split}
    \bra{0}_w \iota^* \ket{\mathbb{V}_u} & = (u - z_n p_n + \hbar z_n \partial \varphi^{\text{in}}(z_n)) \dots (u - z_1 p_1 + \hbar z_1 \partial \varphi^{\text{in}}(z_1))\ket{0}_z \in \mathscr{F}^{\text{in}} \\
    \bra{0}_z \jmath^* \ket{\mathbb{V}_u} & = (u - \widetilde{z}_n \widetilde{p}_n - \hbar \widetilde{p}_n \partial \varphi^{\text{out}}(\widetilde{p}_n)) \dots (u - \widetilde{z}_1 \widetilde{p}_1 - \hbar \widetilde{p}_1 \partial \varphi^{\text{out}}(\widetilde{p}_1)) \ket{0}_w \in \mathscr{F}^{\text{out}}. 
\end{split}
\end{equation}
\end{theorem}
In what follows, the rest of the matrix $\mathbb{S}(u)$ will be determined.

\subsubsection{Pl\"{u}cker relations and fermion propagator}
A more geometric way to rephrase Proposition \ref{prop:taufnforS(u)} is that the map $G: \widetilde{M}(1) \to \text{Gr}(\mathscr{H}_{\text{in}} \oplus \mathscr{H}_{\text{out}})$ constructed above is equivariant with respect to each $\Gamma_+$-action, with the action on $\mathscr{H}_{\text{in}}$ being intertwined with the symplectomorphisms generated by Hamiltonians $H(t) = \sum_m t_m \tr(B_1^m)$ and the action on $\mathscr{H}_{\text{out}}$ being intertwined with symplectomorphisms generated by $\widetilde{H}(\tilde{t}) = \sum_m \tilde{t}_m \tr(B_2^m)$. 

It is clear that these $\Gamma_+$ actions on $\widetilde{M}(1)$ \textit{do not commute}, so one does not expect a simple formula like \eqref{formula:cylindermiuracoeff} for the general matrix coefficient of $\mathbb{S}(u)$. This phenomenon was observed, in a different language, in \cite{Aganagic_2005}. Nonetheless, there is the following workaround. 

Because $\mathbb{S}(u)$ arises from a state $\ket{\mathbb{V}_u}$ in the image of the Pl\"{u}cker embedding, it of course satisfies Pl\"{u}cker relations. It is a classical fact (see e.g. Chapter 9 of \cite{miwa2000solitons}) that Pl\"{u}cker relations characterize the image of the embedding $GL(\mathscr{H}) \hookrightarrow GL(\mathscr{F})$ where $\mathscr{H}$ is a vector space that may be finite or infinite dimensional (one has to take appropriate care in defining $GL(\mathscr{H})$ in the latter case), and $\mathscr{F}$ is the exterior algebra of $\mathscr{H}$. In our current setting, this idea manifests as the following 

\begin{lemma} \label{lemma:wickthm}
The instanton scattering matrix $\mathbb{S}(u)$ is uniquely determined by its vacuum matrix element $\bra{0} \mathbb{S}(u) \ket{0}$ and the object 
\begin{equation}
    G(\zeta, w) := \frac{\bra{0} \widetilde{\psi}(\zeta) \mathbb{S}(u) \psi(w) \ket{0}}{\bra{0} \mathbb{S}(u) \ket{0}} = \frac{\bra{0} \widetilde{\psi}^{in}(\zeta) \psi^{out}(w) \ket{\mathbb{V}_u}}{\braket{0}{\mathbb{V}_u}}.
\end{equation}
\end{lemma}

\begin{proof}
    $\mathbb{S}(u)$ is clearly determined by the vacuum matrix element and 
    \begin{equation}
        \frac{\bra{0} \widetilde{\psi}(\zeta_1) \dots \widetilde{\psi}(\zeta_N) \mathbb{S}(u) \psi(w_1) \dots \psi(w_N) \ket{0}}{\bra{0} \mathbb{S}(u) \ket{0}}
    \end{equation}
    for all $N$. Proposition \ref{prop:VmaptoGr} shows that over the open subset of $\widetilde{M}(1)$ where 
    \begin{equation}
        u \notin \text{Eigenvalues}(B_2 B_1) \cup \hbar \mathbb{Z}
    \end{equation}
    (which is precisely where we may invert $\bra{0} \mathbb{S}(u) \ket{0}$) the fiber of $\mathbb{V}_u$ determines a point in the big cell of $\text{Gr}(\mathscr{H}_{\text{in}} \oplus \mathscr{H}_{\text{out}})$, therefore $\mathbb{S}(u) \cdot (\bra{0} \mathbb{S}(u) \ket{0})^{-1}$ may be written as the exponential of a fermion bilinear over that locus. By Wick's theorem, the expression above is then written as 
    \begin{equation}
        \pm \det_{1 \leq i , j \leq N} G(\zeta_i, w_j).
    \end{equation}
\end{proof}

For physicists, this lemma may be rephrased in a more familiar language. The fact that the state $\ket{\mathbb{V}_u}$ arises from a point in a Grassmannian, not just an arbitrary vector in $\mathscr{F}^{\text{in}} \otimes \mathscr{F}^{\text{out}}$, means that the fermion theory on $zw = u$ is \textit{free}, although it is nonlocal. When $\mathbb{V}_u$ is in the big cell, the state may be written as an explicit Bogoliubov transform of the Fock vacuum. It is of course a familiar fact that any free field theory, not necessarily local, is completely determined by its unnormalized partition function and propagator. 

Now we will compute the propagator $G(\zeta; w)$. For simplicity we discuss it here in the rank $r = 1$ case; the higher rank case is no more difficult but we postpone its discussion to Section \ref{sect:monopolefrominstanton} to discuss its relation with monopole scattering matrices. 

\begin{theorem} \label{thm:fermionpropagator}
    For $\mathbb{S}(u)$ constructed in Propositions \ref{prop:VmaptoGr}, \ref{prop:S(u)vaccoeff}, the propagator is given by the explicit formula
    \begin{equation}
    \begin{split}
        G(\zeta, w) = \frac{\bra{0} \widetilde{\psi}(\zeta) \mathbb{S}(u) \psi(w) \ket{0}}{\bra{0} \mathbb{S}(u) \ket{0}} & = J \frac{w}{w - B_2} \frac{1}{u - B_2 B_1} \frac{1}{\zeta - B_1} I \\
        & + \Big( J \frac{1}{w - B_2} \frac{1}{\zeta - B_1} I + 1 \Big) \sum_{k = 0}^\infty \hbar^k \frac{\Gamma( \frac{u}{\hbar} + k +  1)}{\Gamma(\frac{u}{\hbar} + 1)} w^{-k} \zeta^{- k - 1}.
    \end{split}
    \end{equation}
\end{theorem}
Note that when $\hbar \neq 0$, the infinite series appearing is factorially divergent and must be understood only as a formal power series in $z, w$; this is a hallmark of the nonlocality of the operator $\mathbb{S}(u)$. If $\mathbb{S}(u)$ arose from the path integral of a local holomorphic field theory on a Riemann surface with two boundary components, this expression would extend analytically over that Riemann surface; instead $\mathbb{S}(u)$ arises from the path integral on the ``noncommutative curve'' $zw = u$. Happily, when $\hbar \to 0$ the series collapses to the expansion of $1/(\zeta - u/w)$, the expected behavior of the fermion propagator on the honest rational curve $zw = u$ with $\widetilde{\psi}$ inserted in the $z \to \infty$ patch and $\psi$ inserted in the $w \to \infty$ patch. 

We also note in passing that Proposition \ref{prop:VmaptoGr} implies that for each ADHM data $(B_1, B_2, I, J)$, $$\det_{1 \leq i, j \leq N} G(\zeta_i, w_j)$$ determines an explicit $\tau$-function of the Toda hierarchy in Miwa coordinates $(\zeta_i, w_j)$. 

\begin{proof}
We start by considering the following as a function of $z$:
\begin{equation}
     \frac{\bra{0} \widetilde{\psi}(\zeta) \psi(z) \mathbb{S}(u) \ket{0}}{\bra{0} \mathbb{S}(u) \ket{0}} = \frac{\bra{0} \widetilde{\psi}^{in}(\zeta) \psi^{in}(z) \ket{\mathbb{V}_u}}{\braket{0}{\mathbb{V}_u}}.
\end{equation}
From definitions and the description of $\mathbb{V}_u$ in Appendices \ref{subsect:explicitextforzw}, \ref{subsect:functlzw}, \ref{proof:VmaptoGr}, this expression is determined to be of the form
\begin{equation}
    \frac{\bra{0} \widetilde{\psi}^{in}(\zeta) \psi^{in}(z) \ket{\mathbb{V}_u}}{\braket{0}{\mathbb{V}_u}} = J \frac{1}{z - B_1} \phi_z( \psi_+) + \phi_z(\xi)
\end{equation}
where $\phi_z$ is the morphism from \ref{subsect:functlzw} and $(\psi_+, \psi_-, \xi)$ are the by now familiar triple from Propositions \ref{extgroups}, \ref{prop:explicitextforzw}. The insertion of $\bra{0} \widetilde{\psi}(\zeta)$ uniquely fixes 
\begin{equation}
     \xi = \frac{1}{\zeta - z} = \sum_{k = 0}^\infty z^k \zeta^{- k - 1} \in \widehat{\mathscr{O}}^\ell_{C_u}[[\zeta^{-1}]].
\end{equation}
It is understood that $\phi_z$ is applied to $\xi$ by expansion at $\zeta \to \infty$ and linearity. As we assume $u - B_2 B_1$ is invertible, \eqref{eq:zwextsoln} determines the rest of the triple to be 
\begin{equation}
\begin{split}
    \psi_+ & = -B_1 \frac{1}{u - B_2 B_1} \frac{1}{\zeta - B_1} I \\
    \psi_- & = \frac{1}{u - B_2 B_1} \frac{1}{\zeta - B_1} I w + \frac{1}{\zeta - z}\frac{1}{ \zeta - B_1} I.
\end{split}
\end{equation}
Then again by the description of the map $\mathbb{V}_u \to \mathscr{O}_{\widetilde{M}(r)} \otimes ( \mathscr{H}_{\text{in}} \oplus \mathscr{H}_{\text{out}})$ in Appendix \ref{proof:VmaptoGr}, we must have 
\begin{equation}
    G(\zeta, w) = J \frac{1}{w - B_2} \phi_w(\psi_-) + \phi_w(\xi).
\end{equation}
Noting that 
\begin{equation}
    \phi_w \Big( \frac{1}{\zeta - z} \Big) = \sum_{k = 0}^\infty \zeta^{-k - 1} \phi_w(z^k) = \sum_{k = 0}^\infty \hbar^k \frac{\Gamma( \frac{u}{\hbar} + k + 1)}{\Gamma(\frac{u}{\hbar} + 1)} \zeta^{-k - 1} w^{-k}
\end{equation}
gives the result. 
\end{proof}
By Lemma \ref{lemma:wickthm}, Propositions \ref{prop:VmaptoGr}, \ref{prop:S(u)vaccoeff} and Theorem \ref{thm:fermionpropagator} give a complete determination of the operator $\mathbb{S}(u)$. It would be interesting to investigate its properties further; we will make some remarks and conjectures to this effect at present. 

\subsubsection{$\mathbb{S}(u)$ as a semiclassical $R$-matrix for $\mathsf{Y}_{+1}(\widehat{\mathfrak{gl}}_1)$}
By its very construction as an instanton scattering matrix, $\mathbb{S}(u)$ is expected to be an $R$-matrix for a Poisson limit of $\mathsf{Y}_{+1}(\widehat{\mathfrak{gl}}_1)$. The most basic and essential property one would hope to verify is that the Poisson bracket of $\widetilde{M}(r)$ on matrix elements of $\mathbb{S}(u)$ can be recast as a semiclassical Yang-Baxter equation with $R_{\mathsf{MO}}$ from \cite{mo}, as it was in the monopole case \eqref{eq:monopoleRTT}.  

We have not yet attempted to verify this, for the following reason. As is already evident in the discussion in Section \ref{sect:monopoleyangian} in the semiclassical limit, the monopole scattering matrices \textit{do not} satisfy the Yang-Baxter relation on the nose in the case of dominantly shifted Yangians; their failure to do so was captured by the ``$\Delta$-term''. One certainly expects similar behavior for the dominantly shifted affine Yangians, and for this reason it is unlikely that $\mathbb{S}(u)$ solves the most naive form of the Yang-Baxter equation. Without a general framework to understand the corrections to the Yang-Baxter equation arising for shifted Yangians, it seems premature to attempt to compute those corrections in the special case of the equation satisfied by $\mathbb{S}(u)$ (the meaning of the corrections is already unclear in the much simpler case of shifted Yangians of $\mathfrak{gl}_r$). 

While essentially nothing useful is known about the hypothetical ``corrected Yang-Baxter equation'', one thing is clear: the correction terms always involve \textit{positive} powers of the spectral variable $u$, as we saw with the $\Delta$-term in \eqref{eq:monopoleRTT}. It is heartening to observe that in the above explicit formula for $G(\zeta, w)$, matrix elements of $\mathbb{S}(u)$ involve arbitrarily large positive powers of $u$ appearing as we go deeper into the $\zeta, w$ expansion, while the part of the matrix elements of negative degree in $u$ are essentially uniform and controlled by the poles at $\det(u - B_2 B_1) = 0$ dictated by the general construction of Proposition \ref{prop:VmaptoGr}. 

\subsubsection{$G(\zeta, w)$ and the vector-vector $R$-matrix} \label{subsect:vectRmatrixconj}
$\mathsf{Y}_{+1}(\widehat{\mathfrak{gl}}_1)$ has a representation called the vector representation, such that the Fock representation can be thought of as roughly the semi-infinite wedge of the vector representation (see e.g. \cite{Tsymbaliuk_2017}, \cite{Feigin_2011} for the unshifted case). To this author's knowledge, no formulas for $R$-matrices between a pair of vector representations are known, even in the unshifted case.

It is natural to conjecture based on Theorem \ref{thm:fermionpropagator} that the following quantization of the fermion propagator in the background of $n = 1$ instanton is an $R$-matrix between a pair of vector representations (see the next paragraph for clarifications)
\begin{equation}
    \widehat{G}(\zeta, w; u) = \frac{w}{w - x} \frac{1}{u - x \partial_x} \frac{1}{\zeta - \partial_x} + \Big( \frac{1}{w - x} \frac{1}{\zeta - \partial_x} + 1\Big) \sum_{k = 0}^\infty \frac{\Gamma(u + k + 1)}{\Gamma(u + 1)} w^{-k} \zeta^{-k - 1}
\end{equation}
where we have set all parameters to $1$ for convenience. 

This formula is to be read as an $R$-matrix as follows: extracting a coefficient of $\zeta^{- m - 1} w^{-n}$, one obtains a formal Laurent series in $u^{-1}$ with coefficients in differential operators in the $x$-coordinate. One expects that the differential operators arising in this expansion generate the action of $\mathsf{Y}_{+1}(\widehat{\mathfrak{gl}}_1)$ in the vector representation furnished by an appropriate class of functions of the $x$-coordinate. Conversely, applying this expression to a test function $f(x)$ and pairing with some other function $g(x)$, one may read the coefficient $(g(x), \widehat{G}(\zeta, w; u) \cdot f(x))$ as a Laurent series in $u^{-1}$ with coefficients in integral kernels for operators on functions of either $\zeta$ or $w$; these integral operators should once again generate the $\mathsf{Y}_{+1}(\widehat{\mathfrak{gl}}_1)$ action in a vector module furnished by functions of $\zeta$ or $w$. Here by ``integral'' we mean formal extraction of residue. 

$\widehat{G}$ should also solve some ``Yang-Baxter'' equation, likely to receive correction terms of positive degree in the spectral variable $u$. 

\subsection{Monopole scattering matrices from $\mathbb{S}(u)$} \label{sect:monopolefrominstanton}
In this final section, the circle of ideas will be closed and it will be explained that monopole scattering matrices (at least for certain values of the charges $\lambda$ and $\mu$ in the notation of Appendix \ref{reviewscatter}) can be deduced from $\mathbb{S}(u)$. This shows that all semiclassical $R$-matrices discussed in this paper can be put into one coherent framework. In this section we will need to assume the rank $r > 1$ in order to get interesting monopole scattering matrices.

\subsubsection{Mathematical motivation}
Mathematically, the strategy to understand this follows. Moduli spaces of $U(r)$ monopoles $\mathscr{\widetilde{M}}^{- \lambda}_{-\mu}$ arise as fixed loci of a $\mathbb{C}^\times$-action on $\widetilde{M}(n, r)$ when 
\begin{equation}
    n = \frac{1}{2} (\norm{\lambda}^2 - \norm{\mu}^2). 
\end{equation}
Differential-geometrically, this goes back to an observation of Kronheimer \cite{kronheimer}, and the $\mathbb{C}^\times$-action is viewed as a complexification of a natural $S^1$-action induced by rotation of the fiber of $T^* \mathbb{C} \to \mathbb{R}^3$ (we use here that the complex manifold structure of instanton moduli on $T^* \mathbb{C}$ and a Taub-NUT space are the same). A version of this result that may be more familiar to representation theorists is that in type $A$, transversal slices in the affine Grassmannian an be exhibited as Nakajima quiver varieties see e.g. \cite{mirkovic2002}, \cite{Braverman_2010}. 

The proof of Theorem \ref{thm:fermionpropagator} can be generalized to compute the propagator in the case $r > 1$, which now carries indices as $G\indices{_\alpha^\beta}(\zeta; w)$, $\alpha, \beta = 1, \dots, r$. Restricting $G$ to the fixed locus in $\widetilde{M}(n, r)$, it can be shown to satisfy an equivariance property with respect to the $\mathbb{C}^\times$-action. It may be decomposed using the weight decomposition on $\mathbb{C}[[\zeta^{-1}, w^{-1}]]$ under the naturally induced $\mathbb{C}^\times$-action, and in some appropriate convention its zero weight component should be the monopole scattering matrix/matrix description of affine Grassmannian slice; this is the content of \eqref{eq:Smatrixfromprop} below.

Formula \eqref{eq:Smatrixfromprop} gives, in very concrete terms, an explicit identification between the matrix description of the affine Grassmannian slices (in the terminology of \cite{bfnslice}, \cite{krylov}) and their presentation as type $A$ Nakajima varieties. This is done by exhibiting the monopole scattering matrix coefficients as explicit elements in the invariant ring of the Nakajima variety. It is expected that \eqref{eq:Smatrixfromprop} implements the isomorphism of \cite{mirkovic2002}, directly in the matrix description of slices from \cite{bfnslice}. Note we also get the affine deformation to $\hbar \neq 0$ for free.

Strictly speaking, we do not give a complete mathematical proof that \eqref{eq:Smatrixfromprop} works; the difficult part is to show that there are no more relations in the invariant ring of the quiver variety than the ones implied by the conditions imposed on $S$-matrix coefficients. This seems likely to be true based on experience with the monopole case, but it will not be demonstrated here. It will, however, be illustrated in specific examples. The conjecture can likely be reduced to proving that the morphism constructed in Proposition \ref{prop:VmaptoGr} above is an embedding.

\subsubsection{Physical motivation} 
Let us re-explain the motivation and plan of analysis from the physical perspective; the idea is just to reverse the steps of the argument in Section \ref{sect:monopoletoinst} leading from the monopole scattering matrices to the instanton ones. 

Ignoring the noncommutativity (or just setting it to $0$, which does less harm in the rank $r > 1$ case), the operator $\mathbb{S}(u)$ arises from the chiral fermion path integral on the curve $\{ zw = u \} \subset T^* \mathbb{C}$. This curve is the fiber over a vertical line of fixed $u$ under the map $T^* \mathbb{C} \to \mathbb{R}^3 \simeq \mathbb{R} \times \mathbb{C}$, the copy of $\mathbb{C}$ being the $u$-plane. If we suppose our background instanton configuration is actually $S^1$-invariant for the circle action rotating the fibers, the fermion problem gains a $U(1)$ symmetry. On the other hand, such an $S^1$-invariant instanton configuration is the same as a monopole configuration with a certain Dirac singularity, by Kronheimer's observation \cite{kronheimer}. 

Strictly speaking, to define the $S^1$-action on the moduli space of instantons, we need to choose in addition a lift of the gauge bundle $\mathscr{E}$ to an $S^1$-equivariant vector bundle; this determines the asymptotic monopole charge $\mu$ and the amount of singular monopole charge $\lambda$ at the origin. 

We may use this $U(1)$ action to decompose the fermions on $zw = u$ into Kaluza-Klein modes along the $S^1$ fiber; reversing the construction of Section \ref{sect:monopoletoinst}, the monopole scattering matrix coefficients are determined precisely from the contributions of the zero mode to this path integral. Actually, these are the scattering matrix coefficients, in general, in the exterior algebra $\Lambda^\bullet(\mathbb{C}^r)$ and the conventional $S$-matrix is the induced map on the component of fermion number $1$. Since the fermion theory is free and the KK modes are decoupled from one another, this can be read off precisely from the fermion propagator $G(\zeta; w)$: the propagator breaks up into an infinite collection of decoupled propagators, one for each KK mode, and the one for the zero mode is precisely the monopole scattering matrix.

It is important here that $G(\zeta; w)$ is understood in an asymptotic expansion for $z \to \infty$, $w \to \infty$, as the scattering matrix comes from the expectation value $\langle \psi^\alpha(y = T/2) \widetilde{\psi}_\beta(y = -T/2) \rangle$, $\alpha, \beta = 1, \dots, r$ in the fermionic quantum mechanics only in the limit $T \to \infty$, by the definition of scattering matrices in Appendix \ref{reviewscatter}. 

\subsubsection{Monopole moduli as $\mathbb{C}^\times$-fixed locus} \label{subsect:C*fixed}
Let us review the results of \cite{Braverman_2010}, \cite{kronheimer}, \cite{mirkovic2002} in the form useful for us. Consider the affine Nakajima variety $\widetilde{M}(n, r)$. Choose a homomorphism $\mu: \mathbb{C}^\times \to GL_r$. Without loss of generality we may assume $\mu$ lands in the maximal torus and write 
\begin{equation}
    \mu(t) = \text{diag}(t^{\mu_\alpha })_{\alpha = 1}^n 
\end{equation}
for some integers $\mu_1 \geq \dots \geq \mu_r$, and $t \in \mathbb{C}^\times$.

This defines a $\mathbb{C}^\times$-action on $\widetilde{M}(n, r)$: $t \in \mathbb{C}^\times$ acts by 
\begin{equation}
    t \cdot(B_1, B_2, I, J) = (t B_1, t^{-1} B_2, I \mu(t), \mu^{-1}(t) J). 
\end{equation}
Note the $\mathbb{C}^\times$ action is unchanged under shifts of $\mu$ by a cocharacter of the center of $GL_r$; thus without loss of generality we may assume $\mu_1 = 0$. 

As usual, the fixed locus is a disjoint union of $A_\infty$ quiver varieties: 
\begin{equation} \label{eq:connectecomponents}
    \widetilde{M}(n, r)^{\mathbb{C}^\times} = \bigsqcup_{\sigma/\sim } \widetilde{M}_\sigma
\end{equation}
where the union is over conjugacy classes of $\sigma: \mathbb{C}^\times \to GL_n$ solving \begin{equation} \label{eq:C*equivariance}
\begin{split}
    t B_1 & = \sigma(t) B_1 \sigma^{-1}(t) \\
    t^{-1} B_2 & = \sigma(t) B_2 \sigma^{-1}(t) \\
    I \mu(t) & = \sigma(t) I \\
    \mu^{-1}(t) J & = J \sigma^{-1}(t)
\end{split}
\end{equation}
and, for given $\sigma$,
\begin{equation}
    \widetilde{M}_\sigma := \{ (B_1,B_2, I, J) | \text{$\comm{B_1}{B_2} + IJ = \hbar$ and \eqref{eq:C*equivariance}} \}/G_\sigma
\end{equation}
where $G_\sigma$ is the centralizer of $\sigma$ in $GL_n$. Each $\widetilde{M}_\sigma$ is a type $A_N$ quiver variety, where $N = \#(\text{GL factors in $G_\sigma$})$. Specializing to $\hbar = 0$, we have an analogous decomposition of the fixed locus in the Uhlenbeck moduli space $M_0(n, r)$ (a singular affine variety) into connected components that we denote $M_{0, \sigma}$. 

Now observe that on the fixed locus in instanton moduli the fiber of the universal sheaf $\mathscr{E}$ over $0$ in $T^* \mathbb{C}$ becomes a $\mathbb{C}^\times$-representation, and from the ADHM complex (see Appendix \ref{adhmdetails}) its character is computed via 
\begin{equation} \label{eq:chargeat0}
    \chi\Big(\mathscr{E} \eval_0 \Big) = \tr_W(\mu(t)) - (1 - t)(1 - t^{-1}) \tr_V(\sigma(t)). 
\end{equation}
At given $t$, this is a locally constant function on the fixed locus, and we may write 
\begin{equation}
    \chi \Big( \mathscr{E} \eval_0 \Big) = \sum_{\alpha = 1}^r t^{\lambda_\alpha}
\end{equation}
for integers $\lambda_\alpha$ in the order $\lambda_1 \geq \dots \geq \lambda_r$; we regard these integers as another cocharacter $\lambda$. $\sigma$ determines $\lambda$ and vice versa, up to conjugacy. Note that $\mu$ and $\lambda$ are not related arbitrarily: they must differ by a coroot and satisfy e.g. $\sum_\alpha \lambda_\alpha^2 \geq \sum_\alpha \mu_\alpha^2$. The assertion of \cite{mirkovic2002} can be translated to the claim that
\begin{equation} \label{eq:kronheimeriso}
    M_{0, \sigma} \simeq \overline{\mathscr{M}}^{-\lambda}_{-\mu}
\end{equation}
where $\sigma$, $\mu$ are related by \eqref{eq:chargeat0}, and the monopole moduli space/affine Grassmannian slice $\overline{\mathscr{M}}^{-\lambda}_{-\mu} = \overline{\mathscr{W}}^{-\lambda}_{-\mu}$ is as in Appendix \eqref{eq:genslicedef}; the isomorphism is on the level of affine algebraic varieties. In this form, the result is essentially Kronheimer's observation \cite{kronheimer} via the dictionary between monopole moduli and affine Grassmannian slices initiated in \cite{KW}. 

\textit{Remark on conventions.} To reduce the headache of readers, we should remark that it is customary in the literature to define $\overline{\mathscr{W}}^\lambda_\mu$ for \textit{dominant} $\lambda$. Using \eqref{eq:genslicedef}, the definition makes sense for arbitrary $\lambda$. Acting by the Weyl group on $(\lambda, \mu)$ as well as simultaneous sign reversal $(\lambda, \mu) \mapsto (-\lambda, -\mu)$ does not change the underlying variety, but it does change the scattering matrix/embedding into $GL_r(u)$ via \eqref{eq:genslicedef}: Weyl group action conjugates the scattering matrix by the corresponding permutation matrix and sign reversal replaces the scattering matrix by its inverse. We introduce the sign here to avoid taking matrix inverse in \eqref{eq:Smatrixfromprop} below. 

\textit{Remark on affine deformation.} It is easy to see that the fixed loci $\widetilde{M}_\sigma \simeq \widetilde{\mathscr{M}}^{\lambda}_{\mu}(\hbar)$ correspond to affine deformations of the monopole moduli space in which the Dirac singularities are placed at positions in the $u$-plane which are integer multiples of $\hbar$. This correlates nicely with the poles in $\mathbb{S}(u)$ at $u \in \hbar \mathbb{Z}$ from Proposition \ref{prop:VmaptoGr}, in accordance with our general picture. 

\subsubsection{Fermion propagator in higher rank and equivariance}
In Section \ref{sect:determineS}, the fermion propagator $G(\zeta; w)$ was computed in the rank $r = 1$ case. The proof of Theorem \ref{thm:fermionpropagator} generalizes immediately to the rank $r > 1$ case. The only difference is that now the operator $\mathbb{S}(u)$ corresponds to a tensor in $\mathscr{F}^{\text{in}} \otimes \mathscr{F}^{\text{out}}$, where $\mathscr{F}^{\text{in/out}}$ are $r$-component fermionic Fock spaces. We denote the fermions in components as $\psi^\alpha(w)$, $\widetilde{\psi}_\alpha(\zeta)$, with $\alpha = 1, \dots r$. The index $\alpha$ corresponds to choosing a basis in the framing space $W$; correspondingly denote $J^\alpha \in V^*$ and $I_\alpha \in V$ the components of the ADHM data $I, J$. 

Then we have the following generalization of Theorem \ref{thm:fermionpropagator} for the rank $r > 1$ case:
\begin{prop} \label{prop:highrankpropagator}
\begin{equation}
\begin{split}
    G\indices{_\alpha^\beta}(\zeta; w) := \frac{\bra{0} \widetilde{\psi}_\alpha(\zeta) \mathbb{S}(u) \psi^\beta(w) \ket{0}}{\bra{0} \mathbb{S}(u) \ket{0}} & = J^\beta \frac{w}{w - B_2} \frac{1}{u - B_2 B_1} \frac{1}{\zeta - B_1} I_\alpha \\
    & + \Big( J^\beta \frac{1}{w - B_2} \frac{1}{\zeta - B_1} I_\alpha + \delta^\beta_\alpha\Big) \sum_{k = 0}^\infty \hbar^k \frac{\Gamma(\frac{u}{\hbar} + k + 1)}{\Gamma(\frac{u}{\hbar} + 1)} w^{-k} \zeta^{-k - 1}
\end{split}
\end{equation}
where $\delta^\beta_\alpha$ is the Kronecker delta symbol.
\end{prop}
From this formula, it is immediate that if we restrict $(B_1, B_2, I, J)$ to the fixed locus $\widetilde{M}(r)^{\mathbb{C}^\times} \subset \widetilde{M}(r)$ of a $\mathbb{C}^\times$ action defined by a choice of cocharacter $\mu$ as in Section \ref{subsect:C*fixed}, that $G\indices{_\alpha^\beta}(\zeta; w)$ satisfies the following equivariance property: 
\begin{lemma} \label{lemma:propequivariance}
\begin{equation}
    G\indices{_\alpha^\beta}(t \zeta; t^{-1} w) \eval_{\widetilde{M}(r)^{\mathbb{C}^\times}} = t^{\mu_{\alpha} - \mu_\beta - 1} G\indices{_\alpha^\beta}(\zeta; w) \eval_{\widetilde{M}(r)^{\mathbb{C}^\times}}.
\end{equation}
\end{lemma}

\begin{proof}
    By Proposition \ref{prop:highrankpropagator} and \eqref{eq:C*equivariance}. 
\end{proof}

This means that with respect to the naturally induced $\mathbb{C}^\times$-action on $$\text{End}(W) \otimes H^0 \Big(\widetilde{M}(r)^{\mathbb{C}^\times}, \mathscr{O}_{\widetilde{M}(r)^{\mathbb{C}^\times}} \Big)(u) \otimes \mathbb{C}[[\zeta^{-1}, w^{-1}]],$$ $G\indices{_\alpha^\beta}(\zeta; w)$ belongs to the subspace of weight $-1$. In view of the discussion above, it seems natural to propose the following. 
\begin{conj}
Define an $r \times r$ matrix 
\begin{equation} \label{eq:Smatrixfromprop}
    S\indices{_\alpha^\beta}(u) := \textnormal{Coeff}[w^{\mu_\beta} \zeta^{\mu_\alpha - 1}] \Bigg( G\indices{_\alpha^\beta}(\zeta; w) \eval_{\widetilde{M}(r)^{\mathbb{C}^\times}}\Bigg). 
\end{equation}
Its entries are rational functions on $\mathbb{A}^1_u \times \widetilde{M}(r)^{\mathbb{C}^\times}$ with poles along the locus $\det(u - B_2 B_1) = 0$. Recall that in the conventions of Section \ref{subsect:C*fixed}, $\mu_r \leq \dots \leq \mu_1 = 0$. Then $S\indices{_\alpha^\beta}(u)$ restricts on each connected component of \eqref{eq:connectecomponents} to the monopole scattering matrix for $\widetilde{\mathscr{M}}^{-\lambda}_{-\mu}$/matrix description of affine Grassmannian slice $\widetilde{\mathscr{W}}^{-\lambda}_{-\mu}$ under the isomorphism \eqref{eq:kronheimeriso}.
\end{conj}
Note it is important here that we assume $\mu_1 = 0$, even though the $\mathbb{C}^\times$-action itself is independent of this choice, otherwise some of the entries in $S(u)$ would be identically zero for trivial reasons. The need to do this is traced back to the fact that the Sato Grassmannian $\text{Gr}(\mathscr{H}_{\text{in}} \oplus \mathscr{H}_{\text{out}})$ depends on a choice of attracting direction for the $\mathbb{C}^\times$-action on $T^* \mathbb{C}$. The conjecture should be true as long as $\mu_1 \leq 0$. It would be interesting to more closely investigate this appearance of the translation automorphism of the $\mathfrak{gl}_\infty$ quiver.

If it was known that the morphism $\widetilde{M}(n, r) \to \text{Gr}(\mathscr{H}_{\text{in}} \oplus \mathscr{H}_{\text{out}})$ constructed in Proposition \ref{prop:VmaptoGr} was an embedding, this conjecture would likely follow from Lemma \ref{lemma:propequivariance}. It is amusing to observe that this conjecture implies that the scattering matrix is in some sense independent of $\lambda$. This conjecture will now be illustrated in simple, illustrative examples. 

\subsubsection{Example: when $\widetilde{\mathscr{M}}^{-\lambda}_{-\mu}$ is a coadjoint orbit}
The easiest instance of \eqref{eq:Smatrixfromprop} to understand is the special case when $\widetilde{\mathscr{M}}^{-\lambda}_{-\mu}$ is isomorphic to a coadjoint orbit in $\mathfrak{gl}_r^*$. From the discussion in Section \ref{subsect:C*fixed}, this happens when $\mu = 0$ and 
\begin{equation}
    \tr_V \sigma(t) = n_0 + n_{-1} t^{-1} + \dots + n_{-(\ell - 1)} t^{-(\ell - 1)}. 
\end{equation}
where $r > n_0 > n_{-1} > \dots > n_{-(\ell - 1)} \geq 0$. The corresponding $\widetilde{\mathscr{W}}^{-\lambda}_{-\mu}$ is an affine deformation of the cotangent bundle to the partial flag variety describing flags $V_{-(\ell - 1)} \subset \dots V_{-1} \subset V_0 \subset \mathbb{C}^r$ of length $\ell$, with $\dim V_{-i} = n_{-i}$. $V_{-i}$ corresponds to the $-i$ weight space of $V$ under $\mathbb{C}^\times$.

On the fixed locus, $B_1: V_{-i - 1} \to V_{-i}$, in particular $B_1 \eval_{V_0} = 0$. Then from \eqref{eq:Smatrixfromprop}, 
\begin{equation}
    S\indices{_\alpha^\beta}(u) = \delta^\beta_\alpha + J^\beta \frac{1}{u - B_2 B_1 \eval_{V_0}} I_\alpha = \delta^\beta_\alpha + \frac{1}{u} J^\beta I_\alpha. 
\end{equation}
It is well-known that this is the correct scattering matrix for a coadjoint orbit, because $J^\beta I_\alpha$ coincides with the $\mathfrak{gl}_r$-moment map in the quiver description, and the conditions \eqref{eq:genslicedef} imposed on the scattering matrix fix it to be of the form $1 + u^{-1} \boldsymbol{\mu}$, where $\boldsymbol{\mu}$ is the $\mathfrak{gl}_r$ moment map. This is also the semiclassical $R$-matrix for $\mathsf{Y}(\mathfrak{gl}_r)$, which braids an $r$-dimensional module with a module arising from quantization of $T^*(GL_r/P)$, where the parabolic subgroup $P$ is determined by $\lambda$, equivalently $\sigma$. Note that 
\begin{equation}
\begin{split}
    \det S(u) & = \frac{1}{u^r} \text{det}_W(u + JI) \\
              & = u^{-n_0}  \text{det}_{V_0} (u + IJ) \\
              & = u^{-n_0} \text{det}_{V_0}(u + \hbar -B_{1, -1} B_{2, 0}) \\
              & = u^{-n_0} (u + \hbar)^{n_0 - n_{-1}} \text{det}_{V_{-1}}(u + \hbar - B_{2, 0} B_{1, - 1}) \\
              & = \dots \\
              & = \prod_{j = 0}^\ell (u + j \hbar)^{n_{-(j - 1)} - n_{-j}}
\end{split}
\end{equation}
in accordance with the poles of $\mathbb{S}(u)$ along $u \in \hbar \mathbb{Z}$ in general. 

\subsubsection{Example: the $A_2$ surface}
The simplest example where positive powers of $u$ enter (i.e. we get a dominantly shifted Yangian) is that of the $A_2$ surface. This corresponds to the choice $r = 2$, $\mu = (0, -1)$ and $\lambda = (1, -2)$ in our present convention. Write the components of $I, J$ as $I_0 \in V_0$, $I_{-1} \in V_{-1}$ and similarly for $J^0$, $J^{-1}$. The only nonzero components of $B_1, B_2$ are $B_{1, -1} :V_{-1} \to V_0$, $B_{2, 0} : V_0 \to V_{-1}$, and $\dim V_0 = \dim V_{-1} = 1$. 

We calculate from \eqref{eq:Smatrixfromprop} 
\begin{equation}
\begin{split}
    S\indices{_1^1}(u) & = J^0 \frac{1}{u - B_2 B_1 \eval_{V_0}} I_0 + 1 = 1 + \frac{1}{u} J^0 I_0 \\
    S\indices{_1^2}(u) & = J^{-1} B_{2, 0} \frac{1}{u - B_2 B_1\eval_{V_0}} I_{0} = \frac{1}{u} J^{-1} B_{2, 0} I_0 \\
    S\indices{_2^1}(u) & = J^0 \frac{1}{u - B_2 B_1 \eval_{V_0}} B_{1, -1} I_{-1} = \frac{1}{u} J^0 B_{1, -1} I_{-1} \\
    S\indices{_2^2}(u) & =  \frac{1}{u} J^{-1} B_{2, 0} B_{1, -1} I_{-1} + J^{-1} I_{-1} + u + \hbar.
\end{split}
\end{equation}
One readily sees that the matrix coefficients generate the invariant ring of the $A_2$ quiver and that the ADHM relations 
\begin{equation}
\begin{split}
    B_{1, -1} B_{2, 0}  + I_0J^0 & = \hbar \\
    -B_{2, 0} B_{1, -1} + I_{-1} J^{-1} & = \hbar
\end{split}
\end{equation}
imply 
\begin{equation}
    \det S(u) = \frac{1}{u}(u + \hbar)(u + 2\hbar).
\end{equation}
This has one pole and two zeroes, as dictated by $-\lambda = (-1, 2)$. Taking coefficients in $u$ one finds the usual description of an affine deformation of the $A_2$ surface. 

Clearly, for further dominantly shifted Yangians of $\mathfrak{gl}_2$, more terms from the infinite series in $G\indices{_\alpha^\beta}(\zeta; w)$ will contribute to $S\indices{_2^2}$. 

\subsubsection{Example: 3d mirror to $T^*Gr(n, 2n)$}
The most nontrivial family of examples where it is straightforward to check that \eqref{eq:Smatrixfromprop} works are the varieties related by 3d mirror symmetry to $T^* \text{Gr}(n, 2n)$. These correspond to $r = 2$, $\mu = 0$, $\lambda = (n, -n)$. In this case $\lambda$ determines that 
\begin{equation}
    \tr_V \sigma(t) = n + (n - 1)(t + t^{-1}) + \dots + t^{n - 1} + t^{1 - n}. 
\end{equation}
$B_1$ breaks into the maps $B_{1, i} : V_i \to V_{i + 1}$ and $B_2$ into the maps $B_{2, i} : V_i \to V_{i - 1}$. The moment map conditions are 
\begin{equation}
\begin{split}
    B_{1, -1} B_{2, 0} - B_{2, 1} B_{1, 0} + \sum_\alpha I_\alpha J^\alpha & = \hbar \\
    B_{1, i - 1} B_{2, i} - B_{2, i + 1} B_{1, i} & = \hbar, \, \, \, i \neq 0.
\end{split}
\end{equation}
From \eqref{eq:Smatrixfromprop} we read off
\begin{equation}
    S\indices{_\alpha^\beta}(u) = \delta^\beta_\alpha + J^\beta \frac{1}{u - B_{2, 1} B_{1, 0}} I_\alpha
\end{equation}
for $\alpha, \beta = 1, 2$. Clearly $S(u) = 1 + O(u^{-1})$ for $u \to \infty$, ensuring it satisfies the Gauss factorization condition from \eqref{eq:genslicedef}. To check the second condition in \eqref{eq:genslicedef} it suffices to examine the determinant 
\begin{equation}
    \det S(u) = \frac{\det(u - B_{2, 1} B_{1, 0} + IJ)}{\det(u - B_{2, 1} B_{1, 0})}.
\end{equation}
Repeatedly applying the moment map conditions and using that $\det(1 + AB) = \det (1 + BA)$ for rectangular matrices $A, B$ gives 
\begin{equation}
\begin{split}
    \det(u - B_{2, 1} B_{1, 0}) & = u (u - \hbar) \dots (u - (n -1) \hbar) \\
    \det(u - B_{2, 1} B_{1, 0} + IJ) & = (u + \hbar) \dots (u + n\hbar).
\end{split}
\end{equation}
We see that $\det S(u)$ is a rational function with $n$ zeroes and $n$ poles at integer multiples of $\hbar$, in accordance with $-\lambda = (-n, n)$ and general features of $\mathbb{S}(u)$.

\newpage 

\begin{appendices}

\section{Review of monopole scattering} \label{reviewscatter}
The scattering transform for monopoles is a construction in classical gauge theory studied by Atiyah, Hitchin, Hurtubise, Jarvis, and many others \cite{atiyahhitchin}, \cite{hitchinspectral}, \cite{hurtubise}, \cite{jarvis}. To orient the reader, we will give an introduction to monopole scattering in this appendix. There is nothing essentially new in this appendix, thought it has become clear that a careful examination of the monopole scattering transform may contain valuable insights in constructing a general $RTT$ formalism for shifted Yangians, so we have spelled out as many details about the Poisson structure on the monopole moduli space (the semiclassical imprint of the Yangian) as we could manage. 

We will study the case of the gauge group $G = U(r)$. In this appendix, we give a purely formal treatment, i.e. we do not address subtle analytic issues. 

\subsection{Monopole moduli spaces: basic definitions}
To discuss moduli spaces of monopoles appropriately requires a few preparations and definitions, which we give in this section. 

\subsubsection{Setup and notations}
We work on $\mathbb{R}^3$ with Cartesian coordinates $(x^1, x^2, x^3)$ and the standard Euclidean metric. We will also have occasion to use spherical coordinates, which we denote by $(r, \theta, \phi)$. We study rank $r$ unitary vector bundles $E \to \mathbb{R}^3$, with unitary connection denoted by $A$. In local coordinates, $A$ is given by a one-form $A_i$ valued in $r \times r$ antihermitian matrices, with indices $i = 1, 2, 3$. We also introduce a scalar field $\phi$ which is a real-valued section of the bundle $\text{End}(E)$ (real-valued means valued in the Lie algebra of $U(r)$). 

A pair $(A, \phi)$ is called a monopole configuration if it solves the equation 
\begin{equation} \label{bpsmonopole}
    F_A + \star D_A \phi = 0
\end{equation}
relating the curvature $F_A = dA + A \wedge A$ and the covariant derivative $D_A \phi = d\phi + \comm{A}{\phi}$. Here, $\star$ denotes the Hodge star operation on forms induced via the Euclidean metric on $\mathbb{R}^3$. This is a system of nonlinear first order partial differential equations for the components of $A$ and $\phi$. Note, however, that it becomes linear if the gauge group is replaced by its maximal torus; in more invariant terms, if we have a preferred splitting of $E$ as a direct sum of line bundles and we require $(A, \phi)$ to preserve this splitting. 

We would like to study the moduli space of solutions to the monopole equation, considering two solutions equivalent if they are related by a gauge transformation (smooth automorphism of $E$) which approaches the identity as $r \to \infty$. In order to get a reasonable moduli space, we must specify a boundary condition on $(A, \phi)$ as $r \to \infty$. To carefully state the boundary condition requires a small digression. 

\subsubsection{Dirac monopoles}
The Dirac monopole configuration is a reference configuration defined on $\mathbb{R}^3 \setminus \{ 0 \}$, and depending on a choice of cocharacter $\rho: \mathbb{C}^\times \to  (\mathbb{C}^\times)^r \hookrightarrow GL_r$ that we view as a collection of $r$ integers $(m_1, \dots, m_r)$. Let $m = \text{diag}(m_1, \dots, m_r)$ denote the corresponding diagonal matrix. $m$ is called the charge of the Dirac monopole. 

The underlying vector bundle of the Dirac monopole is topologically non-trivial and may be described abstractly as follows. Let $\pi: \mathbb{R}^3 \setminus \{ 0 \} \to S^2$ denote the quotient by translations in the radial direction, and view $S^2 \simeq \mathbb{CP}^1$. Denote by $\mathscr{O}(k)$ the unique up to isomorphism line bundle on $\mathbb{P}^1$ of degree $k$; it may be viewed as associated to the Hopf bundle $S^3 \to S^2$ via the weight $-k$ representation of $U(1)$. Then $E = \pi^*( \mathscr{O}(-m_1) \oplus \dots \oplus \mathscr{O}(-m_r))$, and relative to this decomposition into line bundles the scalar field is given in coordinates by the formula
\begin{equation}
    \phi_D(r) = -\frac{im}{2r}.
\end{equation}
We readily compute 
\begin{equation}
    \star d\phi_D = \frac{im}{2r^2} \star dr = \frac{im}{2} \sin \theta d\theta \wedge d\phi
\end{equation}
where $(r, \theta, \phi)$ are the usual spherical coordinates in $\mathbb{R}^3$. From this we read off the ``connection'' $A = \frac{im}{2} \cos \theta d\phi$, which solves the equation $F_A + \star D_A \phi = dA_D + \star d \phi_D = 0$. However, this connection is ill-defined at $\theta = 0$ and $\theta = \pi$, since $\phi$ is not a good coordinate there. So we cover $S^2$ by two patches (`N' for north, `S' for south) and in each patch define the connection to be 
\begin{equation}
\begin{split}
    A^S & = \frac{im}{2}(1 +  \cos \theta) d\phi \\
    A^N & = \frac{im}{2}(\cos \theta-1) d\phi. 
\end{split}
\end{equation}
This is engineered so that $A^S$ is well-defined near $\theta = \pi$ and $A^N$ is well-defined near $\theta = 0$. Note that $A^N - A^S = e^{im \phi} d e^{-i m \phi}$, reflecting the fact that $E$ is glued from two trivial bundles using the transition function $e^{im \phi}$. 

We will denote a reference Dirac field of charge $m$ by $(A_{D, m}, \phi_{D, m})$. It is important to keep in mind the well-known fact that the connection defining a Dirac monopole of nontrivial charge cannot be written in a single gauge across all of $\mathbb{R}^3 \setminus \{ 0 \}$; this will play an important role in determining the asymptotics of scattering matrices below. 

\subsubsection{Boundary condition at infinity}
With the notation understood, we may now state the following. Fix a coweight $(\mu_1, \dots, \mu_r)$. We look for the solutions to \eqref{bpsmonopole} such that, as $r \to \infty$ we have
\begin{equation} \label{monopolebcinfty}
\begin{split}
    \phi(\vec{x}) & \to -i\text{diag}( \phi_1, \dots, \phi_r) - \frac{i\mu}{2r} + \dots \\
    A(\vec{x}) & \to A_{D, \mu} + \dots
\end{split}
\end{equation}
where dots denote terms decaying more rapidly for $r \to \infty$. At the risk of belaboring a triviality, in writing these equations we also understand that we require our bundles $E$ to have a smooth isomorphism $E \eval_U \simeq \pi^*(\mathscr{O}(-\mu_1) \oplus \dots \oplus \mathscr{O}(-\mu_r))$ on some open set $U \subset \mathbb{R}^3$ which is the complement of some ball $\{ |x| \leq R \}$ of suitably large radius $R \gg 0$, and the above formulas are written with respect to this identification. The $\phi_i$ are real numbers that we assume distinct and ordered as $\phi_1 < \phi_2 < \dots < \phi_r$. We will refer to the $\phi_i$ as symmetry breaking parameters, since they break the gauge group down to its maximal torus near infinity. 

\subsubsection{The moduli space of smooth monopoles}
Let $\mathscr{G}_\infty \subset \text{Aut}(E)$ be the subgroup of gauge transformations which approach $1$ as $r \to \infty$. Let us denote the moduli space of smooth monopoles of charge $\mu$ as $\mathscr{M}_\mu$ and make the following
\begin{definition}
\begin{equation}
    \mathscr{M}_\mu := \{ (A, \phi) | F_A + \star D_A \phi = 0, \, \, \, \textnormal{boundary conditions $\eqref{monopolebcinfty}$} \}/\mathscr{G}_\infty.
\end{equation}
\end{definition}
Presented in these terms, the moduli space is a quotient of an infinite-dimensional space by an infinite-dimensional group, so it is not clear that it is in fact a reasonable finite-dimensional space. One of the main applications of the scattering transform to be discussed below is that it gives rise (at least formally) to an accessible, finite-dimensional description of the monopole moduli space by trading partial differential equations for algebra-geometric conditions. 

Note that the charge $\mu$ cannot be an arbitrary coweight, for the following reason. As a vector bundle on $\mathbb{R}^3$, $E$ is necessarily trivial topologically, and in particular its restriction to a sphere of any large radius will be trivial. This is compatible with the boundary condition coming from the Dirac singularity only if $e^{-i \mu \phi}$ represents the trivial homotopy class in $\pi_1(U(r))$, which is true if and only if $\det(e^{-i \mu \phi}) = 1$. This implies that $\mu$ must be a coroot. 

\subsubsection{Dirac singularities}
The moduli space of monopoles described above parameterizes solutions to the monopole equation \eqref{bpsmonopole} which are smooth throughout $\mathbb{R}^3$. It is very useful to consider solutions which are singular, but only have prescribed singularities at certain points, see \cite{Moore_2014} for a careful discussion. Excise some points $p_1, \dots, p_k$ from $\mathbb{R}^3$, and consider bundles $E$ and pairs $(A, \phi)$ defined on $\mathbb{R}^3 \setminus \{ p_1, \dots, p_k \}$, solving $F_A + \star D_A \phi = 0$ with a boundary condition at infinity given by a Dirac monopole of charge $\mu$. 

For each point $p_i$, choose a coweight $\lambda_i$ and impose the condition that, relative to some decomposition $E = \oplus_i \mathscr{L}_i$ into a direct sum of line bundles near $p_i$,
\begin{equation} \label{singularbcmonopole}
\begin{split}
    \phi(x) & \sim - \frac{i \lambda_i}{2|\vec{x} - \vec{x}_i|} + \dots \\
    A(x) & \sim A_{D, \lambda_i} + \dots 
\end{split}
\end{equation}
as $\vec{x} \to \vec{x}_i$, where dots denote less singular terms. We let $\mathscr{G}_{\infty, \{ p_i \}} \subset \text{Aut}(E)$ denote the subgroup of gauge transformations approaching identity as $r \to \infty$ and which stabilize $\lambda_i$ near $p_i$. Then we have a moduli space of smooth monopoles in the background of Dirac singularities, described in the following 

\begin{definition}
\begin{equation}
    \mathscr{M}^{\lambda_1, \dots, \lambda_k}_\mu(p_1, \dots, p_k) := \{ (A, \phi) | F_A + \star D_A \phi = 0, \, \, \,  \textnormal{ boundary conditions \eqref{monopolebcinfty} and \eqref{singularbcmonopole}}  \}/ \mathscr{G}_{\infty, \{ p_i \}}. 
\end{equation}
\end{definition}

We will be interested especially in the case where there is only one point, which we can take to be the origin $0 \in \mathbb{R}^3$, and one coweight $\lambda$ there. We denote the corresponding moduli space by $\mathscr{M}^\lambda_\mu$. Similar considerations as in the previous subsection show that $\lambda - \mu$ must be a coroot.  

$\mathscr{M}^\lambda_\mu$ turn out to be smooth manifolds of real dimension $4|\langle \rho, \mu - \lambda \rangle|$ whenever nonempty ($\rho$ denotes the Weyl vector), see \cite{Moore_2014} for a differential-geometric discussion. They are not compact for two reasons: one is an ``infrared'' effect related to the noncompactness of $\mathbb{R}^3$, and the other is an ``ultraviolet'' effect called monopole bubbling \cite{KW} reminiscent of small instanton phenomena in four dimensions. Like instanton moduli, $\mathscr{M}^\lambda_\mu$ admit partial compactifications $\overline{\mathscr{M}}^\lambda_\mu$, in general singular, where ``bubbled'' configurations are added in. We call these Uhlenbeck-type compactifications. Their explicit construction will be given in algebra-geometric terms below, via the scattering matrix. There is no monopole bubbling when $\lambda$ vanishes or is minuscule. 

Since we quotient only by gauge transformations tending to $1$ at infinity, the constant gauge transformations preserving the form of the boundary conditions at infinity act nontrivially on the moduli spaces $\mathscr{M}^\lambda_\mu$. Due to the symmetry breaking parameters $\phi_i$, this means that a priori only the maximal torus $T \subset G$ acts on a given $\mathscr{M}^\lambda_\mu$. The full group that acts is the centralizer of $\mu$, though this is not obvious without the $S$-matrix construction.

\subsubsection{Hyperkahler and symplectic structures}
The monopole moduli spaces $\mathscr{M}^\lambda_\mu$ carry the structure of infinite-dimensional hyperkahler quotients, and therefore are (at least formally) hyperkahler manifolds, see \cite{atiyahhitchin}. While this is formally true, it is still a quite nontrivial issue to precisely describe the hyperkahler structure on them in finite-dimensional terms. 

Let us recall the formal argument underlying these claims. We consider the infinite dimensional space $\mathscr{A}$ of all connections and scalar fields $(A, \phi)$ satisfying our boundary conditions \eqref{monopolebcinfty} and \eqref{singularbcmonopole}, on which we have the metric 
\begin{equation}
    ds^2 = -\int_{\mathbb{R}^3 \setminus \{ p_j \}}  d^3x \Tr \Big( \delta A_i \otimes \delta A_i + \delta \phi \otimes \delta \phi \Big) 
\end{equation}
and the triplet of symplectic forms 
\begin{equation} \label{hyperkahlerform}
    \omega_i = \int_{\mathbb{R}^3 \setminus \{ p_j \}} d^3 x \Tr \Big( \delta A_i \wedge \delta \phi + \frac{1}{2} \varepsilon_{ijk} \delta A_j \wedge \delta A_k \Big) 
\end{equation}
where $\varepsilon_{ijk}$ is the totally antisymmetric symbol. An infinitesimal automorphism $\epsilon(x) \in \text{Lie}(\mathscr{G}_{\infty, \{ p_i \}} )$ acts on the fields by the usual formulas (we denote $D_i \epsilon = \partial_i \epsilon + \comm{A_i}{\epsilon}$):
\begin{equation} \label{gaugetransf}
\begin{split}
\delta A_i & = D_i \epsilon \\
\delta \phi & = \comm{\phi}{\epsilon}. 
\end{split}
\end{equation}
Note the conditions on gauge transformations in $\mathscr{G}_{\infty, \{ p_i \}}$ mean that we require $\epsilon(|x| \to \infty) \to 0$ sufficiently rapidly, and that $\epsilon(x)$ remains regular as $x \to p_i$ and $\comm{\epsilon(x)}{\lambda_i} = (\text{vanishing as $x \to p_i$})$. 

We may compute the hyperkahler moment map for the action of $\mathscr{G}_{\infty, \{ p_i \}}$ on $\mathscr{A}$ as follows. Let $V_\epsilon$ be the vector field on $\mathscr{A}$ corresponding to the Lie algebra element $\epsilon(x)$, and denote by $\iota_V$ the operation of interior multiplication with $V$. Then we may calculate from \eqref{gaugetransf} 
\begin{equation}
\begin{split}
    \iota_{V_\epsilon} \omega_i & = \int_{\mathbb{R}^3 \setminus \{ p_j \}} d^3x \Tr\Big( D_i \epsilon \delta \phi - \delta A_i \comm{\phi}{\epsilon} + \varepsilon_{ijk} D_j \epsilon \delta A_k \Big) \\
    & = -\int_{\mathbb{R}^3 \setminus \{ p_j \}} \Tr \epsilon(x) \Big(D_i \delta \phi + \comm{\delta A_i}{\phi} + \frac{1}{2} \varepsilon_{ijk} (D_j \delta A_k - D_k \delta A_j) \Big) \\
    & = - \delta( \langle \epsilon, \mu_i \rangle) 
\end{split}
\end{equation}
where 
\begin{equation}
    \langle \epsilon, \mu_i \rangle = \int_{\mathbb{R}^3 \setminus \{ p_j \}} d^3 x \Tr\Big( \epsilon(x) \Big( D_i \phi + \frac{1}{2} \varepsilon_{ijk} F_{jk} \Big) \Big). 
\end{equation}
In getting from the first to second line in the formula for $\iota_V \omega_i$, we have integrated by parts and used various properties of the trace. In going from the second to third line, we have recalled the formula $\delta F_{ij} = D_i \delta A_j - D_j \delta A_i$ for the variation of the curvature of a connection. The reader may enjoy verifying that the boundary terms we would acquire in integration by parts from the sphere at infinity and the infinitesimal spheres surrounding $p_i$ vanish by our assumptions on $\epsilon(x)$ and boundary conditions on the fields. 

From this formal calculation, we can see that the equation $F_A + \star D_A \phi = 0$ is equivalent to the vanishing of the triplet of moment maps $\mu_i$, therefore we may rewrite 
\begin{equation}
    \mathscr{M}^\lambda_\mu = \bigcap_{i = 1}^3 \mu_i^{-1}(0) \Big/ \mathscr{G}_{\infty, \{ p_j \}}
\end{equation}
which presents the finite-dimensional space $\mathscr{M}^\lambda_\mu$ as a hyperkahler quotient of an infinite-dimensional affine space by an infinite-dimensional group. 

From \eqref{hyperkahlerform} it is clear that rotations of $\mathbb{R}^3$ rotate $\omega_i$ as a three-dimensional vector; together with basic facts of hyperkahler geometry this means in particular that a choice of direction in $\mathbb{R}^3$ induces a preferred complex structure on the monopole moduli space. From this point forward, we will only be interested in $\mathscr{M}^\lambda_\mu$ as a complex symplectic variety in the complex structure corresponding to the $x^3$ direction. In this complex structure, $\omega_3$ is identified as the Kahler form and $\Omega = \omega_1 + i \omega_2$ becomes the holomorphic symplectic form.  

\subsection{Definition of scattering matrices}
The scattering transform is a technique that allows one to describe the monopole moduli spaces $\mathscr{M}^\lambda_\mu$ in much more concrete terms, modulo standard conjectures to be addressed below. In this section we will recall the definition of the scattering matrix and explain what kind of object it is. 

\subsubsection{Real and holomorphic equations}
Corresponding to our choice of complex structure on $\mathscr{M}^\lambda_\mu$ is an identification $\mathbb{R}^3 \simeq \mathbb{C} \times \mathbb{R} \ni (u, y)$, where $u = x^1 + i x^2$ and $y = x^3$. In these coordinates we have the connection $A = A_u du + A_{\bar{u}} d\bar{u} + A_y dy$, and denote $D_u = \partial_u + A_u$, $\overline{D}_u = \overline{\partial}_u + A_{\bar{u}}$, $D_y = \partial_y + A_y$. The complex symplectic form $\Omega$ is proportional to 
\begin{equation}
    \omega_{\mathbb{C}} = \int d^2 u dy \Tr \Big( \delta A_{\bar{u}} \wedge \delta(A_y + i \phi) \Big).
\end{equation}
From this we see that $(A_{\bar{u}}, A_y + i \phi)$ give distinguished holomorphic coordinates on $\mathscr{A}$. The equations $F_A + \star D_A \phi = 0$ become 
\begin{equation}
\begin{split}
\comm{\overline{D}_u}{D_y + i \phi} & = 0 \\
\comm{D_u}{\overline{D}_u} + \frac{i}{2} D_y \phi & = 0. 
\end{split}
\end{equation}
The first is interpreted as the complex moment map on $\mathscr{A}$ in this complex structure; the second is the real moment map arising from the Kahler form $\omega_3$. By a standard line of reasoning, one expects an infinite-dimensional analog of the Kempf-Ness theorem (expressing the equivalence of symplectic and GIT quotients) to hold, and for the monopole moduli space to be presented as the space of solutions to only the first equation modulo complex-valued gauge transformations and satisfying an appropriate stability condition. This motivates the completeness of the scattering transform (i.e. the fact that the $S$-matrix contains the full data of the monopole solution, up to gauge equivalence). 

\subsubsection{Idea of scattering transform}
The content of the nonlinear equations $F_A + \star D_A \phi = 0$, when expressed in terms of holomorphic data, is the condition that the differential operators $\overline{D}_u$ and $D_y + i \phi$ commute. To extract the important data in the monopole solution, we follow a well-established path in mathematical physics: we fix a value of $u$, and consider the associated linear problem
\begin{equation} \label{assoclinear}
    (\partial_y + A_y(y, u , \bar{u}) + i \phi(y, u, \bar{u})) s = 0
\end{equation}
for a section $s$ of the restriction of $E$ to the vertical line $u = \text{constant}$. This equation will have a monodromy matrix connecting solutions at $y \to -\infty$ to solutions at $y \to + \infty$. The equation $\comm{\overline{\partial}_u + A_{\bar{u}}}{\partial_y + A_y + i \phi} = 0$ then essentially says that, in an appropriate gauge, this monodromy varies holomorphically over the $u$-plane provided we avoid positions of any Dirac singularities. This condition turns out to be strong enough to characterize the monodromy uniquely up to a finite-dimensional ambiguity, and these finitely many parameters vary precisely in the moduli space $\mathscr{M}^\lambda_\mu$. We will now make this picture more precise and explicit. 

\subsubsection{Wilson lines and holomorphic framing}
The monodromy matrix of the linear problem \eqref{assoclinear} can be expressed formally as a Wilson line/holonomy of the complex connection $A_y + i \phi$:
\begin{equation}
    W(y_f, y_i; u, \bar{u}) = \mathscr{P} \exp{ - \int_{{y_i}}^{y_f} dy(A_y(y, u, \bar{u}) + i \phi(y, u, \bar{u}))}. 
\end{equation}
This gives the matrix of parallel transport from the initial position $y_i$ to final position $y_f$ along the line $u = \text{constant}$. As a consequence of $\comm{\overline{D}_u}{D_y + i \phi} = 0$, $W$ is covariantly holomorphic in the sense that 
\begin{equation}
    \overline{D}_u W(y_f, y_i; u, \bar{u}) = \overline{\partial}_u W(y_f, y_i; u, \bar{u}) + A_{\bar{u}}(y_f, u, \bar{u}) W(y_f, y_i; u, \bar{u}) - W(y_f, y_i; u, \bar{u}) A_{\bar{u}}(y_i, u, \bar{u}) = 0. 
\end{equation}
To extract a matrix that is actually holomorphic in $u$ from $W$, we need to present it in a holomorphic frame; such a frame is constructed from sections annihilated by $\overline{D}_u$. Finally, we will be interested only in the asymptotics of this expression for $y_i \to - \infty$, $y_f \to + \infty$. 

Concretely, this means we study solutions $s(y, u, \bar{u})$ to the compatible system of equations 
\begin{equation}
\begin{split}
    (\partial_y + A_y(y, u, \bar{u}) + i \phi(y, u, \bar{u}))s(y, u, \bar{u}) & = 0 \\
    (\overline{\partial}_u + A_{\overline{u}}(y, u, \bar{u})) s(y, u, \bar{u}) & = 0. 
\end{split}
\end{equation}
We consider asymptotic behavior of these solutions for $|y|$ large, while allowing $u$ to be arbitrary. Let $s^+_a$, $a = 1, \dots, r$ denote a basis of solutions defined in the region of large positive $y$, and similarly $s^-_a$ denote a basis of solutions defined for large negative $y$. We form matrices $g_\pm(y, u, \bar{u}) = [s^\pm_1, \dots, s^\pm_r]$ where $g_+$ is defined for large positive $y$ and $g_-$ is defined for large negative $y$. Note that the $g_\pm$ are at this stage ambiguous (in the sense of not being uniquely determined by the monopole configuration $(A, \phi)$) up to right multiplication by invertible matrices that vary holomorphically in $u$; we will explain in the next subsection how to partially fix this ambiguity. 

\begin{definition} The monopole scattering matrix is
\begin{equation} \label{scatteringmatrixdef}
    S(u) :=  \lim_{L \to \infty} g_+^{-1}(L/2, u, \bar{u}) W(L/2, -L/2; u, \bar{u}) g_-(-L/2, u, \bar{u}).
\end{equation}
\end{definition}

By construction, the matrix coefficients of $S(u)$ are holomorphic functions of the variable $u$, which are regular in the $u$-plane away from the positions of Dirac singularities as a consequence of the regularity of the fields $(A, \phi)$. It is also clear by construction that the matrix coefficients vary holomorphically over the moduli space of monopoles in the complex structure determined by the complex coordinates $(A_{\bar{u}}, A_y + i \phi)$. 

\subsubsection{Triangular ambiguities} \label{trianggauge}
Above, the $g_\pm$ matrices were chosen essentially at random; we would like to eliminate this auxiliary choice to the fullest extent possible. This can be done by inspecting the asymptotic behavior of the fields $(A, \phi)$ more carefully. 

The key observation is that the boundary condition at $r \to \infty$ \eqref{monopolebcinfty} gives a distinguished class of bases $s^\pm_a$ of solutions to the equations $(D_y + i \phi)s = \overline{D}_u s = 0$ in the regions of large positive $y$ and large negative $y$. First consider the behavior of solutions for $y \to + \infty$. Keeping terms up to order $1/y$, the equations read from \eqref{monopolebcinfty} (in the north pole gauge for the Dirac field): 
\begin{equation}
\begin{split}
    \Big( \partial_y + \text{diag}(\phi_1, \dots, \phi_r) + \frac{\mu}{2y} \Big) s & = 0 \\
    \overline{\partial}_u s & = 0. 
\end{split}
\end{equation}
From this we see that, for each $a = 1, \dots, r$ there is a solution behaving for $y \to +\infty$ as
\begin{equation}
    s_a^+(y, u, \bar{u}) \sim e^{- \phi_a y} y^{-\mu _a/2} \begin{pmatrix} 0 \\ \vdots \\ c_a(u) \\ \vdots \\ 0 \end{pmatrix}
\end{equation}
where $c_a(u)$ is a nowhere vanishing entire function of $u$ and the only nonzero entry in the column vector is in the $a$-th spot. 

While solutions $s_a^+$ with the given asymptotics certainly exist, the asymptotics alone does not determine them uniquely. A solution $s_a^+$ with the above asymptotics is uniquely determined only up to changing the function $c_a(u)$ and the addition of solutions vanishing more rapidly for $y \to \infty$. If we choose arbitrarily some reference solutions $(s_1^+, \dots, s_r^+)$ with the required asymptotics, this may be expressed in formulas as a redundancy under (we use here the assumption that the symmetry breaking parameters are in the order $\phi_1 < \dots < \phi_r$)  
\begin{equation}
s_a^+(y, u, \bar{u}) \mapsto \sum_{b \geq a} s_b^+(y, u, \bar{u}) (h_+(u))\indices{^b_a}.
\end{equation}
We may view the collection $(h_+(u))\indices{^a_b}$ as an analytic function of $u$ valued in invertible lower triangular matrices, denoted by $h_+(u)$. We denote the subgroup of lower triangular matrices as $B_+ \subset G_{\mathbb{C}} = GL_r$. Thus the matrix $g_+(y, u, \bar{u})$ is in general determined uniquely by the monopole configuration $(A, \phi)$ up to $g_+(y, u, \bar{u}) \mapsto g_+(y, u, \bar{u}) h_+(u)$. 

A similar consideration shows that $g_-$ is defined up to the ambiguity $g_-(y, u, \bar{u}) \mapsto g_-(y, u, \bar{u}) h_-(u)$ where $h_-$ is valued in the subgroup $B_- \subset GL_r$ of upper triangular matrices. 

In the next section, it will be demonstrated that in fact the only singularities the $S$-matrix may have are poles (at $u \to \infty$ and the positions of Dirac singulairites) of some order fixed by the integers in $\mu$ and $\lambda$. That is, in an appropriate (algebraic) class of gauges $S(u)$ will be an element of $G_{\mathbb{C}}(u)$ (invertible matrices with entries in rational functions of the variable $u$). 

The triangular ambiguities preserving such a class of gauges are those in which the functions $h_\pm(u)$ are not merely analytic but are in fact algebraic, so that regularity over the $u$-plane forces them to be polynomial. We denote polynomial functions valued in $B_\pm$ as $B_\pm[u]$---note the diagonal entries of such matrices must be constant by the requirement that the matrices have nonzero determinant for each value of $u$. The scattering matrix then canonically takes values in the double quotient $B_+[u] \backslash  G_{\mathbb{C}}(u) / B_-[u]$, and the assignment of a monopole to its scattering data $(A, \phi) \mapsto S(u)$ provides a canonical map 
\begin{equation} \label{maptodoublequotient}
    \mathscr{M}^\lambda_\mu \longrightarrow B_+[u] \backslash G_{\mathbb{C}}(u) /B_-[u]. 
\end{equation}
One expects the scattering transform to be complete, meaning that this map is actually an embedding. In the rank 2 case with $\lambda = 0$ this follows from \cite{donaldsonrational} as observed in \cite{hurtubise}. The result for $\lambda = 0$ but general $G$ and $\mu$ was obtained by Jarvis \cite{jarvis}. To the knowledge of this author no complete proof has been given for nonzero $\lambda$, though the result is almost certainly true. 

It is conceivable that the result could be shown via comparing our formula \eqref{eq:Smatrixfromprop} to an analysis of Nahm's equations describing the appropriate singular monopole configuration. 

This map can be used to give $\mathscr{M}^\lambda_\mu$ the structure of an algebraic variety, and it is in this way that we will typically regard $\mathscr{M}^\lambda_\mu$ in this paper. Note the complex structure on $\mathscr{M}^\lambda_\mu$ viewed as a variety via this map is automatically compatible with the complex structure arising from gauge theory in which $(A_{\bar{u}}, A_y + i \phi)$ are the holomorphic coordinates. 

\subsection{Determination of scattering matrices}
Now we will see how to determine the asymptotics of $S(u)$ near its singularities. This will give a characterization of $\mathscr{M}^\lambda_\mu$ as a complex manifold and allow for the construction of its Uhlenbeck-type partial compactification $\overline{\mathscr{M}}^\lambda_\mu$ as an affine algebraic variety, assuming the completeness of the scattering transform. Because of the boundary conditions \eqref{monopolebcinfty}, \eqref{singularbcmonopole} the key calculation to do is determine the scattering data for a Dirac monopole of arbitrary charge. 

\subsubsection{Scattering data for Dirac monopoles}
Now we will calculate $S(u)$ in the special case of a Dirac monopole field of charge $m$. In the $(y, u, \bar{u})$ coordinates the fields are given by the expressions (with symmetry breaking parameters included, though they will do essentially nothing):
\begin{equation}
\begin{split}
    \phi(y, u, \bar{u}) & = -i \text{diag}(\phi_1, \dots, \phi_r) - \frac{im}{2\sqrt{y^2 + |u|^2}} \\
    A^{N, S}(y, u, \bar{u}) & = \frac{im}{2}\Big(\frac{y}{\sqrt{y^2 + |u|^2}} \pm 1 \Big) \frac{1}{2i} \Big( \frac{du}{u} - \frac{d\bar{u}}{\bar{u}} \Big) 
\end{split}
\end{equation}
and the sign on $\pm1$ is chosen according to north/south patch; note north is $y \to + \infty$ and south is $y \to -\infty$. In particular we see that to leading order at large $y$, both $A^N$ and $A^S$ vanish in their respective patches, and that the $y$-component of $A$ identically vanishes $A_y = 0$. 

Using the elementary integral 
\begin{equation}
    \int_{-L/2}^{+L/2} \frac{dy}{\sqrt{y^2 + |u|^2}} = 2\ln \Bigg( \frac{L}{2|u|} + \sqrt{1 + \frac{L^2}{4|u|^2}} \Bigg) 
\end{equation}
we easily compute the asymptotics (note path ordering is not necessary here since we deal with abelian gauge fields): 
\begin{equation} \label{diracwilsonline}
\exp{ -\int_{-L/2}^{+L/2} dy(A_y(y, u, \bar{u}) + i \phi(y, u, \bar{u}))} \sim \text{diag}\Bigg(e^{-\phi_1 L} \Big( \frac{L}{|u|} \Big)^{-m_1} , \dots, e^{-\phi_r L} \Big( \frac{L}{|u|} \Big)^{-m_r} \Bigg)  
\end{equation}
for $L \to \infty$. To determine the $g_\pm$ matrices, we just pick out asymptotic solutions for $y \to \pm \infty$. An asymptotic solution $s^-$ for $y \to -\infty$ solves to leading order the equations 
\begin{equation} 
\begin{split}
    (\partial_y + \text{diag}(\phi_1, \dots, \phi_n) -m/2y) s^+ & = 0 \\
    (\overline{\partial}_u + A^S_{\bar{u}})s^- & = 0.
\end{split}
\end{equation}
Since $A^S_{\bar{u}}$ vanishes to leading order at large $y$, the most general asymptotic section $s^-$ has components $s^{-, i}$, $i = 1, \dots, r$ of the form
\begin{equation}
    s^{-, i}(y, u, \bar{u}) \sim  e^{-\phi_i y} (-y)^{m_i/2} h_{-, i}(u)
\end{equation}
for $y \to -\infty$, where $h_{-, i}(u)$ is a nowhere vanishing entire function of $u$. Note that the triangular ambiguity discussed above does not arise because the fields describing a Dirac monopole are given by diagonal matrices so the problem in essence reduces to the $1 \times 1$ case. We may go to a gauge in which $h_-$ is an algebraic function of $u$, not merely analytic, in which case it follows that $h_-$ must be a constant. Discarding the constants, we see that 
\begin{equation} \label{diracg-}
    g_-(-L/2, u, \bar{u}) \sim  \text{diag} \Big( e^{\phi_1 L/2} L^{m_1/2}, \dots, e^{\phi_r L/2} L^{m_r/2} \Big) 
\end{equation}
for $L \to \infty$. A similar computation for $y \to + \infty$, and in the north pole gauge for the Dirac field so that $A^N_{\bar{u}}$ vanishes to leading order, gives the asymptotics 
\begin{equation} \label{diracg+}
    g_+(L/2, u, \bar{u})  \sim \text{diag} \Big( e^{- \phi_1 L/2} L^{-m_1/2}, \dots, e^{-\phi_r L/2} L^{-m_r/2}\Big).  
\end{equation}

Now we may combine \eqref{diracwilsonline}, \eqref{diracg-}, \eqref{diracg+} to compute $S(u)$. Before simply multiplying them together, there is one caveat to take into account: \eqref{diracg-} is written in the south pole patch, while \eqref{diracg+} is written in the north pole patch; if we use \eqref{diracwilsonline} to continue from large negative $y$ to large positive $y$, we must also multiply by the transition function $e^{+im \phi} = (u/\bar{u})^{m/2}$ before multiplying with $g_+^{-1}$ since \eqref{scatteringmatrixdef} is written under the assumption that $W$, $g_-$, and $g_+$ are expressed with respect to compatible trivializations of the bundle. Put differently, since \eqref{diracg+} is written in the north pole patch the correct substitution to make for $W$ in \eqref{scatteringmatrixdef} is not \eqref{diracwilsonline} but $(u/\bar{u})^{m/2} \times \text{\eqref{diracwilsonline}}$. 

Assembling all the contributions, from \eqref{scatteringmatrixdef} we compute (the subscript `D' reminds us that this is the result for a Dirac monopole of charge $m$):
\begin{equation}
\begin{split}
    S_D(u) & = \lim_{L \to \infty} \text{diag} \Big( e^{\phi_i L/2} L^{m_i/2} \times (u/\bar{u})^{m_i/2} \times e^{-\phi_i L} \frac{(u \bar{u})^{m_i/2}}{L^{m_i}} \times e^{\phi_i L/2} L^{m_i/2} \Big)_{i = 1}^r  \\
    & = \text{diag}(u^{m_1}, \dots, u^{m_r}) = u^m. 
\end{split}
\end{equation}
Note that all the $L$ and $\bar{u}$ dependence cancels out in the end, as it must. The fact that the scattering data for Dirac monopoles takes this form underlies the basic link between Dirac singularities and Hecke modification of bundles discussed in \cite{KW}. 

Note that the monopole equations in the case of an abelian gauge group are linear, so that solutions add and scattering matrices multiply. If there are multiple Dirac singularities at positions $(u_\alpha, y_\alpha)$ of charges $m_\alpha$, the total scattering matrix is then $S_D(u) = \prod_\alpha (u - u_\alpha)^{m_\alpha}$. Note that the scattering matrix is insensitive to the position $y_\alpha$ in the vertical direction, since we integrate over a line in this direction in computing $W$. 

\subsubsection{Asymptotics as $u \to \infty$}
The preliminaries are finally in place to determine the matrix $S(u)$ for $u \to \infty$. By the discussion in section \ref{trianggauge}, as $L \to \infty$ we have the asymptotics 
\begin{equation}
    g_\pm(L/2, u, \bar{u}) \sim \text{diag}(e^{\mp \phi_1 L/2} L^{\mp m_1/2}, \dots, e^{\mp \phi_r L/2} L^{\mp m_r/2}) h_\pm(u).
\end{equation}
In equations involving both $L$ and $u$, we understand that we consider asymptotics with $|u|$ large but $L \gg |u|$. The matrices $h_\pm$ are lower (for $+$) or upper (for $-$) triangular and the nonzero entries grow polynomially in $u$ for large $u$. The diagonal entries may be assumed without loss of generality to be $1$ by normalizing the sections $s_a^\pm$. 

Likewise, using the boundary condition \eqref{monopolebcinfty}, we see that we may replace $W$ in \eqref{scatteringmatrixdef} by the Wilson line for the Dirac field computed in the previous section. Then we may borrow the computation of the previous section to conclude the asymptotics 
\begin{equation} \label{ulargeasympt}
    S(u) \sim h^{-1}_+(u) u^\mu h_-(u)
\end{equation}
for $u \to \infty$; we understand this statement to mean that in any gauge of the type discussed in section \ref{trianggauge}, the monopole scattering matrix must have the above leading behavior at large $u$. 

The following conclusions may be drawn from \eqref{ulargeasympt}. Let $\Delta_i(S(u))$ denote the $i$-th principal minor of $S(u)$ (the determinant of the top left $i \times i$ submatrix); since $\Delta_i(h^{-1}_+(u) u^\mu h_-(u)) = \Delta_i(u^\mu)$, we conclude that $\Delta_i(S(u)) \sim u^{\langle \mu, \omega_i \rangle}$ for $u \to \infty$ where brackets denote the canonical pairing of coweights with weights and $\omega_i = (1, \dots, 1, 0, \dots, 0)$ is the $i$-th fundamental weight of $GL_r$. In particular, if we pass to formal power series in $u^{-1}$, $\Delta_i(u)$ become invertible and we have a (necessarily unique) Gauss decomposition 
\begin{equation} \label{gaussdecomp}
    S(u) = E(u) H(u) F(u)
\end{equation}
where $E(u)$ is lower triangular with unit diagonal, $H(u)$ is diagonal, and $F(u)$ is upper triangular with unit diagonal. Then \eqref{ulargeasympt} together with uniqueness of the Gauss decomposition imply that $E(u)$ and $F(u)$ have at most polynomial growth for $u$ large and that $H(u)$ has an expansion for $u \to \infty$ of the form
\begin{equation}
    H(u) = u^\mu(1 + O(u^{-1})).
\end{equation}
However, as discussed in section \ref{trianggauge}, $S(u)$ is ambiguous up to $B_+[u]$ on the left and $B_-[u]$ on the right; the asymptotics \eqref{ulargeasympt} implies that in fact each $S(u)$, regarded as an element of $B_+[u] \backslash G_{\mathbb{C}}(u) /B_-[u]$, has a unique representative in $G_{\mathbb{C}}(u)$ of the form \eqref{gaussdecomp} if we in addition impose the conditions
\begin{equation}
\begin{split}
    E(u) & = 1 + O(u^{-1}) \\ 
    F(u) & = 1 + O(u^{-1}). 
\end{split}
\end{equation}
Informally, we can just exhaust the $B_+[u]$ freedom by using it to kill any positive powers of $u$ appearing in $E(u)$, and likewise $B_-[u]$ to kill any positive powers appearing in $F(u)$. 

In other words, we have shown that \eqref{ulargeasympt} implies that we have a lifting of the canonical map \eqref{maptodoublequotient} to a map 
\begin{equation}
    \mathscr{M}^\lambda_\mu \longrightarrow G_{\mathbb{C}}(u) 
\end{equation}
and composing with the completion (expansion at large $u$) $G_{\mathbb{C}}(u) \hookrightarrow G_{\mathbb{C}}(( u^{-1}))$ its image lands in $$B_{+, 1}[[ u^{-1} ]] u^\mu B_{-, 1}[[ u^{-1} ]]$$ where subscripts $1$ mean kernel of evaluation at $u = \infty$. 

\subsubsection{Asymptotics near Dirac singularities}
Essentially the same consideration determines the asymptotics near the position of a Dirac singularity, which we assume is at $u = 0$. As $u \to 0$, it is clear from the definition \eqref{scatteringmatrixdef} that $W$ will receive a dominant contribution from a small neighborhood of $y = 0$ along the line, since the gauge field is becoming very nearly singular there. This contribution is once again given by the scattering matrix of a Dirac monopole as a result of \eqref{singularbcmonopole}, and by a standard saddle point asymptotics argument one sees that the scattering matrix must behave as 
\begin{equation} \label{usmallasympt}
    S(u) \sim S_+(u) u^\lambda S_-(u)
\end{equation}
for $u \to 0$, where the matrices $S_\pm(u)$ are regular at $u = 0$ but are otherwise arbitrary, and correspond to the contributions of either side of the small interval enclosing $y = 0$ to $W$. If the only Dirac singularity is at $u = 0$, this condition may be restated more algebraically as 
\begin{equation}
    S(u) \in G_{\mathbb{C}}[u] u^\lambda G_{\mathbb{C}}[u] \subset G_{\mathbb{C}}(u)
\end{equation}
where we regard the $S$-matrix as an element of $G_{\mathbb{C}}(u)$ by the construction of the previous subsection. The Uhlenbeck-type compactification $\overline{\mathscr{M}}^\lambda_\mu$ is constructed by relaxing this condition to (bar denotes closure)
\begin{equation}
    S(u) \in \overline{G_{\mathbb{C}}[u] u^\lambda G_{\mathbb{C}}[u]} \subset G_{\mathbb{C}}(u),
\end{equation}
see below. This means we require that the $S$-matrix has a pole at $u = 0$ of order $\leq \lambda$, rather than order exactly $\lambda$. Configurations with poles of order $< \lambda$ are ``bubbled'' configurations in which the singular Dirac charge has been screened by some smooth monopole charge. 

If there are multiple Dirac singularities, one just imposes \eqref{usmallasympt} separately at the location of each singularity, and uses the charge of the singularity for $\lambda$. It is likewise straightforward to restate this condition algebraically. 

\subsubsection{Monopole moduli spaces and generalized affine Grassmannian slices}
The subspace 
\begin{equation} \label{eq:genslicedef}
    \overline{\mathscr{W}}^\lambda_\mu : = \overline{G_{\mathbb{C}}[u] u^\lambda G_{\mathbb{C}}[u]} \bigcap B_{+, 1}[[u^{-1}]] u^\mu B_{-, 1}[[u^{-1}]] \subset G_{\mathbb{C}}(u)
\end{equation}
coincides with the generalized affine Grassmannian slice defined in \cite{bfnslice}, see Section 2(xi) of that paper. It is known to be an irreducible affine algebraic variety of dimension $2|\langle \rho, \mu - \lambda \rangle|$, and the arguments of the previous two sections show that the map \eqref{maptodoublequotient} lifts to $G_{\mathbb{C}}(u)$ and factors through a map 
\begin{equation}
    \mathscr{M}^\lambda_\mu \longrightarrow \overline{\mathscr{W}}^\lambda_\mu.
\end{equation}
It is expected that the completeness of the scattering transform identifies this map with the inclusion of $\mathscr{M}^\lambda_\mu$ into its Uhlenbeck-type compactification $\overline{\mathscr{M}}^\lambda_\mu \simeq \overline{\mathscr{W}}^\lambda_\mu$, and that $\mathscr{M}^\lambda_\mu$ itself is recovered as the smooth locus of $\overline{\mathscr{M}}^\lambda_\mu$. 

\subsubsection{Complete determination in rank $2$}
For $G = U(2)$ we will make these results fully explicit. Denote a prospective $S$-matrix by 
\begin{equation}
    S(u) = \begin{bmatrix} Q(u) & U^+(u) \\ U^-(u) & \widetilde{Q}(u) \end{bmatrix}. 
\end{equation}
We will choose the charge at the origin to be $\lambda = (0, n)$ and the charge at infinity to be $\mu = (k, n - k)$ for $n, k \geq 0$. Then requiring 
\begin{equation}
    S(u) = E(u) H(u) F(u)
\end{equation}
with $E(u)$, $F(u)$ of the form $1 + O(u^{-1})$ and $H(u) = u^\mu(1 + O(u^{-1}))$ implies that $Q(u) = u^k + \dots$, and $U^\pm(u) = O(u^{k - 1})$ for $u \to \infty$. Likewise $S(u) \in \overline{ G_{\mathbb{C}}[u] u^\lambda G_{\mathbb{C}}[u]}$ if and only if all entries of $S(u)$ are polynomials in $u$ and 
\begin{equation}
    \det S(u) = Q(u) \widetilde{Q}(u) - U^+(u) U^-(u) = u^n
\end{equation}
(we see that the leading coefficient of the determinant must be $1$ by comparing with $u \to \infty$ asymptotics). Note that from the determinant condition it follows that
\begin{equation}
    \deg \widetilde{Q}(u) \leq  
    \begin{cases}
    k - 2 & n < 2k - 1 \\ 
    n - k & n \geq 2k - 1
    \end{cases}
\end{equation}
and $\widetilde{Q}(u)$ is monic in the second case. Taking coefficients of $u$ in $\det S(u) = u^n$ gives explicit equations cutting out the monopole moduli space $\overline{\mathscr{M}}^\lambda_\mu \simeq \overline{\mathscr{W}}^\lambda_\mu$ as an affine variety; its coordinate ring is thus obviously generated by the coefficients of matrix elements of $S(u)$. 

Note that displacing the positions of the Dirac singularities to $x_\alpha = (u = a_\alpha, y = 0)$, so that each singularity has charge $(0, 1)$, deforms the moduli space to the affine variety cut out by 
\begin{equation}
    \det S(u) = \prod_{\alpha = 1}^n (u - a_\alpha)
\end{equation}
which is smooth as long as $a_\alpha$ are distinct. Displacing the singular monopoles in the $y$-direction while fixing them at $u = 0$ resolves the singularities of $\overline{\mathscr{W}}^\lambda_\mu$; we will not make use of the resolutions in this paper. 

\subsection{Poisson bracket}
As discussed above, the moduli space of monopoles when viewed as a complex symplectic variety comes with the canonically defined symplectic form 
\begin{equation}
    \omega_{\mathbb{C}} = \int_{\mathbb{R}^3 \setminus \{ p_j \}} d^2 u dy \Tr( \delta A_{\bar{u}} \wedge \delta (A_y + i \phi))
\end{equation}
descending from the space of all connections and scalars. As the scattering matrix is defined directly in terms of the gauge fields and scalars, this symplectic form may be used to compute the Poisson brackets of the matrix coefficients, at least modulo a certain detail addressed below. 

Recall that on any symplectic manifold with symplectic form $\Omega$ and local coordinates $x^I$, 
\begin{equation}
    \Omega = \frac{1}{2} \Omega_{IJ} dx^I \wedge dx^J
\end{equation}
that the Poisson bracket of functions is defined as 
\begin{equation}
    \acomm{f}{g} = \Omega^{IJ} \pdv{f}{x^I} \pdv{g}{x^J}
\end{equation}
where $\Omega^{IJ}$ is the matrix inverse to $\Omega_{IJ}$. Then computing it is equivalent to the following formal recipe: compute $\delta f \wedge \delta g$, where variations are taken with respect to changes of coordinates $x \to x + \delta x$, and replace $\delta x^I \wedge \delta x^J$ with $\Omega^{IJ}$.

This continues to make sense in the infinite-dimensional situation. Let $m, n = 1, \dots, r$ denote vector bundle indices. Then the components of connections and scalars evaluated at points $(A_{\bar{u}})\indices{^m_n}(y, u, \bar{u})$, similarly for $A_y + i \phi$, may be viewed as an infinite set of coordinates on the space $\mathscr{A} \ni (A, \phi)$. The symplectic form is then written in terms of these coordinates as (implied sum on all indices) 
\begin{equation}
    \omega_{\mathbb{C}} = \int d^2 u d^2 u' dy dy' \delta (A_{\bar{u}})\indices{^m_n}(y, u, \bar{u}) \delta(A_y + i \phi)\indices{^n_m}(y', u', \bar{u}') \delta(y - y') \delta^2(u  - u').
\end{equation}
$\delta(y - y')$ is the Dirac delta function and $\delta^2(u - u')$, its two-dimensional analog. To compute Poisson brackets for two functions on the moduli space of monopoles, which descend from functionals of $(A_{\bar{u}}, A_y + i \phi)$, one follows the following procedure: compute the product of variations $\delta f \delta g$ with respect to the connection and scalars, and then wherever it appears make the replacement 
\begin{equation}
    (\delta A_{\bar{u}})\indices{^m_n}(y, u, \bar{u}) \delta(A_y + i \phi)\indices{^{m'}_{n'}}(y', u', \bar{u}') \to \delta^m_{n'} \delta^{m'}_n \delta(y - y') \delta^2(u - u')
\end{equation}
(and of course, replace any $\delta A_{\bar{u}} \delta A_{\bar{u}}$ terms by $0$, similarly for $\delta(A_y +i \phi) \delta (A_y + i \phi)$). Thus the determination of the Poisson structure on the monopole moduli space is reduced to a relatively straightforward exercise in calculus of variations. We will carry out this analysis below for the Poisson brackets of scattering matrix elements $\acomm{S\indices{^i_j}(u)}{S\indices{^k_\ell}(v)}$.

\subsubsection{Variations}
The scattering matrix is defined as 
\begin{equation}
    S(u) = \lim_{L \to \infty} g^{-1}_+(L/2, u, \bar{u}) W(L/2, -L/2; u, \bar{u}) g_-(-L/2, u, \bar{u}) = \lim_{L \to \infty} S_L(u).
\end{equation}
Therefore, to calculate Poisson brackets it suffices to understand the variation of $g_\pm$ and $W$ under a change of connection. 

For the Wilson line, it is well-known and follows straightforwardly from the definition that 
\begin{equation}
    \delta W(L/2, -L/2; u, \bar{u}) = -\int_{-L/2}^{L/2} dy W(L/2, y; u, \bar{u})(\delta A_y(y, u, \bar{u}) + i \delta\phi(y, u, \bar{u}))W(y, -L/2; u, \bar{u}).   
\end{equation}
Note in particular that the Wilson line depends only on $A_y + i \phi$. To get interesting Poisson brackets will require some $A_{\bar{u}}$ dependence, which comes from $g_\pm$. 

Each of $g_\pm$ are defined so that they satisfy 
\begin{equation} \label{gconj}
    g^{-1}_\pm \overline{D}_u g_\pm = \overline{\partial}_u
\end{equation}
as an operator equation; this is just restating that the columns of the matrix $g_\pm$ are a basis of solutions to $\overline{D}_u s = 0$ in the region $y \to \pm \infty$. In this equation it is understood that $\overline{D}_u  = \overline{\partial}_u + A_{\bar{u}}$ where $A_{\bar{u}}$ is evaluated at $y = \pm L/2$ for the equation involving $g_\pm$. 

Taking the variation of \eqref{gconj} gives (we suppress the subscript $\pm$ for clarity)
\begin{equation} \label{dgequation}
\begin{split}
    -g^{-1} \delta g g^{-1} \overline{D}_u g + g^{-1} \delta A_{\bar{u}} g + g^{-1} \overline{D}_u g g^{-1} \delta g   & = \comm{\overline{\partial}_u}{g^{-1} \delta g} + g^{-1} \delta A_{\bar{u}} g \\
    & = \overline{\partial}_u(g^{-1} \delta g) + g^{-1} \delta A_{\bar{u}} g = 0. 
\end{split}
\end{equation}
Treating $\delta A_{\bar{u}}$ as given, the unique $\delta g$ solving this equation and vanishing at $|u| \to \infty$ may be written as the integral 
\begin{equation} \label{dgsolution}
    g^{-1}_\pm \delta g_\pm(\pm L/2, u, \bar{u}) = - \int_{\mathbb{C}} d^2 v \frac{1}{u - v} g^{-1}_\pm \delta A_{\bar{u}} g_\pm (\pm L/2, v, \bar{v}).
\end{equation}
The content of this equation is simply that $\overline{\partial}_u (u - v)^{-1} = \delta^2(u - v)$. We will not discuss convergence of such integrals or related matters, since we are just doing a formal computation. 

\subsubsection{Main calculation}
Now the results just obtained will be applied to calculate the Poisson bracket of the matrix elements $\acomm{S\indices{^i_j}(u)}{S\indices{^k_\ell}(v)}$. For the variation of the $S$-matrix one computes (the cutoff $L$ is being kept finite to avoid a later nuisance with factors of $2$):
\begin{equation}
\begin{split}
    \delta S_L(u) & = - g_+^{-1} \delta g_+ g_+^{-1} Wg_- + g_+^{-1} \delta W g_- + g_+^{-1} W g_- g_-^{-1} \delta g_- \\
    & = + \int_{\mathbb{C}} d^2 w \frac{1}{u - w} g_+^{-1} \delta A_{\bar{u}} g_+(L/2, w, \bar{w}) S_L(u) \\& - \int_{-L/2}^{+L/2} dy g_+^{-1}W(L/2, y; u, \bar{u}) \delta(A_y + i \phi)(y, u, \bar{u}) W(y, -L/2; u, \bar{u}) g_- \\
    & - \int_{\mathbb{C}} d^2 w \frac{1}{u - w} S_L(u) (g_-^{-1} \delta A_{\bar{u}} g_-)(-L/2, w, \bar{w})  .
\end{split}
\end{equation}
To finish out the formal computation we follow our recipe: we compute $\delta S_L(u) \otimes \delta S_{L'}(v)$, then replace $$\delta (A_{\bar{u}})\indices{^m_n}(y, w, \bar{w}) \wedge \delta(A_y + i \phi)\indices{^{m'}_{n'}}(y', w', \bar{w}') \to \delta^m_{n'} \delta^{m'}_n \delta(y - y') \delta^2(w - w')$$
wherever it appears. To avoid troubles with delta functions supported at the boundaries of integration, we adopt the following ``point splitting'' regularization procedure: set $L' = L + \varepsilon$ in the computation for small $\varepsilon > 0$, then at the end send $\varepsilon \to 0$ and then $L \to \infty$. It is easy to see that nothing essential depends on this choice, and that other choices will lead to equivalent results (e.g. interpreting an integral with a delta function on the boundary as $1/2$ leads to a slightly more tedious calculation with the same answer). Let us also restore the indices on $S_L$ and $S_{L'}$, so that we are computing $\delta (S_L)\indices{^i_j} \delta (S_{L'})\indices{^k_\ell}$. 

Now it is clear that, in computing the product, only terms involving $\delta A_{\bar{u}} \delta(A_y + i \phi)$ will contribute to the Poisson bracket, so we can ignore any $\delta A_{\bar{u}} \delta A_{\bar{u}}$ terms and likewise any $\delta(A_y + i \phi) \delta (A_y + i \phi)$ terms. Also, because we assume $L < L'$ the term in $\delta S_L$ involving $\delta(A_y + i \phi)$ will contribute zero when contracted with either $\delta A_{\bar{u}}$ appearing in $\delta S_{L'}$, since the latter are supported at $\pm L'/2$ while the $\delta(A_y + i \phi)$ variation is supported at some $y \in (-L/2, L/2)$, so $\delta(\pm L'/2 - y) = 0$ identically. 

By this discussion, only the first and third term in $\delta S_{L}$ can produce any nonzero contractions, and they must do so with the middle term in $\delta S_{L'}$, so we have 
\begin{equation}
\begin{split}
    (\delta S_L(u))\indices{^i_j} (\delta S_{L'}(v))\indices{^k_\ell} & = \\
    -\int \int_{-L'/2}^{+L'/2} d^2w dy' \frac{1}{u - w} &  (g_+^{-1})\indices{^i_m}(L/2, w, \bar{w}) (g_+(w, \bar{w})S_L(u))\indices{^n_j} (g_+^{-1}W)\indices{^k_{m'}}(L'/2, y'; v, \bar{v}) \times \\ & (Wg_-)\indices{^{n'}_\ell}(y', -L'/2; v, \bar{v})
    \delta(A_{\bar{u}})\indices{^m_n}(L/2, w, \bar{w}) \delta(A_y + i \phi)\indices{^{m'}_{n'}}(y', v, \bar{v}) \\
    + \int \int_{-L'/2}^{+L'/2} d^2 w dy' \frac{1}{u - w} & (S_L(u) g_-^{-1}(w, \bar{w}))\indices{^i_m} (g_-)\indices{^n_j}(-L/2, w, \bar{w})  (g_+^{-1}W)\indices{^k_{m'}}(L'/2, y'; v, \bar{v}) \times \\ & (Wg_-)\indices{^{n'}_\ell}(y', -L'/2; v, \bar{v})) \delta(A_{\bar{u}})\indices{^m_n}(-L/2, w, \bar{w}) \delta(A_y + i \phi)\indices{^{m'}_{n'}}(y', v, \bar{v}) \\
    & + \dots 
\end{split}
\end{equation}
where dots denote terms that do not contribute to the Poisson bracket. 

Consider the first line above. The contraction of $\delta A_{\bar{u}}$ with $\delta(A_y + i \phi)$ is just the instruction to kill the $dy'$ integral and set $y' = L/2$, and likewise kill the $d^2 w$ integral and set $w = v$. The term $\delta^m_{n'} \delta^{m'}_n$ is the instruction to contract various remaining terms together, so the contribution of the first line to the Poisson bracket is 
\begin{equation}
    -\frac{1}{u - v} (g_+^{-1} W g_{-})\indices{^i_\ell}(L/2, -L'/2; v, \bar{v}) (g_+^{-1} W)\indices{^k_n}(L'/2, L/2; v, \bar{v}) (g_+(L/2, v, \bar{v}) S_L(u))\indices{^n_j}. 
\end{equation}
Upon setting $L' = L + \varepsilon$ and letting $\varepsilon \to 0$, the factor  $(g_+^{-1}W)(L'/2, L/2; v, \bar{v})$ tends to $g_+^{-1}(L/2, v, \bar{v})$ which cancels against the $g_+$ in $g_+S_L(u)$. Then sending $L \to \infty$, we get $S(u)$ from those factors by definition, and likewise $(g_+^{-1} W g_-)(L/2, -L/2,; v, \bar{v})$ tends to $S(v)$ by definition. In the end these terms contribute
\begin{equation}
    - \frac{1}{u - v} S\indices{^i_\ell}(v) S\indices{^k_j}(u)
\end{equation}
to the Poisson bracket. The consideration for the second term is very similar: upon performing contractions it reduces to 
\begin{equation}
    \frac{1}{u - v} (S_L(u) g_-^{-1}(-L/2, v, \bar{v}))\indices{^i_m} (Wg_{-})\indices{^m_\ell}(-L/2, -L'/2; v, \bar{v}) (g_+^{-1} W g_-)\indices{^k_j}(L'/2, -L/2; v, \bar{v}).
\end{equation}
Setting $L' = L + \varepsilon$, as $\varepsilon \to 0$ the factor $(Wg_-)(-L/2, -L'/2, v, \bar{v})$ tends to $g_-(-L/2, v, \bar{v})$ which cancels the $g_-^{-1}$ in $S_L(u) g_-^{-1}$. Then when $L \to \infty$ we get $S(u)$ and likewise $(g_+^{-1}Wg_-)(L/2, -L/2; v, \bar{v})$ becomes $S(v)$, so that this term contributes 
\begin{equation}
    \frac{1}{u - v} S\indices{^i_\ell}(u) S\indices{^k_j}(v).
\end{equation}
We conclude that the Poisson bracket is given by 
\begin{equation} \label{monopolerttbracket}
    \acomm{S\indices{^i_j}(u)}{S\indices{^k_\ell}(v)} = \frac{S\indices{^i_\ell}(u) S\indices{^k_j}(v) - S\indices{^i_\ell}(v) S\indices{^k_j}(u)}{u - v}.
\end{equation}

\subsubsection{A caveat}
This computation was somewhat formal, so it is unsurprising that it may be corrected in sufficiently concrete situations. Indeed, the key issue is that in passing from \eqref{dgequation} to \eqref{dgsolution}, we assumed that we wanted solutions tending to $0$ as $|u| \to \infty$ which fixed $\delta g$ in terms of $\delta A_{\bar{u}}$ uniquely. However it was explained at length above that $g_\pm$ are \textit{not} uniquely determined by $A_{\bar{u}}$: any $g'_\pm = g_\pm h_\pm(u)$ with $h_\pm(u)$ holomorphic on the $u$-plane and triangular will also do.

And indeed, as reviewed above, one needs to make use of the ability to redefine $g_\pm$ to be able to write concrete formulas for matrix coefficients of $S(u)$, so this is not a freedom that may be ignored. Thus \eqref{dgsolution} should be read as holding modulo adding a triangular matrix depending polynomially on $u$, and the true Poisson bracket takes the form 
\begin{equation}
    \acomm{S\indices{^i_j}(u)}{S\indices{^k_\ell}(v)} = \frac{S\indices{^i_\ell}(u) S\indices{^k_j}(v) - S\indices{^i_\ell}(v) S\indices{^k_j}(u)}{u - v} + \dots 
\end{equation}
where the dots involve terms of higher degree in $u$. Unfortunately, there does not appear to be any universal or useful form for these terms. 

In fact, one can avoid determining them at all by the following trick. If we apply the naive Poisson bracket, ignoring the corrections, to functions of the matrix coefficients which are invariant under the polynomial triangular transformations, then the derivation above goes through unchanged. For example, for $U(2)$ monopoles the scattering matrices take the form 
\begin{equation}
    S(u) = \begin{bmatrix} Q(u) & U^+(u) \\ U^-(u) & \widetilde{Q}(u) \end{bmatrix}
\end{equation}
and such invariant functions are given by 
\begin{equation}
\begin{split}
    h(u) & = Q(u) \\
    e(u) & = \Big[ \frac{U^-(u)}{Q(u)} \Big]_- \\
    f(u) & = \Big[ \frac{U^+(u)}{Q(u)} \Big]_-
\end{split}
\end{equation}
where the subscript $-$ means take the part of negative degree in the $u \to \infty$ expansion. In the language of quantum groups, the scattering matrix $S(u)$ plays the role of a transfer matrix while these expressions are essentially the Drinfeld currents; it is well-known in that context that a Poisson bracket of the form we have computed for $S(u)$ implies the standard Drinfeld-style relations on the $e(u), f(u), h(u)$ generators. 

The Poisson brackets on the invariants are recorded here, just for convenience of the readers. Recall that on any given moduli space of $U(2)$ monopoles with fixed locations of Dirac singularities we have the relation $Q(u) \widetilde{Q}(u) - U^+(u) U^-(u) = P(u)$ for given fixed polynomial $P(u)$ encoding the positions of Dirac singularities. The Poisson brackets are $\acomm{Q(u)}{Q(v)} = 0$ and
\begin{equation}
\begin{split}
    \acomm{Q(u)}{\frac{U^\pm(v)}{Q(v)}} & = \pm \frac{Q(u)}{u - v} \Bigg[\frac{U^\pm(u)}{Q(u)}  - \frac{U^\pm(v)}{Q(v)} \Bigg] \\
    \acomm{\frac{U^+(u)}{Q(u)}}{\frac{U^-(v)}{Q(v)}} & = \frac{1}{u - v} \Bigg[ \frac{P(v)}{Q(v)^2} - \frac{P(u)}{Q(u)^2}\Bigg]
\end{split}
\end{equation}
where we understand that wherever a $U^\pm(a)/Q(a)$ appears, we take the negative degree part of the expansion of both sides in powers of $1/a$. This accounts for the ``$\dots$'' above; their role is to cancel the positive powers appearing. 

For a sample calculation illustrating how these are obtained: 
\begin{equation}
\begin{split}
    \acomm{\frac{U^+(u)}{Q(u)}}{\frac{U^-(v)}{Q(v)}} & = \frac{1}{Q(u)} \acomm{U^+(u)}{\frac{U^-(v)}{Q(v)}} - U^+(u) \cdot \frac{1}{Q(u)^2} \acomm{Q(u)}{ \frac{U^-(v)}{Q(v)}} \\
    & = \frac{1}{Q(u) Q(v)} \acomm{S\indices{^1_2}(u)}{S\indices{^2_1}(v)} - \frac{1}{Q(u)} \acomm{S\indices{^1_2}(u)}{S\indices{^1_1}(v)} \frac{1}{Q(v)^2} U^-(v) \\
    & - \frac{U^+(u)}{Q(u)} \cdot \frac{1}{Q(u) Q(v)} \acomm{S\indices{^1_1}(u)}{S\indices{^2_1}(v)} \\
    & = \frac{1}{Q(u)Q(v)} \frac{1}{u - v} \Big[ S\indices{^1_1}(u) S\indices{^2_2}(v) - S\indices{^2_2}(u) S\indices{^1_1}(v) - \frac{U^-(v)}{Q(v)}(S\indices{^1_1}(u) S\indices{^1_2}(v) - S\indices{^1_2}(u) S\indices{^1_1}(v)) \\
    & - \frac{U^+(u)}{Q(u)} (S\indices{^1_1}(u) S\indices{^2_1}(v) - S\indices{^2_1}(u) S\indices{^1_1}(v)) \Big] \\
    & =  \frac{1}{u - v} \Big[ \frac{\widetilde{Q}(v)}{Q(v)} - \frac{\widetilde{Q}(u)}{Q(u)} - \frac{U^-(v) U^+(v)}{Q(v)^2} + \frac{U^+(u) U^-(v)}{Q(u)Q(v)} - \frac{U^+(u) U^-(v)}{Q(u) Q(v)} + \frac{U^+(u) U^-(u)}{Q(u)^2} \Big] \\
    & = \frac{1}{u - v} \Bigg[ \frac{P(v)}{Q(v)^2} - \frac{P(u)}{Q(u)^2} \Bigg].
\end{split}
\end{equation}
The first line is the product and chain rule for Poisson brackets applied in the first argument, the second is the same applied in the second argument together with $\acomm{Q(u)}{Q(v)} = 0$ and the evident fact that $Q(u) = S\indices{^1_1}(u)$, $U^+(u) = S\indices{^1_2}(u)$ etc, the third line uses \eqref{monopolerttbracket} and in going to the last line the relation $Q(u) \widetilde{Q}(u) - U^+(u) U^-(u) = P(u)$ was invoked.

\newpage 

\section{Recollections on free fermions and Grassmannians} \label{fermionreview}
In this appendix, we will give a self-contained overview of the aspects of free fermions and infinite Grassmannians which we will need, following \cite{miwa2000solitons}, \cite{Mulase2002ALGEBRAICTO}, \cite{segal-wilson}, \cite{witten}. Our goal is not to be complete or fully rigorous here, since this material is well-covered in existing literature, but just to say enough to make this paper self-contained. 

\subsection{Boundary conditions}
The place that Grassmannians enter the story for free fermions is via considering the space of possible boundary conditions for the free fermion field theory on the finite disk $D$ with coordinate $|z| \leq 1$: 
\begin{equation}
    S = \int_D \widetilde{\psi} \overline{\partial} \psi.
\end{equation}
The free fermion path integral on $D$ at least formally computes $\det \overline{\partial}$, but it is easy to see that without any boundary conditions on $\partial D$, $\overline{\partial}$ has infinitely many zero modes spanned by $z^n$ for $n \geq 0$. To get something that has a chance of being well-defined as a regularized infinite product, one must choose boundary conditions that eliminate all but perhaps finitely many zero modes. 

The boundary values of $\psi, \widetilde{\psi}$ are naturally required to lie in the vector space $\mathscr{H}$ of (roughly) smooth functions on $S^1 = \partial D$, so to remove all the zero modes one can require that $\psi \eval_{\partial D}, \widetilde{\psi} \eval_{\partial D} \in \mathscr{H}_-$ where $\mathscr{H}_-$ is the subspace of $\mathscr{H}$ generated by boundary values of $z^{-1}, z^{-2}, \dots$. 

To make this completely precise, one needs to view $\mathscr{H}$ as a completion of the space of Laurent polynomials in $z$, so that ``generation'' means generation under span and closure. Some references such as \cite{segal-wilson} follow an analytic approach and use the Hilbert space completion $\mathscr{H} = L^2(S^1)$. For the goals of this paper, it will be preferable to take a more algebraic approach and use formal Laurent series in $z^{-1}$, $\mathscr{H} = \mathbb{C}((z^{-1}))$ see e.g. \cite{Mulase2002ALGEBRAICTO}; this is also consistent with the approach of \cite{witten}. Physically, one can think of the Laurent series version as replacing the finite disk $D$ by the complex plane $\mathbb{C}$ and expanding functions on the $S^1$ ``at infinity''. In this setup, $\mathscr{H}_- = z^{-1} \mathbb{C}[[z^{-1}]] \subset \mathscr{H}$. 

We can contemplate having more general boundary conditions $\psi \eval_{\partial D} \in W \subset \mathscr{H}$ labeled by some subspace $W$ of $\mathscr{H}$. To get reasonable boundary conditions one wishes to have at most finitely many zero modes, so one requires that $\dim (W \cap \mathscr{H}_+) < \infty$, where $\mathscr{H}_+ = \mathbb{C}[z] \subset \mathscr{H}$. Likewise, for any subspace $W$ one may define its dual $\widetilde{W}$ with respect to the natural residue pairing ($\oint_{S^1}$ means residue at $z = \infty$)
\begin{equation}
    \widetilde{W} = \Bigg\{ f(z) \Bigg| \oint_{S^1} dz f(z) g(z) = 0 \, \, \, \forall g(z) \in W \Bigg\}
\end{equation}
and absence of boundary terms in integration by parts requires $\widetilde{\psi} \eval_{\partial D} \in \widetilde{W}$. Finitely many zero modes again requires $\dim (\widetilde{W} \cap \mathscr{H}_+) < \infty$. 

We can then make the following 

\begin{definition}
Sato's Grassmannian $\textnormal{Gr}(\mathscr{H})$ is defined to be the collection of all closed subspaces $W \subset \mathscr{H}$ such that the natural projection $W \to \mathscr{H}_-$ has finite-dimensional kernel and cokernel.
\end{definition}
It will be convenient to also consider the ``dual'' Grassmannian $\text{Gr}^*(\mathscr{H})$, consisting of closed subspaces $W \subset \mathscr{H}$ such that $W \to \mathscr{H}_+$ has finite-dimensional kernel and cokernel. 

Clearly we may think of $\text{Gr}(\mathscr{H})$ as the moduli space of boundary conditions of the free fermion CFT. For any $W \in \text{Gr}(\mathscr{H})$, we may define $\text{index}(W) := \dim (W \cap \mathscr{H}_+) - \dim( \widetilde{W} \cap \mathscr{H}_+)$. This is a locally constant function on $\text{Gr}(\mathscr{H})$, and we have a decomposition into connected components 
\begin{equation}
    \text{Gr}(\mathscr{H}) = \bigsqcup_{n \in \mathbb{Z}} \text{Gr}_n(\mathscr{H})
\end{equation}
with $n = \text{index}(W)$. The index zero connected component $\text{Gr}_0(\mathscr{H})$ is of principal interest, and may be thought of as the ``Grassmannian of middle-dimensional subspaces in $\mathbb{C}^\infty$''. 

We will be interested in doing algebraic geometry on $\text{Gr}(\mathscr{H})$. In fact, all we need is to be able to define the dual of the determinant line bundle $\text{DET}^*$ and its global sections. To this end, we approximate $\text{Gr}(\mathscr{H})$ by finite-dimensional Grassmannians, following \cite{Mulase2002ALGEBRAICTO}, \cite{witten} as follows. Define the $2n$-dimensional vector space $V_{2n} := \mathbb{C} z^{-n} \oplus \mathbb{C}z^{- n + 1} \oplus \dots \oplus \mathbb{C}z^{n - 1} \subset \mathscr{H}$. Denote 
\begin{equation}
    \text{Gr}(V_{2n}) := \bigsqcup_{k = 0}^{2n} \text{Gr}_k(V_{2n}). 
\end{equation}
where $\text{Gr}_k(V)$ is the usual Grassmannian of $k$-planes in $V$. $\text{Gr}(V_{2n})$ includes into $\text{Gr}(\mathscr{H})$ by defining, given a $U \in \text{Gr}(V_{2n})$, a subspace $W = z^{-n - 1} \mathbb{C}[[z^{-1} ]] \oplus U \subset \mathscr{H}$. Moreover these inclusions are clearly compatible with the natural inclusions $\text{Gr}_k(V_{2m}) \subset \text{Gr}_{k + n - m}(V_{2n})$ for $m \leq n$ so we have an inclusion of $\cup_{n \geq 0} \text{Gr}(V_{2n}) \subset \text{Gr}(\mathscr{H})$. 

In the analytic approach of \cite{segal-wilson}, $\text{Gr}(\mathscr{H})$ has the structure of an infinite-dimensional complex manifold and $\cup_{n \geq 0} \text{Gr}(V_{2n})$ is a dense subspace, so sections of vector bundles are determined by their restrictions to every $\text{Gr}(V_{2n})$. In the algebraic approach, it will suffice\footnote{For the functor of points approach to the Sato Grassmannian, see \cite{infinitegrassmannianfunctor}.} for our purposes to define $\text{DET}^*$ by the property that it restricts coherently to the usual determinant line bundle on each $\text{Gr}(V_{2n})$. This is the correct interpretation of the statement that the fiber of $\text{DET}$ over $W \subset \mathscr{H}$ should be the top exterior power of $W$; see the discussion around eq. 33 of \cite{witten}.  

\subsection{Determinant line bundle and Fock space}
To see that this leads to concrete and sensible results, it will be convenient to recall some basic facts about determinant line bundles on the finite-dimensional Grassmannians $\text{Gr}_k(V_{2n})$. The universal $k$-plane $\mathscr{E}$ sweeps out a vector bundle on $\text{Gr}_k(V)$ and $\text{DET} := \Lambda^k(\mathscr{E})$ is a line bundle. 

\begin{lemma} There is a canonical isomorphism
\begin{equation}
    H^0(\textnormal{Gr}_k(V_{2n}), \textnormal{DET}^*) \simeq \Lambda^k(V^*_{2n}).
\end{equation}
\end{lemma}
This is well-known and may be obtained, e.g. as a special case of the Borel-Weil-Bott theorem. In particular if we take the total Grassmannian $\text{Gr}(V_{2n}) = \bigsqcup_k \text{Gr}_k(V_{2n})$ we get the whole exterior algebra 
\begin{equation}
    H^0(\text{Gr}(V_{2n}), \text{DET}^*) \simeq \Lambda^\bullet(V^*_{2n}) 
\end{equation}
which can be thought of as the (finite-dimensional) Fock space for free fermions valued in $V$. For $m \leq n$ we have natural pullback maps $H^0(\text{Gr}(V_{2n}), \text{DET}^*) \to H^0(\text{Gr}(V_{2m}), \text{DET}^*)$ associated to the inclusions discussed above; it is a nice exercise to determine these maps in the exterior algebra language. 

By definition, the global sections of $\text{DET}^* \to \text{Gr}(\mathscr{H})$ is the inverse limit of the system consisting of the global sections over each $\text{Gr}(V_{2n})$, with morphisms given by the pullback maps associated to the inclusions $\text{Gr}(V_{2m}) \subset \text{Gr}(V_{2n})$ for $m < n$, 
\begin{equation}
    H^0(\text{Gr}(\mathscr{H}), \text{DET}^*) := \varprojlim H^0(\text{Gr}(V_{2n}), \text{DET}^*) \simeq \mathscr{F}.
\end{equation}
Using the exterior algebra description, this inverse limit may be described very explicitly using the standard fermionic Fock space $\mathscr{F}$, which may be thought of as consisting of semi-infinite differential forms on $\mathscr{H}$. It has the following description. Introduce operators $\psi_r, \widetilde{\psi}_r$ for $r \in \mathbb{Z} + 1/2$, with the standard anticommutation relations $\acomm{\psi_r}{\widetilde{\psi}_s} = \delta_{r + s, 0}$, $\acomm{\psi_r}{\psi_s} = \acomm{\widetilde{\psi}_r}{\widetilde{\psi}_s} = 0$. There is a unique vacuum vector $\ket{0} \in \mathscr{F}$ satisfying 
\begin{equation}
    \psi_r \ket{0} = \widetilde{\psi}_r \ket{0} = 0
\end{equation}
for $r > 0$, and $\mathscr{F}$ is generated\footnote{``Generate'' means under span and closure in the natural topology on the inverse limit. In plain English, we allow vectors in $\mathscr{F}$ to include infinite linear combinations of the standard basis vectors; we understand such an infinite linear combination as by definition its sequence of truncations to the basis vectors only including $\psi_{-r}, \widetilde{\psi}_{-r}$ for $r < N$ for some $N$. Each of the truncations is a finite sum.} by vectors of the form $\prod_i \psi_{-r_i} \prod_j \widetilde{\psi}_{-s_j} \ket{0}$, where $r_i, s_j > 0$. The connection with semi-infinite differential forms arises by viewing $\widetilde{\psi}_{-r}$ as wedging with $z^{-r - 1/2}$, $\psi_{-r}$ as contraction with $z^{r - 1/2}$, and $\ket{0}$ as $z^0 \wedge z^1 \wedge z^2 \wedge \dots$. 

We have a natural linear functional $\bra{0}$ on $\mathscr{F}$ which satisfies $\braket{0}{0} = 1$ and annihilates all the other standard basis vectors. $\bra{0}$ may be thought of as wedging with $\dots \wedge z^{-2} \wedge z^{-1}$ and taking the coefficient of $\dots \wedge z^{-2} \wedge z^{-1} \wedge z^0 \wedge z^1 \wedge \dots$. 
\subsubsection{Primitive sections}
Note that, because $\text{Gr}_k(V_{2n})$ is proper, any global function on it is necessarily constant. Therefore a global section of $\text{DET}^*$ may be characterized up to a scalar multiple by specifying its divisor of zeroes (the ratio of any two sections with the same zero set would be a function with no poles). 

In this way we may define canonical sections associated to subspaces $W \subset V_{2n}$ of dimension $2n - k$, which we call primitive sections following \cite{witten}. So let us fix some subspace $W$ of dimension $2n - k$. By dimension counting, for generic $U \in \text{Gr}_k(V_{2n})$, $\dim(W \cap U) = 0$. Then we may characterize a section $\sigma_W(U)$ up to overall scalars by insisting it has a simple zero on the divisor $\dim(U \cap W) > 0$, i.e. the locus where $U$ is not transverse to $W$. A formula for this section may be written as follows. Fix once and for all a volume form $\alpha$ on $V_{2n}$, e.g. by $\alpha = z^{-n} \wedge \dots \wedge z^{n - 1}$ using the standard basis of monomials in $V_{2n}$. Then, as $W$ is fixed, choose a volume form $\eta \in \Lambda^{2n - k}(W)$ on it as well. Then if $U = \text{Span}(u_1, \dots, u_k)$, we get a section 
\begin{equation}
    \sigma_W(U) = \frac{u_1 \wedge \dots \wedge u_k \wedge \eta}{\alpha}. 
\end{equation}
Of course the study of such sections is classical and related to the geometry of the Pl\"{u}cker embedding. Note that if we allow $W$ to vary in $\text{Gr}_{2n - k}(V)$, $\sigma_W$ varies as a section of $\text{DET}^* \boxtimes \text{DET}^*$ over $\text{Gr}_k(V) \times \text{Gr}_{2n - k}(V)$. 

The reader may verify that, for $n \to \infty$, the statement is that subspaces $W \in \text{Gr}^*(\mathscr{H})$ determine primitive sections supported on the component of $\text{Gr}(\mathscr{H})$ labeled by $-\text{index}(W)$, and that as $W$ varies these sections vary as sections of $\text{DET}^* \to \text{Gr}^*(\mathscr{H})$.

In the semi-infinite wedge description, primitive sections will be denoted $\ket{W}$ and admit the following construction, see \cite{Mulase2002ALGEBRAICTO}, \cite{segal-wilson} (for simplicity assume $\text{index}(W) = 0$). Since $W \in \text{Gr}^*_0(\mathscr{H})$ we may choose a basis $w_j(z)$, $j \geq 0$ for $W$ such that $w_j(z) = z^j(1 + O(z^{-1}))$ for all but finitely many $j$. Then 
\begin{equation}
    \ket{W} = w_0(z) \wedge w_1(z) \wedge w_2(z) \wedge \dots
\end{equation}
Since we assume that all but finitely many $w_j(z)$ are monic, the semi-infinite wedge can be expanded in terms of the standard basis states in $\mathscr{F}$ by the usual rules of linear algebra; each coefficient is well-defined but the total expression is in general an infinite linear combination of the standard basis states $z^{s_0} \wedge z^{s_1} \wedge \dots $ (cf. the footnote about the topology on the inverse limit above). 

\subsubsection{Clifford algebra characterization}
The following lemma is an immediate corollary of the semi-infinite wedge description of the primitive section $\ket{W}$. To set up the statement, we introduce some notation. Recall that $\mathscr{F}$ is a module over the Clifford algebra generated by the $\psi_r, \widetilde{\psi}_r$, $r \in \mathbb{Z} + 1/2$. Write the standard mode expansions
\begin{equation}
\begin{split}
    \psi(z) & = \sum_{r \in \mathbb{Z} + \frac{1}{2}} \frac{\psi_r}{z^{r + 1/2}} \\ 
    \widetilde{\psi}(z) & = \sum_{r \in \mathbb{Z} + \frac{1}{2}} \frac{\widetilde{\psi}_r}{z^{r + 1/2}}
\end{split}
\end{equation}
and define, for $f, g \in \mathscr{H}$
\begin{equation}
\begin{split}
    Q_f & = \oint_\gamma f(z) \psi(z) \\ 
    \widetilde{Q}_g & = \oint_\gamma g(z) \widetilde{\psi}(z)
\end{split}    
\end{equation}
where $\oint_\gamma$ denotes extraction of the residue at $z = \infty$. Finally, recall that $\widetilde{W}$ denotes the dual of $W$ with respect to the residue pairing. 

We have the ``Clifford algebra characterization'' of primitive sections:
\begin{lemma} \label{cliffward}
The state $\ket{W}$ corresponding to a primitive section $\sigma_W(U) \in H^0(\textnormal{Gr}(\mathscr{H}), \textnormal{DET}^*)$ is characterized up to scalars by $Q_f \ket{W} = 0 = \widetilde{Q}_g \ket{W}$, for all $f \in \widetilde{W}$ and $g \in W$. 
\end{lemma}

The above lemma gives a criterion for determining a primitive section $\ket{W}$ assuming the space $W$ is known. A related question is to determine primitive sections among all sections; this characterization is classical and provided by the Pl\"{u}cker relations. 

\begin{lemma}
A ray $\ket{\omega}$ in the fermionic Fock space is of the form $\ket{W}$ for some (necessarily unique) $W \in \textnormal{Gr}^*(\mathscr{H})$ if and only if it satisfies the quadratic relations
\begin{equation}
    \sum_{r \in \mathbb{Z} + 1/2} \psi_r \ket{\omega} \otimes \widetilde{\psi}_{-r} \ket{\omega} = 0. 
\end{equation}
\end{lemma}
Note that the operator entering the Pl\"{u}cker relations may be rewritten as 
\begin{equation}
    \sum_{r \in \mathbb{Z} + 1/2} \psi_r \otimes \widetilde{\psi}_{-r} = \oint_\gamma \frac{dz}{2\pi i} \psi(z) \otimes \widetilde{\psi}(z). 
\end{equation}
The Pl\"{u}cker relations give an explicit characterization of the well-known (at least in the finite-dimensional context) Pl\"{u}cker embedding $\text{Gr}^*(\mathscr{H}) \hookrightarrow \mathbb{P}(\mathscr{F})$.

\subsection{States from determinants}
The purpose of introducing this formalism is to provide a framework to understand certain fermion path integrals using geometry of $\text{Gr}(\mathscr{H})$; this eventually leads to a purely algebraic description of certain functional determinants. Let us describe how this works in a standard concrete historical example. For motivational purposes, and to make an explicit connection with quantum field theory, we will briefly return to the analytic setting. 

Free fermions can be defined on a compact Riemann surface $C$ on which we are given a line bundle $\mathscr{L} \to C$ with unitary connection $A$. The field $\psi(z)$ is valued in smooth sections of $K_C^{1/2} \otimes \mathscr{L}$ where $K_C^{1/2}$ denotes a choice of square root of the canonical bundle $K_C$, and $\widetilde{\psi}(z)$ is valued in smooth sections of $K_C^{1/2} \otimes \mathscr{L}^{-1}$. Let $\overline{\partial}_A: \Omega^{0, 0}(C, K_C^{1/2} \otimes \mathscr{L}) \to \Omega^{0, 1}(C, K_C^{1/2} \otimes \mathscr{L})$ denote the Cauchy-Riemann/Dirac operator valued in this bundle. The partition function of the free fermion theory on $C$ is formally expressed as the functional integral over the space of all fermionic smooth sections:
\begin{equation}
    \int D\widetilde{\psi} D \psi \exp( - \int_C \widetilde{\psi} \overline{\partial}_A \psi ) = \det \overline{\partial}_A.
\end{equation}
To understand the determinant in the context of the formalism discussed so far, we excise a small neighborhood of a point $p \in C$. We choose a local coordinate $z^{-1}$ near $p$, so that $z \to \infty$ in local coordinates corresponds to $p$. We also assume $\mathscr{L}$ is trivialized in this neighborhood of $p$. Deleting the neighborhood of $p$ produces a Riemann surface $\Sigma$ with a single $S^1$ boundary component, such that $K_C^{1/2} \otimes \mathscr{L}$ is naturally trivialized on the boundary, so we may choose a boundary condition labeled by a subspace $U \in \text{Gr}(\mathscr{H})$. We denote $\overline{\partial}_A^U$ the Dirac operator acting on smooth sections of $K_C^{1/2} \otimes \mathscr{L}$ with boundary values in $U$ to remind us of its dependence on $U$. The path integral over fields on $\Sigma$ with boundary condition $U$ computes, formally, $\det \overline{\partial}_A^U$. 

Giving a direct definition of $\det \overline{\partial}_A^U$ typically involves some complicated regularization procedure. Independent of the details, it is clear that any reasonable definition of $\det \overline{\partial}_A^U$ must satisfy two properties: 
\begin{enumerate}
    \item As $U$ varies in $\text{Gr}(\mathscr{H})$, $\det \overline{\partial}_A^U$ varies as a holomorphic section of $\text{DET}^*$, i.e. if we work over all $U$, the determinant is canonically an element of $H^0(\text{Gr}(\mathscr{H}), \text{DET}^*)$. 

    \item As a section of $\text{DET}^*$, $\det \overline{\partial}_A^U$ vanishes precisely when there exists a $\psi \in \Omega^{0, 0}(\Sigma, K_\Sigma^{1/2} \otimes \mathscr{L})$ such that $\overline{\partial}_A \psi = 0$ and $\psi \eval_{\partial \Sigma} \in U$. 
\end{enumerate}
Informally, the first condition says that the determinant is a determinant and the second condition says that the determinant vanishes if and only if the corresponding operator has a nonzero kernel. 

From the discussion above, the second condition in fact uniquely determines the functional determinant up to an overall scalar multiple (in this section we will disregard such overall scalars): it is the primitive section of $\text{DET}^* \to \text{Gr}(\mathscr{H})$ associated to the space $W \subset \mathscr{H}$ of boundary values of holomorphic sections of $K_C^{1/2} \otimes \mathscr{L}$. We denote the corresponding ray in the fermionic Fock space as $\ket{W} \in \mathbb{P}(\mathscr{F})$; we will sometimes abuse notation and denote by $\ket{W}$ also a lift of this ray to a vector in $\mathscr{F}$.

This logic can be inverted to give a definition of the functional determinants in the algebra-geometric setting. Now $C$ is a smooth projective curve, $p \in C$ is a point, $z^{-1}$ is a formal coordinate at $p$ and $\Sigma = C \setminus p$; instead of restriction to $\partial \Sigma$, we restrict to the formal punctured disk $\mathbb{D}^\times_p$ at $p$, and assume $\mathscr{L}$ is trivialized over the formal disk $\mathbb{D}_p$. Expansion at $p$ gives rise to a map 
\begin{equation}
    H^0(C \setminus p, K^{1/2}_C \otimes \mathscr{L}) \hookrightarrow \mathbb{C}(( z^{-1})) = \mathscr{H}. 
\end{equation}
The closure of its image is denoted $W \subset \mathscr{H}$; in fact $W$ defines a point of $\text{Gr}^*(\mathscr{H})$ with $\text{index}(W) = \deg \mathscr{L} - 1$. This is shown by the following consideration: since $C \setminus p$ is an affine variety, we may compute cohomology using the Cech complex associated to the (formal) covering $C = (C \setminus p) \cup \mathbb{D}_p$, and unwinding definitions we see there are natural isomorphisms 
\begin{equation}
\begin{split}
    \ker(W \to \mathscr{H}_+) & \simeq H^0(C, K_C^{1/2} \otimes \mathscr{L}(-p) ) \\ 
    \text{coker}(W \to \mathscr{H}_+) & \simeq H^1(C, K_C^{1/2} \otimes \mathscr{L}(-p))
\end{split}
\end{equation}
from which it follows that each of these are finite dimensional (as $C$ is projective) and $\text{index}(W) = \deg \mathscr{L} - 1$ (from Riemann-Roch). Therefore we have a well-defined primitive state $\ket{W}$ assigned to $W$ and make this a 
\begin{definition}
    Notations as above, the functional determinant $\det \overline{\partial}_A^U$ is defined as the primitive state $\ket{W} \in \mathbb{P}(\mathscr{F})$ assigned to $H^0(C \setminus p, K_C^{1/2} \otimes \mathscr{L})$. 
\end{definition}
Note that to recover the numerical determinant $\det \overline{\partial}_A$ in this approach, we just evaluate this section at the point $U = \mathbb{C}[[z^{-1}]] \in \text{Gr}_1(\mathscr{H})$, using the natural trivialization of $\text{DET}^*$ over the big cell. To state meaningful results for numerical determinants it is of course important to pin down the normalization of $\ket{W}$ by some other means. 

There is a more geometric way to describe how $\ket{W}$ is produced. Let $\mathfrak{M}$ denote the moduli stack of data $(C, \mathscr{L}, p, z,  \phi_p)$ with notations as above ($\phi_p$ denotes a trivialization of $\mathscr{L}$ in the formal neighborhood of $p$). The extraction of $W$ from this data defines a morphism 
\begin{equation}
    \mathfrak{K}: \mathfrak{M} \to \text{Gr}^*(\mathscr{H})
\end{equation}
called the Krichever map, and composing it with the Pl\"{u}cker embedding $\text{Gr}^*(\mathscr{H}) \hookrightarrow \mathbb{P}(\mathscr{F})$ produces the state $\ket{W}$ assigned to a point of $\mathfrak{M}$. It is this structure that we will imitate in the instanton problem. 

\subsubsection{Example}
To illustrate the language, let's consider the path integral on $\mathbb{C} = \mathbb{P}^1 \setminus \infty$, $\mathscr{L}$ trivial, and $K_C^{1/2} \simeq \mathscr{O}(-z_1)$ locally trivialized in such a way as to be identified with the sheaf of functions vanishing at some point $z_1 \in \mathbb{P}^1 \setminus \infty$. Write $W(z_1) = H^0(\mathbb{P}^1 \setminus \infty, \mathscr{O}(-z_1))$; by the above construction we are supposed to produce a state $\ket{W(z_1)}$. It is instructive to understand this state explicitly. 

$W(z_1)$ has a basis given by $w_j(z) = (z - z_1)^j$, for $j \geq 1$. Then 
\begin{equation}
\begin{split}
        \ket{W(z_1)} & = (z - z_1) \wedge (z - z_1)^2 \wedge (z - z_1)^3 \wedge \dots \\ 
        & = (z - z_1) \wedge (z^2 - z_1 z) \wedge (z^3 - z_1 z^2) \wedge \dots \\
        & = z \wedge z^2 \wedge z^3 \wedge \dots - z_1 \cdot z^0 \wedge z^2 \wedge z^3 \wedge \dots  \\ 
        & + z_1^2 \cdot z^0 \wedge z \wedge z^3 \wedge \dots  + O(z_1^3) \\
        & = \sum_{r = 1/2}^\infty \psi_{-r} z_1^{r - 1/2} \ket{0} = \psi(z_1) \ket{0}
\end{split}
\end{equation}
where we recall that $\psi_{-r}$ is contraction with $z^{r - 1/2}$ and $\ket{0} = z^0 \wedge z^1 \wedge \dots$. $\widetilde{W}(z_1)$ is identified with functions that have at most a simple pole at $z_1$; Lemma \ref{cliffward} is checked by noting that 
\begin{equation}
\begin{split}
    Q_f \ket{W(z_1)} & = \oint_{\gamma_{z_1}} dz f(z) \psi(z) \psi(z_1) \ket{0} = 0, \, \, \, \text{if $(z - z_1)f(z)$ is regular at $z_1$} \\ 
    \widetilde{Q}_g \ket{W(z_1)} & = \oint_{\gamma_{z_1}} dz g(z) \widetilde{\psi}(z) \psi(z_1) \ket{0} = 0, \, \, \, \text{if $g(z_1) = 0$.}
\end{split}
\end{equation}
We used the residue theorem to shrink the contour to a small loop $\gamma_{z_1}$ enclosing $z_1$, and the operator product expansions 
\begin{equation}
\begin{split}
    \psi(z) \psi(z_1)  & = O(z - z_1) \\ 
    \widetilde{\psi}(z) \psi(z_1) & = \frac{1}{z - z_1} + \text{reg}
\end{split}
\end{equation}
as $z \to z_1$. The reader may check as an exercise that $\ket{\widetilde{W}(z_1)} = \widetilde{\psi}(z_1) \ket{0}$. 

In \cite{witten} it was observed that twisting/Hecke transforming a line bundle at a point on $C$ is equivalent in the above sense to inserting $\psi$ or $\widetilde{\psi}$ at that point; this was interpreted as a ``multiplicative Ward identity'' in the free fermion conformal field theory. In essence, it says that the operation of Hecke transformation is realized as a genuine local operator in the free fermion theory, and this operator is nothing but insertion of the basic local fields $\psi(z)$, $\widetilde{\psi}(z)$. This is just a geometric interpretation of the basic operator product expansions. 

One of the main results of this paper can be phrased as a generalization of this, in the following direction. Instead of an abstract algebraic curve $C$, we consider an embedded curve in a noncommutative algebraic surface, as in \cite{Aganagic_2005}, \cite{DHS}, \cite{DHSV}; say the zero section in $T^*_\hbar \mathbb{C}$. Noncommutative $U(1)$ instantons on $T^*_\hbar \mathbb{C}$ can be thought of as rank 1 sheaves on $T^*_\hbar \mathbb{C}$ with full support, obtained by taking the trivial line bundle and performing some complex two-dimensional analog of a Hecke transformation at several points $(z_i, p_i) \in T^* \mathbb{C}$. If we place the free fermion theory on the zero section $\mathbb{C} \subset T^* \mathbb{C}$, then the ``Hecke transformation'' at a point $(z_i, p_i)$ on the surface $T^*_\hbar \mathbb{C}$ is implemented by insertion of the local operator 
\begin{equation}
    p_i - \hbar \partial \varphi(z_i)
\end{equation}
where $\partial \varphi(z)$ is the chiral boson equivalent to the fermions under bosonization, see below. Note that if the noncommutativity is sent to zero, $\hbar \to 0$, the local operator becomes essentially trivial, as expected since the instanton can simply be moved away from the zero section. The nonlocality introduced by the noncommutativity is essential for this formula to make sense. 

For more general holomorphic Lagrangians in $T^* \mathbb{C}$, the nonlocality of the free fermion theory in this situation is more difficult to suppress (see \cite{Aganagic_2005}, \cite{DHS}, \cite{DHSV}), so the corresponding statement becomes more complicated, though it is a natural generalization of the construction of \cite{Aganagic_2005}, \cite{DHS}, \cite{DHSV} that works over the moduli space of instantons.

\subsection{Bosonization and coherent state basis}
Thus far, the primitive sections $\ket{W}$ have been described by considering their expansion in the standard basis of the free fermion Fock space coming from the action of negative modes of $\psi$ and $\widetilde{\psi}$ on the vacuum; this amounts to giving explicit Pl\"{u}cker coordinates for the subspace $W \subset \mathscr{H}$. For certain applications, it is convenient to use a ``bosonized'' description of the states, which proceeds as follows. We assume basic familiarity with bosonization, see e.g. \cite{miwa2000solitons} or \cite{qftias}, pg. 1202-1205 for textbook accounts.

\subsubsection{Bosonic description of $\mathscr{F}$}
Because $\text{Gr}(\mathscr{H})$ has $\mathbb{Z}$-many connected components, the space of global sections of $\text{DET}^*$ splits as a sum of sections supported on each component, 
\begin{equation}
    H^0(\text{Gr}(\mathscr{H}), \text{DET}^*) \simeq \mathscr{F} = \bigoplus_{N \in \mathbb{Z}} \mathscr{F}_N.
\end{equation}
The modes of (colons denote normal ordering)
\begin{equation}
    : \psi(z) \widetilde{\psi}(z) : \, \, = \sum_{n \in \mathbb{Z}} \frac{\alpha_n}{z^{n + 1}} = \partial \varphi(z)
\end{equation}
generate a Heisenberg algebra 
\begin{equation}
    \comm{\alpha_m}{\alpha_n} = m \delta_{m + n, 0}
\end{equation}
and each $\mathscr{F}_N$ is isomorphic to a Fock module of the Heisenberg algebra (which we will call bosonic Fock space) where the zero mode $\alpha_0$ acts by $N$, 
\begin{equation}
    \mathscr{F}_N = \mathbb{C}[\alpha_{-1}, \alpha_{-2, \dots}] \ket{N}.
\end{equation}
The state $\ket{N} = z^N \wedge z^{N + 1} \wedge z^{N + 2} \wedge \dots$ in terms of the semi-infinite wedge description. We will usually concentrate on the bosonic Fock space $\mathscr{F}_0$ corresponding to the index zero component in this paper. Note our notation is consistent in the sense that the fermionic vacuum $\ket{0} = \ket{N = 0}$.  

\subsubsection{Loop group action on $\text{Gr}(\mathscr{H})$}
Define the operator (all but finitely many $t_n$ must be zero)
\begin{equation}
    g(t) := \exp \Big(\sum_{n = 1}^\infty t_n \alpha_n \Big) 
\end{equation}
then the assignment $\ket{\omega} \mapsto \bra{0}g(t) \ket{\omega}$ identifies $\mathscr{F}_0$ with functions of infinitely many variables $t_i$ in terms of which, for $n > 0$, 
\begin{equation}
\begin{split}
    \alpha_n & \to \pdv{}{t_n} \\
    \alpha_{-n} & \to nt_n. 
\end{split}
\end{equation}
We will define $\bra{t} := \bra{0} g(t)$ and view the function associated to the state as its expression in the coherent state basis for the operators $\alpha_n$. The wavefunction of a primitive state $\ket{W}$ 
\begin{equation}
    \tau_W(t) = \braket{t}{W}
\end{equation}
is called its tau function for historical reasons related to integrable hierarchies that will play a minimal role for us. However it is an economical way to encode the data of $\ket{W}$.  

Note that the operator $g(t)$ corresponds to the action of half of the loop group of $GL_1$ on $\text{Gr}^*(\mathscr{H})$: for $W \in \text{Gr}^*(\mathscr{H})$, the would-be action is given by taking any basis $w_j(z)$ of $\ket{W}$ of the kind above, and replacing $w_j(z) \mapsto w_j(z) e^{\sum_{n = 1}^\infty t_n z^n}$. However, since we use $\mathscr{H} = \mathbb{C}((z^{-1}))$, this formula does not preserve $\mathscr{H}$ and therefore is not well-defined. One can circumvent this by using the Pl\"{u}cker embedding, the fact that the action on $\mathscr{F}$ is well-defined using $g(t)$, and observing that $g(t) \ket{W}$ satisfies the Pl\"{u}cker relations if $\ket{W}$ does. In this way one obtains an action of this group (called $\Gamma_+$ in \cite{segal-wilson}) on $\text{Gr}^*(\mathscr{H}) \hookrightarrow \mathbb{P}(\mathscr{F})$. We will denote the image of $W$ under the action as $g(t) \cdot W$, so $g(t) \ket{W} = \ket{g(t) \cdot W}$. 

In the classical theory of Krichever solutions (i.e. when $W = H^0(C \setminus p, K_C^{1/2} \otimes \mathscr{L})$ as above), the above characterizations can be used to deduce explicit formulas for $\tau_W(t)$ using the theta functions of the curve $C$; we will not review them here as this is outside our primary narrative, but refer to \cite{dubrovin}, or Sections 6 and 9 of \cite{segal-wilson} for details. Our formulas \eqref{formula:ISMcoefficients}, \eqref{formula:cylindermiuracoeff} can be similarly viewed in this spirit, but for us the interesting moduli space is the moduli of instantons, while the curve is usually fixed and rational. As such, the expression is given by a simple finite-dimensional determinant of the ADHM matrices. 

\subsubsection{Higher rank analogs}
In this appendix we have reviewed the theory ofa single free fermion $\psi(z)$ and its conjugate $\widetilde{\psi}(z)$. There is a natural generalization of much of the story to $r$-component free fermions, $\psi^a(z)$, $\widetilde{\psi}_a(z)$, $a = 1, \dots, r$. All discussion of the Grassmannian goes through essentially unchanged with the replacements $\mathscr{H} = \mathbb{C}^r((z^{-1}))$, $\mathscr{H}_+ = \mathbb{C}^r[z]$, $\mathscr{H}_- = z^{-1} \mathbb{C}^r[[z^{-1} ]]$. Certain new features emerge, however; for instance one may consider Cauchy-Riemann operators $\overline{\partial}_A$ coupled to more general vector bundles (and our noncommutative instanton operators should be thought of in this form). Likewise, the bosonic description of the Fock space is more involved, and uses the integrable representations of $\widehat{\mathfrak{gl}}_r$ at level $k = 1$ (this is the nonabelian bosonization of \cite{witten84}). This Kac-Moody algebra is generated by the currents 
\begin{equation}
    J\indices{^a_b}(z) = :\psi^a(z) \widetilde{\psi}_b(z):.
\end{equation}

\newpage

\section{Details with ADHM complex} \label{adhmdetails}
In this technical appendix we will prove various facts about the noncommutative instanton bundles $\mathscr{E}$ that rely on their description using the ADHM construction. 

\subsubsection{Stability and costability}
Let $V, W$ denote the vector spaces of dimension $n$ and $r$ entering the ADHM description of $\widetilde{M}(n, r)$. We recall that the ADHM data $(B_1, B_2, I, J)$ must solve the moment map equation 
\begin{equation} \label{momentmapeq}
    \comm{B_1}{B_2} + IJ = \hbar. 
\end{equation}
We have the stability and co-stability conditions that may be satisfied by such a quadruple $(B_1, B_2, I, J)$: 
\begin{itemize}
    \item Stability condition: $S \subseteq V$ is a subspace such that $\text{Im} \,  I \subseteq S$ and $S$ is an invariant subspace for $B_1$ and $B_2 \implies S = V$. 
    \item Costability condition: $S \subseteq V$ is a subspace such that $S \subseteq \ker J$ and $S$ is an invariant subspace for $B_1$ and $B_2 \implies S = 0$. 
\end{itemize}
The following is well-known in the context of Nakajima quiver varieties, but we give the proof here. 

\begin{lemma} \label{stabandcostab}
If $(B_1, B_2, I, J)$ solve \eqref{momentmapeq} with $\hbar \neq 0$, then they satisfy both the stability and costability conditions. 
\end{lemma}

\begin{proof}
Toward the costability condition, assume $S \subseteq V$ is $B_1, B_2$ invariant and contained in $\ker J$. Then the left hand side of \eqref{momentmapeq} preserves $S$ and coincides there with $\comm{B_1}{B_2} |_S = \comm{B_1 |_S}{B_2 |_S}$, so taking trace we must have $ 0 = \hbar \dim S$. If $\hbar \neq 0$, $\dim S = 0$. 

Toward the stability condition, assume $S \subseteq V$ is $B_1, B_2$ invariant and contains $\text{Im} \, I$. Let $\phi \in \text{Ann}(S) \subseteq V^*$; then precomposing $\phi$ with \eqref{momentmapeq} we must have $\phi \cdot\comm{B_1}{B_2} = \hbar \phi$; then $\text{Ann}(S)$ must carry a finite-dimensional representation of the Weyl algebra so $\text{Ann}(S) = 0$, or $S = V$. 
\end{proof}

The following is a useful corollary of the stability condition. Let $e_i, i = 1, \dots, r$ denote any basis in $W$.
\begin{lemma} \label{stabilitybasis}
    When $\hbar \neq 0$, we have $ V= \textnormal{Span}(B_2^m B_1^n I(e_i))_{m, n \geq 0, 1 \leq i \leq r}$. 
\end{lemma}

\begin{proof}
    Let $S \subseteq V$ denote the span of the $B_2^m B_1^n I(e_i)$. Clearly $\text{Im} \, I \subseteq S$ and $S$ is $B_2$-invariant. To see that $S$ is also $B_1$-invariant, compute  
    \begin{equation}
    \begin{split}
        B_1 B_2^m B_1^n I(e_i) & = \sum_{k = 0}^{m - 1} B_2^k \comm{B_1}{B_2} B_2^{m - 1 - k} B_1^n I(e_i) + B_2^m B_1^{n + 1} I(e_i) \\ 
        & = \hbar m B_2^{m - 1} B_1^nI(e_i) - \sum_{k = 0}^{m - 1} B_2^k I(J B_2^{m - 1 - k} B_1^n I(e_i)) + B_2^m B_1^{n + 1} I(e_i) \in S
    \end{split}
    \end{equation}
    where we used \eqref{momentmapeq} and the assumption that $e_i$ is a basis of $W$. By the stability condition, $S = V$.
\end{proof}

\subsubsection{ADHM complex}
Part of the fundamental assertion of the noncommutative ADHM construction is that the instanton bundle $\mathscr{E}$, as a right $\mathscr{D}_\hbar$-module, coincides with the middle cohomology of the complex 
\begin{equation} \label{eq:adhmcomplex}
\begin{tikzcd}
 C^\bullet := V \otimes \mathscr{D}_\hbar \arrow[r, "\alpha"] & (V \oplus V \oplus W) \otimes \mathscr{D}_\hbar \arrow[r, "\beta"] & V \otimes \mathscr{D}_\hbar 
\end{tikzcd}
\end{equation}
The maps in the complex are 
\begin{equation}
\begin{split}
    \alpha & = \begin{pmatrix} B_1 - z \\ B_2 - w \\ J \end{pmatrix} \\ 
    \beta & = \begin{pmatrix} -B_2 + w && B_1 - z && I  \end{pmatrix}.
\end{split}
\end{equation}
Whenever the elements $z, w \in \mathscr{D}_\hbar$ appear in $\alpha$ or $\beta$, we understand the operators of left multiplication by these elements, so that the morphisms in the complex are indeed morphisms of right $\mathscr{D}_\hbar$-modules. Note that $\beta \alpha = 0$ using \eqref{momentmapeq} and $\comm{z}{w} = - \hbar$. First we have 
\begin{lemma} \label{adhmcohvanish}
    For $\hbar \neq 0$, the cohomology of $C^\bullet$ is entirely concentrated in the middle degree, i.e. $H^0(C^\bullet) = H^2(C^\bullet) = 0$. 
\end{lemma}

\begin{proof}
$H^0(C^\bullet)$ consists of those $f(z, w) \in V \otimes \mathscr{D}_\hbar$ such that $(B_1 - z) f = 0 = (B_2 - w) f$, and $Jf(z, w) = 0$. Expanding $f$ in the monomial basis of $\mathscr{D}_\hbar$, $f(z, w) = \sum_{m, n} f_{m, n} z^m w^n$, $f_{m, n} \in V$, these equations imply that $\text{Span}(f_{m, n})_{m, n \geq 0} \subseteq V$ is a $B_1, B_2$ invariant subspace of $\ker J$. By the costability condition, $S = 0$ and we have $f_{m, n} = 0$ for all $m, n$, thus $f(z, w) = 0$. 

To see the vanishing of $H^2(C^\bullet)$, take any $f(z, w) \in V \otimes \mathscr{D}_\hbar$ and expand $f(z, w) = \sum_{m, n} f_{m, n} z^m w^n$ as above. Now apply Lemma \ref{stabilitybasis} to get an expansion 
\begin{equation}
    f_{m, n} = \sum_{p, q, i} C_{mnpqi} B_2^p B_1^q I(e_i) \in V
\end{equation}
and observe 
\begin{equation}
\begin{split}
    f(z, w) & = \sum_{m, n, p, q, i} C_{mnpqi} (B_2 - w + w)^p B_1^q I(e_i) z^m w^n \\ 
    & = \sum_{m, n, p, q, i} C_{mnpqi} (B_1 - z+ z)^q I(e_i) w^p z^m w^n \, \, \, \mod (B_2 - w) V \otimes \mathscr{D}_\hbar \\ 
    & = \sum_{m, n, p, q, i} C_{mnpqi} I(e_i) z^q w^p z^m w^n \, \, \, \mod (B_1 - z) V \otimes \mathscr{D}_\hbar \\ 
    & = 0 \, \, \, \mod I (W \otimes \mathscr{D}_\hbar)
\end{split}
\end{equation}
showing that $f(z, w) = 0 \in \text{coker} \, \beta$. Since $f$ was arbitrary, $\text{coker} \, \beta = 0$. 
\end{proof}
Note that if we were to consider a similar complex but with $\mathscr{D}_\hbar$ replaced by any left $\mathscr{D}_\hbar$-module of countable dimension over $\mathbb{C}$ (in particular any cyclic $\mathscr{D}_\hbar$ module), the same argument implies the cohomology vanishes outside the middle degree (replace the monomial basis wherever it appears by some choice of basis in the vector space underlying the module). This observation will be important in the proof of Proposition \ref{extgroups} below. 

\subsubsection{Computing Ext groups}
The ADHM description of the $\mathscr{D}_\hbar$-module $\mathscr{E}$ is useful for giving an explicit description of the kind of $\text{Ext}$-groups we need to understand in this paper. Let $P(z, w) \in \mathscr{D}_\hbar$ be any nonzero element. We have the right and left $\mathscr{D}_\hbar$-modules: 
\begin{equation}
\begin{split}
    \widehat{\mathscr{O}}_C & : = \mathscr{D}_\hbar/ P(z, w) \mathscr{D}_\hbar \\ 
    \widehat{\mathscr{O}}^\ell_C & := \mathscr{D}_\hbar/\mathscr{D}_\hbar P(z, w). 
\end{split}
\end{equation}
The notation signifies that these modules should be viewed as noncommutative analogs of the structure sheaves of the affine curve $C$ defined by $P(z, w) = 0$ in the $\hbar \to 0$ limit. 

Define the complex

\begin{equation}
\begin{tikzcd}
    C^\bullet_P := V \otimes \widehat{\mathscr{O}}^\ell_C \arrow[r, "\alpha"] & (V \oplus V \oplus W) \otimes \widehat{\mathscr{O}}^\ell_C \arrow[r, "\beta"] & V \otimes \widehat{\mathscr{O}}^\ell_C.
\end{tikzcd}
\end{equation}
$\alpha$ and $\beta$ denote the natural maps induced by the maps in the ADHM complex $C^\bullet$, which by abuse of notation we continue to call $\alpha$ and $\beta$. Note that the cohomology of $C^\bullet_P$ vanishes outside the middle degree by the remark following Lemma \ref{adhmcohvanish}. 

\begin{prop} \label{extgroups}
    We have 
    \begin{equation}
    \begin{split}
        \textnormal{Hom}_{\mathscr{D}_\hbar}(\widehat{\mathscr{O}}_C, \mathscr{E}) & = 0 \\ 
        \textnormal{Ext}^1_{\mathscr{D}_\hbar}(\widehat{\mathscr{O}}_C, \mathscr{E}) & \simeq H^1(C^\bullet_P). 
    \end{split}
    \end{equation}
\end{prop}

\begin{proof}
Consider the double complex
\begin{equation}
E^{\bullet, \bullet} := 
\begin{tikzcd}
    & V \otimes \mathscr{D}_\hbar \arrow[d, "\cdot P"] \arrow[r, "\alpha"] & (V \oplus V \oplus W) \otimes \mathscr{D}_\hbar \arrow[d, "\cdot P"] \arrow[r, "\beta"] & V \otimes \mathscr{D}_\hbar \arrow[d, "\cdot P"] \\
    & V \otimes \mathscr{D}_\hbar \arrow[r, "\alpha"] & (V \oplus V \oplus W) \otimes \mathscr{D}_\hbar \arrow[r, "\beta"] & V \otimes \mathscr{D}_\hbar. 
\end{tikzcd}
\end{equation}
Vertical arrows are given by \textit{right} multiplication by $P$. There are two spectral sequences that compute the cohomology of $E^{\bullet, \bullet}$; comparing them will yield the statement of the proposition. 

First consider the spectral sequence with horizontal orientation. Invoking Lemma \ref{adhmcohvanish}, it has $E_1$ page 
\begin{equation}
\begin{tikzcd}
    & 0 \arrow[r] \arrow[d] & \mathscr{E} \arrow[d, " \cdot P "] \arrow[r] & 0 \arrow[d] \\
    & 0 \arrow[r] & \mathscr{E} \arrow[r] & 0. 
\end{tikzcd}
\end{equation}
We see that, for degree reasons, the spectral sequence degenerates at the $E_2$ page. Comparing to the computation of $\text{Ext}^\bullet_{\mathscr{D}_\hbar}(\widehat{\mathscr{O}}_C, \mathscr{E})$ using the free resolution 
\begin{equation}
\begin{tikzcd}
    0 \arrow[r] & \mathscr{D}_\hbar \arrow[r, "P \cdot"] & \mathscr{D}_\hbar \arrow[r] & \widehat{\mathscr{O}}_C \arrow[r] & 0 
\end{tikzcd}
\end{equation}
(where the map is \textit{left} multiplication by $P$) we see that the only a priori nonvanishing cohomologies are $H^1(E^{\bullet, \bullet}) \simeq \text{Ext}^1_{\mathscr{D}_\hbar}(\widehat{\mathscr{O}}_C, \mathscr{E})$, and $H^2(E^{\bullet, \bullet}) \simeq \text{Hom}_{\mathscr{D}_\hbar}(\widehat{\mathscr{O}}_C, \mathscr{E})$. 

On the other hand, consider the spectral sequence with vertical orientation. Since $\mathscr{D}_\hbar$ has no zero divisors, the $E_1$ page reads 
\begin{equation}
\begin{tikzcd}
    & 0 \arrow[r] \arrow[d] & 0 \arrow[r] \arrow[d] & 0 \arrow[d] \\
    & V \otimes \widehat{\mathscr{O}}^\ell_C \arrow[r, "\alpha"] & (V \oplus V \oplus W) \otimes \widehat{\mathscr{O}}^\ell_C \arrow[r, "\beta"] & V \otimes \widehat{\mathscr{O}}^\ell_C.
\end{tikzcd}
\end{equation}
We see that once again, for degree reasons the spectral sequence degenerates at the $E_2$ page, and that $H^i(E^{\bullet, \bullet}) \simeq H^i(C^\bullet_P)$. Using that the cohomology of $C^\bullet_P$ vanishes outside the middle degree yields the result.  
\end{proof}

\subsubsection{Explicit description when $P(z, w) = w$}
We can give an even more explicit description of the relevant $\text{Ext}$ in the case where $C$ is the zero section in $T^*_\hbar \mathbb{C}$, i.e. $P(z, w) = w$. In this case, 
\begin{equation}
    \widehat{\mathscr{O}}_C^\ell \simeq \mathbb{C}[z]
\end{equation}
with $z$ acting by multiplication and $w$ acting by $\hbar \partial_z$. In this case there is the following result. 

\begin{prop} \label{explicitext}
We have an inclusion $\textnormal{Ext}^1_{\mathscr{D}_\hbar}(\widehat{\mathscr{O}}_{\mathbb{C}}, \mathscr{E}) \hookrightarrow W[z] \oplus V$ as the locus of $(\xi(z), \psi_+) \in W[z] \oplus V$ such that $\xi(B_1) I = B_2 \psi_+$. 
\end{prop}
Here, by $\xi(B_1)I$ we mean the following. Expand $\xi(z) = \sum_n \xi_n z^n$ uniquely for some $\xi_n \in W$; then $\xi(B_1) I: = \sum_n B_1^n I\xi_n \in V$. 

Note that the above inclusion also holds on the level of vector bundles over the instanton moduli space, i.e. it gives us a short exact sequence 
\begin{equation}
\begin{tikzcd}
    0 \arrow[r] & \text{Ext}^1_{\mathscr{D}_\hbar}(\widehat{\mathscr{O}}_{\mathbb{C}}, \mathscr{E}) \arrow[r] & \mathscr{W}[z] \oplus \mathscr{V} \arrow[r] & \mathscr{V} \arrow[r] & 0 
\end{tikzcd}
\end{equation}
where the second nonzero map is $(\xi(z), \psi_+) \mapsto \xi(B_1)I - B_2 \psi_+$, with $B_1, B_2, I$ now interpreted as the corresponding tautological sections. 

\begin{proof}
    The proof amounts to an analysis of $H^1(C^\bullet_P)$ in the case $P(z, w) = w$. In this case, a cohomology class is the same as a triple $(\psi_+(z), \psi_-(z), \xi(z)) \in V[z] \oplus V[z] \oplus W[z]$ satisfying the equation 
    \begin{equation} \label{kerbetaeq}
        -(B_2 - \hbar \partial_z ) \psi_+(z) + (B_1 - z) \psi_-(z) + I \xi(z) = 0 
    \end{equation}
    and considered modulo 
    \begin{equation}
    \begin{split}
        \psi_+(z) & \sim \psi_+(z) + (B_1 - z) a(z) \\ 
        \psi_-(z) & \sim \psi_-(z) + (B_2 - \hbar \partial_z) a(z) \\ 
        \xi(z) & \sim \xi(z)  + Ja(z) 
    \end{split}
    \end{equation}
    with $a(z) \in V[z]$. First observe that any equivalence class has a unique representative such that $\psi_+(z) = \text{constant}$: simply expand, for $\psi_{+, n} \in V$
    \begin{equation}
        \psi_+(z) = \sum_n \psi_{+, n} z^n = \sum_n (z - B_1 + B_1)^n\psi_{+, n}  = \sum_n B_1^n \psi_{+, n} + (z - B_1)(\dots)
    \end{equation}
    which determines $a(z)$ uniquely since $z - B_1: V[z] \to V[z]$ is injective. Passing to this representative, let $\psi_+ \in V$ denote the corresponding constant. Expanding in a similar fashion $I\xi(z) = \xi(B_1) I + (z - B_1)(\dots)$, we see that \eqref{kerbetaeq} implies that $\xi(B_1) I = B_2 \psi_+$ and that $\psi_-(z)$ is uniquely determined by $\xi(z)$: 
    \begin{equation}
        \psi_-(z) = \frac{1}{z - B_1}(I \xi(z) - \xi(B_1) I) \in V[z].
    \end{equation}
\end{proof}

\subsubsection{$\tau$-function via Slater determinants}
This section offers a complementary proof of Proposition \ref{taufnformula} compared to the one in the main text, using the combinatorics of the semi-infinite wedge representation. Consequently we borrow notations from the main text. 

It is a general result (see e.g. \cite{dijkgraaf92} for a review and explanation) that if a subspace $\mathbb{W}$ belongs to the big cell of the index zero component of $\text{Gr}(\mathscr{H})$, then if we make Miwa's substitution 
\begin{equation}
    t_n = - \frac{1}{n} \sum_{i = 1}^N x_i^{-n}
\end{equation}
that 
\begin{equation} \label{tauslater}
    \frac{\tau_{\mathbb{W}}(t)}{\tau_{\mathbb{W}}(0)} = \frac{\det_{1 \leq i,j \leq N}(w_{j - 1}(x_i))}{\det_{1 \leq i, j \leq N} (x_i^{j - 1})}
\end{equation}
where $w_j(z)$, $j = 0, 1, 2, \dots$ is any basis of $\mathbb{W}$ with $w_j(z) = z^j(1 + O(z^{-1}))$, which exists since we assume $\mathbb{W}$ is in the big cell. We may apply this to the vector bundle $\mathbb{W}$ of the main text, if we restrict ourselves to working over the open subset $\iota: U' \hookrightarrow \widetilde{M}(1)$ where $\det B_2 \neq 0$; in particular we have the following basis of sections of $\iota^* \mathbb{W}$:
\begin{equation}
    w_j(z) = z^j + J\frac{1}{z - B_1} B_2^{-1} B_1^j I. 
\end{equation}
Note that the right hand side of \eqref{tauslater} vanishes if there exists a nonzero linear combination $f(z) = \sum_{j = 0}^{N - 1} c_j w_j(z)$ such that $f(x_i) = 0$ for $i = 1, \dots, N$. We may assume that the $x_i$ are distinct. Such a nonzero linear combination is the same as a section of $\iota^* \mathbb{W}$ such that $f(z) = O(z^{N - 1})$ as $z \to \infty$. This $f(z)$ is of the form 
\begin{equation}
    f(z) = \xi(z) + J \frac{1}{z - B_1} \psi_+
\end{equation}
with $\xi(z) \in \mathscr{O}_{U'} \otimes \mathbb{C}[z]$, $\deg \xi(z) \leq N - 1$ and $\psi_+$ a section of $\iota^* \mathscr{V}$ satisfying $B_2 \psi_+ = \xi(B_1) I$. The degree of $\xi(z)$ and the conditions $f(x_i) = 0$ uniquely determine $\xi(z)$ to be 
\begin{equation}
    \xi(z) = - \sum_{i = 1}^N J \frac{1}{x_i - B_1} \psi_+ \prod_{j (\neq i)} \frac{z - x_j}{x_i - x_j}
\end{equation}
by polynomial interpolation. $\xi(B_1) I = B_2 \psi_+$ becomes the condition 
\begin{equation}
    \Big( B_2 + \sum_{i = 1}^N \prod_{j (\neq i)} \frac{B_1 - x_j}{x_i - x_j} IJ \frac{1}{x_i - B_1} \Big) \psi_+ = 0  
\end{equation}
whence a nonzero $f(z)$ exists iff the matrix above has a nontrivial kernel, therefore
\begin{equation}
    \text{RHS of \eqref{tauslater}} \propto \det\Big(B_2 + \sum_{i = 1}^N \prod_{j (\neq i)} \frac{B_1 - x_j}{x_i - x_j} IJ \frac{1}{x_i -B_1} \Big). 
\end{equation}
The operator above may be simplified by writing $IJ = \hbar - \comm{B_1}{B_2} = \hbar - \comm{B_1 - x_i}{B_2}$, then we have that the expression in parenthesis coincides with 
\begin{equation}
    B_2 - \sum_{i = 1}^N \prod_{j (\neq i)} \frac{B_1 - x_j}{x_i - x_j} (B_1 - x_i) B_2 \frac{1}{x_i - B_1} - \sum_{i = 1}^N \prod_{j (\neq i)} \frac{B_1 - x_j}{x_i - x_j} B_2 + \hbar \sum_{i = 1}^N \prod_{j (\neq i)} \frac{B_1 - x_j}{x_i - x_j} \frac{1}{x_i - B_1}. 
\end{equation}
The third term is easy to simplify, because 
\begin{equation}
    -\sum_{i = 1}^N \prod_{j (\neq i)} \frac{B_1 - x_j}{x_i - x_j} = - \sum_{i = 1}^N \prod_{j (\neq i)} \frac{z - x_j}{x_i - x_j} \eval_{z = B_1} = - 1
\end{equation}
by observing that this is a polynomial of degree at most $N - 1$ taking the value $-1$ at $N$ points. We note that the second term is 
\begin{equation}
    - \prod_{k = 1}^N (B_1 - x_k) B_2 \sum_{i = 1}^N \prod_{j (\neq i)} \frac{1}{x_i - x_j} \frac{1}{x_i - B_1}
\end{equation}
and 
\begin{equation}
    \sum_{i = 1}^N \prod_{j (\neq i)} \frac{1}{x_i - x_j} \frac{1}{x_i - B_1} = \oint_\gamma \frac{du}{2\pi i} \frac{1}{\prod_j (u  - x_j)} \frac{1}{u - z} \eval_{z = B_1} = - \frac{1}{\prod_j(B_1 - x_j)}
\end{equation}
where $\gamma$ is a loop encircling the $x_i$ which in the last line we deform to a loop surrounding $z$, then substitute $z = B_1$. Then 
\begin{equation}
\begin{split}
    \det\Big( B_2 + \sum_{i = 1}^N \prod_{j (\neq i)}\frac{B_1 - x_j}{x_i - x_j} IJ \frac{1}{x_i - B_1} \Big) & = \det\Big( \prod_j (B_1 - x_j) B_2 \frac{1}{\prod_j (B_1- x_j)} + \hbar \sum_{i = 1}^N \prod_{j (\neq i)} \frac{B_1 - x_j}{x_i - x_j} \frac{1}{x_i - B_1} \Big) \\
    & = \det\Big( B_2 + \hbar \sum_{i = 1}^N \prod_{j (\neq i)} \frac{B_1 - x_j}{x_i - x_j} \frac{1}{x_i - B_1} \Big)
\end{split}
\end{equation}
since the term proportional to $\hbar$ commutes with $B_1$. Finally observe that 
\begin{equation}
\begin{split}
    \sum_{i = 1}^N \prod_{j (\neq i)} \frac{B_1 - x_j}{x_i - x_j} \frac{1}{x_i - B_1} & = -\oint_\gamma \frac{du}{2\pi i } \prod_j \frac{z - x_j}{u - x_j} \frac{1}{(u - z)^2} \eval_{z = B_1} \\
    & = \frac{d}{du} \Big( \prod_{j} \frac{z - x_j}{u - x_j} \Big) \eval_{u = z} \eval_{z = B_1} \\
    & = \sum_{i = 1}^N \frac{1}{x_i - B_1}.
\end{split}
\end{equation}
In all, we have 
\begin{equation}
    \frac{\tau_{\mathbb{W}}(t)}{\tau_{\mathbb{W}}(0)} = \frac{\det_{1 \leq i,j \leq N}(w_{j - 1}(x_i))}{\det_{1 \leq i, j \leq N} (x_i^{j - 1})} = C(x_1, \dots, x_N) \det\Big( 1 + \hbar B_2^{-1} \sum_{i = 1}^N \frac{1}{x_i - B_1} \Big) 
\end{equation}
where $C(x_1, \dots, x_N)$ is a rational function of each $x_i$ (since each $w_{j - 1}(x_i)$ is rational), with no pole in $x_i$. Moreover, $C(x_1, \dots, x_N)$ approaches $C(x_1, \dots, \widehat{x}_i, \dots, x_N)$ at $x_i \to \infty$ (hat denotes omission of the $i$-th variable) due to the normalization of $w_j(z)$. We conclude by induction on $N$ that $C(x_1, \dots, x_N) = 1$ identically, and invoking the normalization $\tau_{\mathbb{W}}(0) = \det(B_2)$ from Proposition \ref{vaccomponent}, expanding each $(x_i - B_1)^{-1}$ in a geometric series and substituting $t_n = - \frac{1}{n} \sum_i x_i^{-n}$ yields Proposition \ref{taufnformula}. 

\subsubsection{Explicit description of Ext when $P(z, w) = zw - u$} \label{subsect:explicitextforzw}
In this subsection we will prove an analog of Proposition \ref{explicitext} for the second case of interest of the main body, $P(z, w) = zw - u$. It is first worth noting that 
\begin{equation}
    \widehat{\mathscr{O}}^\ell_{C_u} = \mathscr{D}_\hbar/\mathscr{D}_\hbar(zw - u)
\end{equation}
has a vector space basis given by $\{ 1, z^n, w^n \}_{n \geq 1}$. Any element $f(z, w) \in \widehat{\mathscr{O}}^\ell_{C_u}$ can be expanded uniquely as 
\begin{equation}
    f(z, w) = f^+(z) + wf^-(w)
\end{equation}
with $f^+(z) \in \mathbb{C}[z] \subset \widehat{\mathscr{O}}^\ell_{C_u}$, $f^-(w) \in \mathbb{C}[w] \subset \widehat{\mathscr{O}}^\ell_{C_u}$. In terms of the scattering problems we are emphasizing in this paper, the $f^+$ and $f^-$ should be thought of as incoming/outgoing waves. 

We have the following elementary 
\begin{lemma} \label{lemma:vanishinzw}
Let $P(w) \in V[w] \subset V \otimes \widehat{\mathscr{O}}^\ell_{C_u}$, $Q(z) \in V[z] \subset V \otimes \widehat{\mathscr{O}}^\ell_{C_u}$, $R \in V \subset V \otimes \widehat{\mathscr{O}}^\ell_{C_u}$. Then 
\begin{equation}
    (B_2 - w)P(w) + (B_1 - z)Q(z) + R = 0 
\end{equation}
implies $P(w) = Q(z) = R = 0$.
\end{lemma}

\begin{proof}
    By induction on the degree of $P$ and $Q$. 
\end{proof}

\begin{prop} \label{prop:explicitextforzw}
    We have an inclusion $\textnormal{Ext}^1_{\mathscr{D}_\hbar}(\widehat{\mathscr{O}}_{C_u}, \mathscr{E}) \hookrightarrow W \otimes \widehat{\mathscr{O}}^\ell_{C_u} \oplus V$ as the locus of $(\xi(z, w), \psi_0) \in W \otimes \widehat{\mathscr{O}}^\ell_{C_u} \oplus V$ such that $(u - B_2 B_1) \psi_0 = \xi(B_1, B_2) I$. 
\end{prop}

We define $\xi(B_1, B_2)I$ by writing $\xi(z, w) = \xi^+(z) + w\xi^-(w)$ uniquely as above and defining $\xi(B_1, B_2)I := \xi^+(B_1) I + B_2 \xi^-(B_2)I$, understanding single variable polynomials valued in $W$ acting on $I$ as in Proposition \ref{explicitext}. Note that this proposition also works in families over the instanton moduli space, i.e. we have an exact sequence of vector bundles 
\begin{equation}
\begin{tikzcd}
    0 \arrow[r] & \text{Ext}^1_{\mathscr{D}_\hbar}(\widehat{\mathscr{O}}^\ell_{C_u}, \mathscr{E}) \arrow[r] & \mathscr{W} \otimes \widehat{\mathscr{O}}^\ell_{C_u} \oplus \mathscr{V} \arrow[r] & \mathscr{V} \arrow[r] & 0.
\end{tikzcd}
\end{equation}

\begin{proof}
    As usual, by Proposition \ref{extgroups}, a class in $\text{Ext}^1_{\mathscr{D}_\hbar}(\widehat{\mathscr{O}}^\ell_{C_u}, \mathscr{E})$ is the same as a triple $$(\psi_+(z, w), \psi_-(z, w), \xi(z, w)) \in (V \oplus V \oplus W) \otimes \widehat{\mathscr{O}}^\ell_{C_u}$$ satisfying the equation 
    \begin{equation}
        -(B_2 - w) \psi_+(z, w) + (B_1 - z) \psi_-(z, w) + I \xi(z, w) = 0
    \end{equation}
    and considered modulo 
    \begin{equation}
    \begin{split}
        \psi_+(z, w) & \sim \psi_+(z, w) + (B_1 - z) a(z, w) \\
        \psi_-(z, w) & \sim \psi_-(z, w) + (B_2 - w) a(z, w) \\
        \xi(z, w) & \sim \xi(z, w) + J a(z, w) 
    \end{split}
    \end{equation}
    with $a(z, w) \in V \otimes \widehat{\mathscr{O}}^\ell_{C_u}$. We will imitate the proof of Proposition \ref{explicitext}. The first step is to uniquely decompose 
    \begin{equation}
        a(z, w) = a^+(z) + w a^-(w)
    \end{equation}
    with $a^+(z) \in V[z] \subset V \otimes \widehat{\mathscr{O}}^\ell_{C_u}$, $a^-(w) \in V[w] \subset V \otimes \widehat{\mathscr{O}}^\ell_{C_u}$ as above. Arguing as in the proof of Proposition \ref{explicitext}, for arbitrary $\psi_+(z, w)$ there is a unique $a^+(z)$ such that $\psi_+(z, w) + (B_1 - z) a^+(z) \in V[w] \subset V \otimes \widehat{\mathscr{O}}^\ell_{C_u}$. Moreover, shifting by $w a^-(w)$ preserves this class of choices, because 
    \begin{equation}
        (B_1 - z) w a^-(w) = B_1 wa^-(w) + w\hbar \partial_w a^-(w) - u a^-(w) \in V[w]. 
    \end{equation}
    Then, again arguing as in the proof of Proposition \ref{explicitext}, we see there is a unique $a^-(w)$ such that 
    \begin{equation}
        \psi_-(z, w) + (B_2 - w) w a^-(w) = \psi_0 w + \psi_-(z)
    \end{equation}
    with $\psi_0 \in V \subset V \otimes \widehat{\mathscr{O}}^\ell_{C_u}$ and $\psi_-(z) \in V[z] \subset V \otimes \widehat{\mathscr{O}}^\ell_{C_u}$. Likewise the shift by $a^+(z)$ preserves this class of choices. 
    
    We see that if we assume $\psi_+(z, w) = \psi_+(w) \in V[w]$, $\psi_-(z, w) = \psi_0 w + \psi_-(z)$ with $\psi_0 \in V$, $\psi_-(z) \in V[z]$ we may dispense with the quotient, and we need only solve the equation 
    \begin{equation} \label{eq:zwext}
        -(B_2 - w)\psi_+(w) + (B_1 - z)(\psi_0 w + \psi_-(z)) + I \xi(z, w) = 0. 
    \end{equation}
    Writing 
    \begin{equation}
    \begin{split}
        (B_1 - z) \psi_0 w & = w B_1 \psi_0 - \psi_0 u = - (B_2 - w) B_1 \psi_0 + (B_2 B_1 - u) \psi_0 \\
        I \xi(z, w) & = I \xi^+(z) + w I \xi^-(w) = \xi^+(B_1) I + B_2 \xi^-(B_2) I + (B_1 - z)(\dots) + (B_2 - w)(\dots)
    \end{split}
    \end{equation}
    where the expressions in parentheses are polynomials in $z$ and $w$, Lemma \ref{lemma:vanishinzw} applied to \eqref{eq:zwext} implies 
    \begin{equation} \label{eq:zwextsoln}
    \begin{split}
        (B_2 B_1 - u) \psi_0 + \xi^+(B_1) I + B_2 \xi^-(B_2)I & = 0 \\
        \psi_+(w) + B_1 \psi_0 + \frac{1}{w - B_2} (w I \xi^-(w) - B_2 \xi^-(B_2)I) & = 0 \\
        \psi_-(z) - \frac{1}{z - B_1}(I \xi^+(z) - \xi^+(B_1) I) & = 0.
    \end{split}
    \end{equation}
    The last two equations just say that $\psi_+(w)$ and $\psi_-(z)$ are uniquely determined by a pair $(\xi, \psi_0)$ solving the first equation, which is what we wanted to prove. 
    \end{proof}

    \subsubsection{Functional realizations of $\widehat{\mathscr{O}}^\ell_{C_u}$} \label{subsect:functlzw}
    As a further preliminary for the proof of Proposition \ref{prop:VmaptoGr}, it is useful to describe two functional realizations of $\widehat{\mathscr{O}}^\ell_{C_u}$. Define a left $\mathscr{D}_\hbar$-module 
    \begin{equation}
        \mathscr{O}^{\text{in}} := \mathbb{C}[z, z^{-1}]
    \end{equation}
    in which $z$ acts by multiplication and $w$ acts by $\hbar \partial_z + u z^{-1}$. Likewise define a left $\mathscr{D}_{\hbar}$ module 
    \begin{equation}
        \mathscr{O}^{\text{out}} := \mathbb{C}[w, w^{-1}]
    \end{equation}
    in which $w$ acts by multiplication and $z$ acts by $- \hbar \partial_w + (u + \hbar) w^{-1}$. 

    It is readily verified that the morphisms 
    \begin{equation}
    \begin{split}
        & \phi_z  : \widehat{\mathscr{O}}^\ell_{C_u} \to \mathscr{O}^{\text{in}} \\
        & \phi_w : \widehat{\mathscr{O}}^\ell_{C_u} \to \mathscr{O}^{\text{out}}
    \end{split}
    \end{equation}
    defined by ($\Gamma$ denotes the gamma function)
    \begin{equation}
    \begin{split}
        \phi_z(z^n) & = z^n  \\
        \phi_z(w^n) & = z^{-u/\hbar} (\hbar \partial_z)^n z^{u/\hbar} = \frac{u(u - \hbar) \dots (u - (n - 1)\hbar)}{z^n} = \hbar^n\frac{\Gamma(\frac{u}{\hbar} + 1)}{\Gamma(\frac{u}{\hbar} - n + 1)} z^{-n} \\
        \phi_w(z^n) & = w^{u/\hbar + 1} (-\hbar \partial_w)^n w^{-u/\hbar - 1} = \frac{(u + \hbar) \dots (u + n \hbar)}{w^n} = \hbar^n \frac{\Gamma(\frac{u}{\hbar} + n + 1)}{\Gamma(\frac{u}{\hbar} + 1)} w^{-n} \\
        \phi_w(w^n) & = w^n
    \end{split}
    \end{equation}
    define isomorphisms of $\mathscr{D}_\hbar$-modules as long as $u \notin \hbar \mathbb{Z}$. Note $\phi_w \circ \phi_z^{-1}$ is essentially the Fourier transform.

    \subsubsection{Proof of Proposition \ref{prop:VmaptoGr}} \label{proof:VmaptoGr}
    Now we will give the proof of Proposition \ref{prop:VmaptoGr}, assuming familiarity with the contents of \ref{subsect:explicitextforzw}, \ref{subsect:functlzw}. 

    Using the description of $\mathbb{V}_u$ as in Proposition \ref{extgroups} together with the functional realizations of section \ref{subsect:functlzw}, there is essentially only one morphism 
    \begin{equation}
        \mathbb{V}_u \to \mathscr{O}_{\widetilde{M}(r)} \otimes (\mathscr{H}_{\text{in}} \oplus \mathscr{H}_{\text{out}}) 
    \end{equation}
    that can be written down: it sends $(\psi_+, \psi_-, \xi) \mapsto (f_{\text{in}}(z), f_{\text{out}}(w))$ where 
    \begin{equation}
    \begin{split}
        f_{\text{in}}(z) & = J \frac{1}{z - B_1} \phi_z(\psi_+) + \phi_z(\xi) \in \mathscr{H}_{\text{in}}\\
        f_{\text{out}}(w) & = J \frac{1}{w - B_2} \phi_w( \psi_-) + \phi_w(\xi) \in \mathscr{H}_{\text{out}}
    \end{split}
    \end{equation}
    where we understand $f_{\text{in}}, f_{\text{out}}$ as elements of $\mathscr{H}_{\text{in}}, \mathscr{H}_{\text{out}}$ via expansion at $z \to \infty$, $w \to \infty$. 

    The first assertion that needs to be checked is that, if $u \notin \hbar \mathbb{Z}$, this map is an inclusion of vector bundles. For this, it is useful to pass to the description of $\mathbb{V}_u$ as in Proposition \ref{prop:explicitextforzw}, so that we may assume $\psi_+ = \psi_+(w)$, $\psi_- =  \psi_-(z) + \psi_0 w$, solving equations \eqref{eq:zwextsoln} and not subject to any quotient. 

    To check that we have an inclusion, suppose $(f_{\text{in}}, f_{\text{out}}) = 0$. Then
    \begin{equation} \label{eq:fsvanishzw}
    \begin{split}
        \phi_z( \xi) & = - J \frac{1}{z - B_1} \phi_z(\psi_+(w)) \\
        \phi_w(\xi) & = - J \frac{1}{w - B_2} \phi_w(\psi_0 w + \psi_-(z)). 
    \end{split}
    \end{equation}
    Note that since $\phi_z(\psi_+(w)) = O(1)$ as $z \to \infty$, the RHS of the first line is $O(z^{-1})$ for $z \to \infty$. By uniqueness of the decomposition $\xi = \xi^+(z) + w \xi^-(w)$ and the fact that $\phi_z$ is an isomorphism, this implies that $\xi^+(z) = 0$. Likewise, the RHS of the second line is $O(1)$ for $w \to \infty$, implying that $\xi^-(w) = 0$. Then \eqref{eq:zwextsoln} gives that $\psi_-(z) = 0$, $\psi_+(w) = - B_1 \psi_0$, $(u - B_2 B_1) \psi_0 = 0$, and we have 
    \begin{equation}
    \begin{split}
        J \frac{1}{z - B_1} B_1 \psi_0 & = 0 \\
        J \frac{w}{w - B_2} \psi_0 & = 0
    \end{split}
    \end{equation}
    from \eqref{eq:fsvanishzw} and the vanishing of $\xi$. Expanding in $z^{-1}$ and $w^{-1}$, it follows that $J \psi_0 = JB_1^n \psi_0 = J B_2^n \psi_0 = 0$, so that $S = \text{Span}(\psi_0, B_1^n \psi_0, B_2^n \psi_0)_{n \geq 1} \subseteq \ker J$. $S$ is both $B_1$ invariant and $B_2$ invariant, because for $n \geq 1$:
    \begin{equation}
    \begin{split}
        B_2 B_1^n \psi_0 & = \comm{B_2}{B_1^{n - 1}} B_1 \psi_0 + B_1^{n - 1} B_2 B_1 \psi_0 \\
        & = \sum_{k = 0}^{n - 2} B_1^k \comm{B_2}{B_1} B_1^{n - 2 - k} B_1 \psi_0  + u B_1^{n - 1}\psi_0 \\
        & = - \hbar(n - 1) B_1^{n - 1} \psi_0 + \sum_{k = 0}^{n - 2} B_1^k IJ B_1^{n - 1 - k} \psi_0  + u B_1^{n - 1} \psi_0 \\
        & = (u - (n - 1) \hbar) B_1^{n - 1} \psi_0
    \end{split}
    \end{equation}
    and likewise $B_1 B_2^n \psi_0 = (u + \hbar n) B_2^{n - 1} \psi_0$ by a similar computation. Lemma \ref{stabandcostab} then implies that $S = 0$, so $\psi_0 = 0$ establishing the injectivity. 

    Next we investigate the kernel of the induced projection $\mathbb{V}_u \to \mathscr{O}_{\widetilde{M}(r)} \otimes \mathscr{H}_+$, with $\mathscr{H}_+ = W[z] \oplus w W[w]$. In the description of $\mathbb{V}_u$ from Proposition \ref{prop:explicitextforzw}, this map is identified with $(\xi, \psi_0) \mapsto (\phi_z(\xi^+(z)), \phi_w(w \xi^-(w))$, so the kernel precisely consists of the pairs $(\xi, \psi_0)$ with $\xi = 0$; by Proposition \ref{prop:explicitextforzw}, these are the same as $\psi_0$ solving $(B_2 B_1 - u) \psi_0 = 0$ establishing 
    \begin{equation}
        \ker( \mathbb{V}_u \to \mathscr{O}_{\widetilde{M}(r)} \otimes \mathscr{H}_+) \simeq \ker(u - B_2 B_1). 
    \end{equation}

    Finally we need to characterize $\text{coker}(\mathbb{V}_u \to \mathscr{O}_{\widetilde{M}(r)} \otimes \mathscr{H}_+)$. Note that the assignment 
    \begin{equation}
       \mathscr{H}_+ \ni  (p^+(z), wp^-(w)) \mapsto p^+(B_1) I + B_2p^-(B_2)I \in \text{coker}(u - B_2 B_1)
    \end{equation}
    induces (by Proposition \ref{prop:explicitextforzw}) an injective homomorphism of sheaves on $\widetilde{M}(r)$ $$\text{coker}(\mathbb{V}_u \to \mathscr{O}_{\widetilde{M}(r)}  \otimes \mathscr{H}_+) 
    \longrightarrow \text{coker}(u - B_2 B_1);$$ surjectivity is established by an argument with the stability condition and Lemma \ref{stabilitybasis} essentially identical to the one in the proof of Proposition \ref{maptoGr}. 
\end{appendices}

\newpage
\printbibliography

@article{witten,
    author = {E. Witten},
    title = {Quantum field theory, Grassmannians, and algebraic curves},
    journal = {Comm. Math. Phys.},
    year = {1988},
    volume={113},
    pages={529-600},
}

@article{Polychronakos_2019,
   title={Feynman’s proof of the commutativity of the Calogero integrals of motion},
   volume={403},
   ISSN={0003-4916},
   url={http://dx.doi.org/10.1016/j.aop.2019.02.005},
   DOI={10.1016/j.aop.2019.02.005},
   journal={Annals of Physics},
   publisher={Elsevier BV},
   author={Polychronakos, A. P.},
   year={2019},
   month=apr, pages={145–151} }

@article{wilson,
    author = {G. Wilson},
    title = {Collisions of Calogero-Moser particles  and an adelic Grassmannian(With an Appendix by I.G. Macdonald)},
    journal = {Inventiones mathematicae},
    year = {1998}
}

@misc{mo,
      title={Quantum Groups and Quantum Cohomology}, 
      author={D. Maulik and A. Okounkov},
      year={2018},
      eprint={1211.1287},
      archivePrefix={arXiv},
      primaryClass={math.AG},
      url={https://arxiv.org/abs/1211.1287}, 
}

@book{atiyahhitchin,
 URL = {http://www.jstor.org/stable/j.ctt7zv206},
 author = {M. F. Atiyah and N. Hitchin},
 publisher = {Princeton University Press},
 title = {The Geometry and Dynamics of Magnetic Monopoles},
 year = {1988}
}

@article{frassekpestuntsymb,
   title={Lax matrices from antidominantly shifted Yangians and quantum affine algebras: A-type},
   volume={401},
   ISSN={0001-8708},
   url={http://dx.doi.org/10.1016/j.aim.2022.108283},
   DOI={10.1016/j.aim.2022.108283},
   journal={Advances in Mathematics},
   publisher={Elsevier BV},
   author={Frassek, R. and Pestun, V. and Tsymbaliuk, A.},
   year={2022},
   month=jun, pages={108283} }

@article{Gerasimov_2005,
   title={On a Class of Representations of the Yangian and Moduli Space of Monopoles},
   volume={260},
   ISSN={1432-0916},
   url={http://dx.doi.org/10.1007/s00220-005-1417-3},
   DOI={10.1007/s00220-005-1417-3},
   number={3},
   journal={Communications in Mathematical Physics},
   publisher={Springer Science and Business Media LLC},
   author={Gerasimov, A. and Kharchev, S and Lebedev, D. and Oblezin, S.},
   year={2005},
   month=aug, pages={511–525} }

@article{nekrasovpestun,
   title={Seiberg-Witten Geometry of Four-Dimensional $\mathcal N=2$ Quiver Gauge Theories},
   ISSN={1815-0659},
   url={http://dx.doi.org/10.3842/SIGMA.2023.047},
   DOI={10.3842/sigma.2023.047},
   journal={Symmetry, Integrability and Geometry: Methods and Applications},
   publisher={SIGMA (Symmetry, Integrability and Geometry: Methods and Application)},
   author={Nekrasov, N. and Pestun, V.},
   year={2023},
   month=jul }

@article{nekrasovlee,
   title={Quantum spin systems and supersymmetric gauge theories. Part I},
   volume={2021},
   ISSN={1029-8479},
   url={http://dx.doi.org/10.1007/JHEP03(2021)093},
   DOI={10.1007/jhep03(2021)093},
   number={3},
   journal={Journal of High Energy Physics},
   publisher={Springer Science and Business Media LLC},
   author={Lee, N. and Nekrasov, N.},
   year={2021},
   month=mar }

@misc{BDG,
      title={The Coulomb Branch of 3d $\mathcal{N}=4$ Theories}, 
      author={M. Bullimore and T. Dimofte and D. Gaiotto},
      year={2015},
      eprint={1503.04817},
      archivePrefix={arXiv},
      primaryClass={hep-th},
      url={https://arxiv.org/abs/1503.04817}, 
}

@misc{costelloQop,
      title={Q-operators are 't Hooft lines}, 
      author={K. Costello and D. Gaiotto and J. Yagi},
      year={2021},
      eprint={2103.01835},
      archivePrefix={arXiv},
      primaryClass={hep-th},
      url={https://arxiv.org/abs/2103.01835}, 
}

@article{Moore_2014,
   title={Parameter counting for singular monopoles on $\mathbb{R}^3$},
   volume={2014},
   ISSN={1029-8479},
   url={http://dx.doi.org/10.1007/JHEP10(2014)142},
   DOI={10.1007/jhep10(2014)142},
   number={10},
   journal={Journal of High Energy Physics},
   publisher={Springer Science and Business Media LLC},
   author={Moore, G. W. and Royston, A. B. and Van den Bleeken, D.},
   year={2014},
   month=oct }

@misc{KW,
      title={Electric-Magnetic Duality And The Geometric Langlands Program}, 
      author={A. Kapustin and E. Witten},
      year={2007},
      eprint={hep-th/0604151},
      archivePrefix={arXiv},
      primaryClass={hep-th},
      url={https://arxiv.org/abs/hep-th/0604151}, 
}

@misc{bfnslice,
      title={Coulomb branches of $3d$ $\mathcal N=4$ quiver gauge theories and slices in the affine Grassmannian (with appendices by Alexander Braverman, Michael Finkelberg, Joel Kamnitzer, Ryosuke Kodera, Hiraku Nakajima, Ben Webster, and Alex Weekes)}, 
      author={A. Braverman and M. Finkelberg and H. Nakajima},
      year={2018},
      eprint={1604.03625},
      archivePrefix={arXiv},
      primaryClass={math.RT},
      url={https://arxiv.org/abs/1604.03625}, 
}

@article{FADDEEV_1995,
   title={Algebraic Aspects of Bethe Ansatz},
   volume={10},
   ISSN={1793-656X},
   url={http://dx.doi.org/10.1142/S0217751X95000905},
   DOI={10.1142/s0217751x95000905},
   number={13},
   journal={International Journal of Modern Physics A},
   publisher={World Scientific Pub Co Pte Lt},
   author={Faddeev, L. D.},
   year={1995},
   month=may, pages={1845–1878} }

@misc{tamagnishiftop,
      title={Nonabelian shift operators and shifted Yangians}, 
      author={S. Tamagni},
      year={2024},
      eprint={2412.17906},
      archivePrefix={arXiv},
      primaryClass={math.AG},
      url={https://arxiv.org/abs/2412.17906}, 
}

@misc{krylov,
      title={Almost dominant generalized slices and convolution diagrams over them}, 
      author={V. Krylov and I. Perunov},
      year={2021},
      eprint={1903.08277},
      archivePrefix={arXiv},
      primaryClass={math.RT},
      url={https://arxiv.org/abs/1903.08277}, 
}

@article{Nekrasov_1998,
   title={Instantons on Noncommutative $\mathbb{R}^4$, and (2,0) Superconformal Six Dimensional Theory},
   volume={198},
   ISSN={1432-0916},
   url={http://dx.doi.org/10.1007/s002200050490},
   DOI={10.1007/s002200050490},
   number={3},
   journal={Communications in Mathematical Physics},
   publisher={Springer Science and Business Media LLC},
   author={Nekrasov, N. and Schwarz, A.},
   year={1998},
   month=nov, pages={689–703} }

@article{adhm,
    author = {M. Atiyah and V. Drinfeld and N. Hitchin and Y. Manin},
    title = {Construction of instantons},
    journal = {Physics Letters A},
    year = {1978}, 
    volume = {65}, 
    issue = {3}, 
    pages = {185-187},
}

@article{segal-wilson,
    author = {G. Segal and G. Wilson},
    title = {Loop groups and equations of KdV type},
    journal = {Publications Mathématiques de L’Institut des Hautes Scientifiques},
    year = {1985},
}

@article{Aganagic_2005,
   title={Topological Strings and Integrable Hierarchies},
   volume={261},
   ISSN={1432-0916},
   url={http://dx.doi.org/10.1007/s00220-005-1448-9},
   DOI={10.1007/s00220-005-1448-9},
   number={2},
   journal={Communications in Mathematical Physics},
   publisher={Springer Science and Business Media LLC},
   author={Aganagic, M. and Dijkgraaf, R. and Klemm, A. and Mariño, M. and Vafa, C.},
   year={2005},
   month=oct, pages={451–516} }

@article{DHSV,
   title={Supersymmetric gauge theories, intersecting branes and free fermions},
   volume={2008},
   ISSN={1029-8479},
   url={http://dx.doi.org/10.1088/1126-6708/2008/02/106},
   DOI={10.1088/1126-6708/2008/02/106},
   number={02},
   journal={Journal of High Energy Physics},
   publisher={Springer Science and Business Media LLC},
   author={Dijkgraaf, R. and Hollands, L. and Sułkowski, P. and Vafa, C.},
   year={2008},
   month=feb, pages={106–106} }

@article{DHS,
   title={Quantum curves and $\mathscr{D}$-modules},
   volume={2009},
   ISSN={1029-8479},
   url={http://dx.doi.org/10.1088/1126-6708/2009/11/047},
   DOI={10.1088/1126-6708/2009/11/047},
   number={11},
   journal={Journal of High Energy Physics},
   publisher={Springer Science and Business Media LLC},
   author={Dijkgraaf, R. and Hollands, L. and Sułkowski, P.},
   year={2009},
   month=nov, pages={047–047} }

@misc{gaiottorapcak2020,
      title={Miura operators, degenerate fields and the M2-M5 intersection}, 
      author={D. Gaiotto and M. Rapcak},
      year={2020},
      eprint={2012.04118},
      archivePrefix={arXiv},
      primaryClass={hep-th},
      url={https://arxiv.org/abs/2012.04118}, 
}

@misc{gaiottorapcakzhou24,
      title={Deformed Double Current Algebras, Matrix Extended $\mathcal W_{\infty}$ Algebras, Coproducts, and Intertwiners from the M2-M5 Intersection}, 
      author={D. Gaiotto and M. Rapčák and Y. Zhou},
      year={2024},
      eprint={2309.16929},
      archivePrefix={arXiv},
      primaryClass={hep-th},
      url={https://arxiv.org/abs/2309.16929}, 
}

@misc{zenkevich24,
      title={Spiralling branes, affine qq-characters and elliptic integrable systems}, 
      author={Y. Zenkevich},
      year={2024},
      eprint={2412.20926},
      archivePrefix={arXiv},
      primaryClass={hep-th},
      url={https://arxiv.org/abs/2412.20926}, 
}

@article{Prochazka_2019,
   title={Instanton R-matrix and $\mathcal{W}$-symmetry},
   volume={2019},
   ISSN={1029-8479},
   url={http://dx.doi.org/10.1007/JHEP12(2019)099},
   DOI={10.1007/jhep12(2019)099},
   number={12},
   journal={Journal of High Energy Physics},
   publisher={Springer Science and Business Media LLC},
   author={Procházka, T.},
   year={2019},
   month=dec }

@article{miuraCSfeynman1,
   title={R-matrices and Miura operators in 5d Chern-Simons theory},
   volume={8},
   ISSN={2666-9366},
   url={http://dx.doi.org/10.21468/SciPostPhysCore.8.1.003},
   DOI={10.21468/scipostphyscore.8.1.003},
   number={1},
   journal={SciPost Physics Core},
   publisher={Stichting SciPost},
   author={Ishtiaque, N. and Jeong, S. and Zhou, Y.},
   year={2025},
   month=jan }

@misc{miuraCSfeynman2,
      title={R-matrices from Feynman Diagrams in 5d Chern-Simons Theory and Twisted M-theory}, 
      author={M. Ashwinkumar},
      year={2024},
      eprint={2408.15732},
      archivePrefix={arXiv},
      primaryClass={hep-th},
      url={https://arxiv.org/abs/2408.15732}, 
}

@misc{haouzijeong,
      title={Miura operators as R-matrices from M-brane intersections}, 
      author={N. Haouzi and S. Jeong},
      year={2024},
      eprint={2407.15990},
      archivePrefix={arXiv},
      primaryClass={hep-th},
      url={https://arxiv.org/abs/2407.15990}, 
}

@misc{koderanakajima,
      title={Quantized Coulomb branches of Jordan quiver gauge theories and cyclotomic rational Cherednik algebras}, 
      author={R. Kodera and H. Nakajima},
      year={2016},
      eprint={1608.00875},
      archivePrefix={arXiv},
      primaryClass={math.RT},
      url={https://arxiv.org/abs/1608.00875}, 
}

@misc{arbesfeldschiffmann,
      title={A presentation of the deformed $W_{1+\infty}$ algebra}, 
      author={N. Arbesfeld and O. Schiffmann},
      year={2012},
      eprint={1209.0429},
      archivePrefix={arXiv},
      primaryClass={math.RT},
      url={https://arxiv.org/abs/1209.0429}, 
}

@article{witten84,
    author = {E. Witten},
    title = {Nonabelian bosonization in two dimensions},
    journal = {Commun. Math. Phys},
    year = {1984}, 
    volume = {92}, 
    pages = {455-472},
}

@article{atiyahinstanton,
    author = {M. F. Atiyah},
    title = {Instantons in two and four dimensions},
    journal = {Commun. Math. Phys.},
    year = {1984}, 
    volume = {93}, 
    pages = {437-451}, 
}

@misc{bottatamagni, 
author = {T. M. Botta and S. Tamagni}, 
title = {In preparation}
}

@misc{nairtamagni, 
author = {S. Nair and S. Tamagni}, 
title = {In preparation}
}

@article{seibergwitten,
   title={String theory and noncommutative geometry},
   volume={1999},
   ISSN={1029-8479},
   url={http://dx.doi.org/10.1088/1126-6708/1999/09/032},
   DOI={10.1088/1126-6708/1999/09/032},
   number={09},
   journal={Journal of High Energy Physics},
   publisher={Springer Science and Business Media LLC},
   author={Seiberg, N. and Witten, E.},
   year={1999},
   month=sep, 
    pages={032–032},
}

@misc{aganagic2017,
      title={Topological Chern-Simons/Matter Theories}, 
      author={M. Aganagic and K. Costello and J. McNamara and C. Vafa},
      year={2017},
      eprint={1706.09977},
      archivePrefix={arXiv},
      primaryClass={hep-th},
      url={https://arxiv.org/abs/1706.09977}, 
}

@misc{nekrasovtying,
      title={Tying up instantons with anti-instantons}, 
      author={N. Nekrasov},
      year={2020},
      eprint={1802.04202},
      archivePrefix={arXiv},
      primaryClass={hep-th},
      url={https://arxiv.org/abs/1802.04202}, 
}

@article{Yagi_2014,
   title={$\Omega$-deformation and quantization},
   volume={2014},
   ISSN={1029-8479},
   url={http://dx.doi.org/10.1007/JHEP08(2014)112},
   DOI={10.1007/jhep08(2014)112},
   number={8},
   journal={Journal of High Energy Physics},
   publisher={Springer Science and Business Media LLC},
   author={Yagi, J.},
   year={2014},
   month=aug }

@misc{costello2016,
      title={M-theory in the Omega-background and 5-dimensional non-commutative gauge theory}, 
      author={K. Costello},
      year={2016},
      eprint={1610.04144},
      archivePrefix={arXiv},
      primaryClass={hep-th},
      url={https://arxiv.org/abs/1610.04144}, 
}

@misc{nekrasov2002,
      title={Seiberg-Witten Prepotential From Instanton Counting}, 
      author={N. A. Nekrasov},
      year={2002},
      eprint={hep-th/0206161},
      archivePrefix={arXiv},
      primaryClass={hep-th},
      url={https://arxiv.org/abs/hep-th/0206161}, 
}

@misc{nekrasovokounkov,
      title={Membranes and Sheaves}, 
      author={N. Nekrasov and A. Okounkov},
      year={2014},
      eprint={1404.2323},
      archivePrefix={arXiv},
      primaryClass={math.AG},
      url={https://arxiv.org/abs/1404.2323}, 
}

@misc{costelloyangian,
      title={Supersymmetric gauge theory and the Yangian}, 
      author={K. Costello},
      year={2013},
      eprint={1303.2632},
      archivePrefix={arXiv},
      primaryClass={hep-th},
      url={https://arxiv.org/abs/1303.2632}, 
}

@article{costellowittenyamazaki,
   title={Gauge Theory And Integrability, I},
   volume={6},
   ISSN={2326-4845},
   url={http://dx.doi.org/10.4310/ICCM.2018.v6.n1.a6},
   DOI={10.4310/iccm.2018.v6.n1.a6},
   number={1},
   journal={Notices of the International Congress of Chinese Mathematicians},
   publisher={International Press of Boston},
   author={Costello, K. and Witten, E. and Yamazaki, M.},
   year={2018},
   pages={46–119}, }

@misc{zenkevich2023spirallingbranesrmatrices,
      title={Spiralling branes and R-matrices}, 
      author={Y. Zenkevich},
      year={2023},
      eprint={2312.16990},
      archivePrefix={arXiv},
      primaryClass={hep-th},
      url={https://arxiv.org/abs/2312.16990}, 
}

@misc{kronheimer, 
    title={Monopoles and Taub-NUT Metrics}, 
    author={P. Kronheimer}, 
    year={1985}, 
    url={https://people.math.harvard.edu/~kronheim/MSc-Thesis-Oxford-1985.pdf}
}

@book{miwa2000solitons,
  title={Solitons: Differential Equations, Symmetries and Infinite Dimensional Algebras},
  author={Miwa, T. and Jimbo, M. and Date, E.},
  isbn={9780521561617},
  series={Cambridge Tracts in Mathematics},
  url={https://books.google.com/books?id=kQDw1ZcqLjUC},
  year={2000},
  publisher={Cambridge University Press}
}

@inproceedings{Mulase2002ALGEBRAICTO,
  title={Algebraic Theory of the KP Equations},
    booktitle={Perspectives in Mathematical Physics},
    editors={R. Penner and S. T. Yau},
  author={M. Mulase},
  year={1994},
  url={https://api.semanticscholar.org/CorpusID:1736764}, 
}

@article{dubrovin, 
title={Theta functions and nonlinear equations}, 
author={B. Dubrovin},
journal={Russian Mathematical Surveys},
year={1981},
volume={36},
issue={2},
pages={11},
url={https://doi.org/10.1070/RM1981v036n02ABEH002596},
}

@misc{dijkgraaf92,
      title={Intersection Theory, Integrable Hierarchies and Topological Field Theory}, 
      author={R. Dijkgraaf},
      year={1992},
      eprint={hep-th/9201003},
      archivePrefix={arXiv},
      primaryClass={hep-th},
      url={https://arxiv.org/abs/hep-th/9201003}, 
}

@article{hitchinspectral,
author={N. J. Hitchin}, 
title={Monopoles and geodesics},
journal={Commun. Math. Phys.}, 
year={1982},
volume={83},
pages={579-692},
}

@article{donaldsonrational,
author={S. K. Donaldson},
title={Nahm's equations and the classification of monopoles},
journal={Commun. Math. Phys.}, 
year={1984},
volume={96},
pages={387-407},
}

@article{hurtubise, 
author={J. Hurtubise},
title={Monopoles and rational maps: A note on a theorem of Donaldson},
journal={Commun. Math. Phys.},
year={1985},
volume={100},
pages={191-196},
}

@article{jarvis, 
author={S. Jarvis}, 
title={Euclidean monopoles and rational maps}, 
journal={Proceedings of the London Mathematical Society},
year={1998},
volume={77},
pages={170-192},
}

@article{corrigangoddard,
    author={E. Corrigan and P. Goddard},
    title={Construction of instanton and monopole solutions and reciprocity},
    journal={Annals of Physics},
    volume={154}, 
    issue={1},
    year={1984},
    pages={253-279},
}

@misc{infinitegrassmannianfunctor,
      title={The algebraic formalism of soliton equations over arbitrary base fields}, 
      author={A. Álvarez Vázquez and J. M. Muñoz Porras and F. J. Plaza Martín},
      year={1996},
      eprint={alg-geom/9606009},
      archivePrefix={arXiv},
      primaryClass={alg-geom},
      url={https://arxiv.org/abs/alg-geom/9606009}, 
}

@article{Tsymbaliuk_2017,
   title={The affine Yangian of $\mathfrak{gl}_1$ revisited},
   volume={304},
   ISSN={0001-8708},
   url={http://dx.doi.org/10.1016/j.aim.2016.08.041},
   DOI={10.1016/j.aim.2016.08.041},
   journal={Advances in Mathematics},
   publisher={Elsevier BV},
   author={Tsymbaliuk, A.},
   year={2017},
   month=jan, pages={583–645} }

@article{Feigin_2011,
   title={Quantum continuous $\mathfrak{gl}_\infty$: Semiinfinite construction of representations},
   volume={51},
   ISSN={2156-2261},
   url={http://dx.doi.org/10.1215/21562261-1214375},
   DOI={10.1215/21562261-1214375},
   number={2},
   journal={Kyoto Journal of Mathematics},
   publisher={Duke University Press},
   author={Feigin, B. and Feigin, E. and Jimbo, M. and Miwa, T. and Mukhin, E.},
   year={2011},
   month=jan }

@misc{mirkovic2002,
      title={On quiver varieties and affine Grassmannians of type A}, 
      author={I. Mirković and M. Vybornov},
      year={2002},
      eprint={math/0206084},
      archivePrefix={arXiv},
      primaryClass={math.AG},
      url={https://arxiv.org/abs/math/0206084}, 
}

@misc{bgk02,
      title={Wilson's grassmannian and a noncommutative Quadric}, 
      author={V. Baranovsky and V. Ginzburg and A. Kuznetsov},
      year={2002},
      eprint={math/0203116},
      archivePrefix={arXiv},
      primaryClass={math.AG},
      url={https://arxiv.org/abs/math/0203116}, 
}

@misc{bgk01,
      title={Quiver varieties and a non-commutative $\mathbb{P}^2$}, 
      author={V. Baranovsky and V. Ginzburg and A. Kuznetsov},
      year={2001},
      eprint={math/0103068},
      archivePrefix={arXiv},
      primaryClass={math.AG},
      url={https://arxiv.org/abs/math/0103068}, 
}

@misc{benzvi03,
      title={From solitons to many-body systems}, 
      author={D. Ben-Zvi and T. Nevins},
      year={2003},
      eprint={math/0310490},
      archivePrefix={arXiv},
      primaryClass={math.AG},
      url={https://arxiv.org/abs/math/0310490}, 
}

@article{Braverman_2010,
   title={Pursuing the double affine Grassmannian, I: Transversal slices via instantons on Ak-singularities},
   volume={152},
   ISSN={0012-7094},
   url={http://dx.doi.org/10.1215/00127094-2010-011},
   DOI={10.1215/00127094-2010-011},
   number={2},
   journal={Duke Mathematical Journal},
   publisher={Duke University Press},
   author={Braverman, A. and Finkelberg, M.},
   year={2010},
   month=apr }

@misc{witten2009geometriclanglandsdimensions,
      title={Geometric Langlands From Six Dimensions}, 
      author={E. Witten},
      year={2009},
      eprint={0905.2720},
      archivePrefix={arXiv},
      primaryClass={hep-th},
      url={https://arxiv.org/abs/0905.2720}, 
}

@misc{bfg,
      title={Uhlenbeck spaces via affine Lie algebras}, 
      author={A. Braverman and M. Finkelberg and D. Gaitsgory},
      year={2012},
      eprint={math/0301176},
      archivePrefix={arXiv},
      primaryClass={math.AG},
      url={https://arxiv.org/abs/math/0301176}, 
}

@article{Kapustin_2001,
   title={Noncommutative Instantons and Twistor Transform},
   volume={221},
   ISSN={0010-3616},
   url={http://dx.doi.org/10.1007/PL00005576},
   DOI={10.1007/pl00005576},
   number={2},
   journal={Communications in Mathematical Physics},
   publisher={Springer Science and Business Media LLC},
   author={Kapustin, A. and Kuznetsov, A. and Orlov, D.},
   year={2001},
   month=jul, pages={385–432} }

@book{qftias,
    editor = {Deligne, P. and Etingof, P. and Freed, D. S. and Jeffrey, L. C. and Kazhdan, D. and Morgan, J. W. and Morrison, D. R. and Witten, E.},
    title = {Quantum fields and strings: A course for mathematicians. Vol. 2},
    isbn = {978-0-8218-2012-4},
    year = {1999},
    publisher = {American Mathematical Society}
}

@misc{tonggaugethy,
    author={D. Tong}, 
    title={Lectures on Gauge Theory}, 
    date={2018},
    url={https://www.damtp.cam.ac.uk/user/tong/gaugetheory/gt.pdf}
}

@inproceedings{Nekrasov_2001,
   title={Trieste Lectures on Solitons in Noncommutative Gauge Theories},
   url={http://dx.doi.org/10.1142/9789812810274_0004},
   DOI={10.1142/9789812810274_0004},
   booktitle={Superstrings and Related Matters},
   publisher={WORLD SCIENTIFIC},
   author={Nekrasov, N. A.},
   year={2001},
   month=apr, pages={141–205} }

@misc{AO,
      title={Quasimap counts and Bethe eigenfunctions}, 
      author={M. Aganagic and A. Okounkov},
      year={2017},
      eprint={1704.08746},
      archivePrefix={arXiv},
      primaryClass={math-ph},
      url={https://arxiv.org/abs/1704.08746}, 
}

\end{document}